\documentclass[11pt,a4paper,reqno]{amsart}

\usepackage[foot]{amsaddr}

\usepackage[textwidth = 2cm]{todonotes}
\usepackage{amsmath}
\usepackage{amsfonts}
\usepackage{amssymb}
\usepackage{amsthm}
\usepackage{verbatim}
\usepackage{dsfont}
\usepackage{fullpage}
\usepackage{microtype}
\usepackage[libertine]{newtxmath}
\usepackage{newtxtext}
\allowdisplaybreaks

\usepackage[tt=false]{libertine} 
\usepackage{caption}
\usepackage{bm}
\usepackage{bbm}
\usepackage{hyperref, color}
\hypersetup{colorlinks=true,citecolor=blue, linkcolor=blue, urlcolor=blue}
\usepackage[linesnumbered,boxed,ruled,vlined]{algorithm2e}
\usepackage{algorithmicx}
\usepackage[numbers]{natbib}
\usepackage{xcolor}
\usepackage{enumerate} 
\usepackage[shortlabels]{enumitem}
\usepackage{tabularx}
\usepackage{array}
\usepackage{cleveref}
\usepackage{pifont}
\usepackage{xifthen}
\usepackage{xstring}

\usepackage{stackengine}
\usepackage{scalerel}
\usepackage{graphicx}
\usepackage{float}
\usepackage{tikz}
\usepackage{adjustbox}
\usepackage{fancyhdr, comment, environ}

\usetikzlibrary{positioning}
\usetikzlibrary{tikzmark, fit, backgrounds}
\usetikzlibrary{matrix,fit,backgrounds}  

\DeclareMathAlphabet{\mathcal}{OMS}{cmsy}{m}{n}

\newcolumntype{L}[1]{>{\raggedright\arraybackslash}p{#1}}
\newcolumntype{C}[1]{>{\centering\arraybackslash}m{#1}}
\newcolumntype{R}[1]{>{\raggedleft\arraybackslash}p{#1}}

\usepackage{makecell}
\usepackage{footnote}
\makesavenoteenv{tabular}

\newtheorem{theorem}{Theorem}[section]

\newtheorem*{claim*}{Claim}
\newtheorem{condition}[theorem]{Condition}

\newtheorem{fact}[theorem]{Fact}
\newtheorem{lemma}[theorem]{Lemma}

\newtheorem{corollary}[theorem]{Corollary}
\theoremstyle{definition}

\newtheorem{definition}[theorem]{Definition}

\newtheorem*{remark*}{Remark}

\def\^#1{\mathbb{#1}} 
\def\*#1{\mathbf{#1}} 
\def\+#1{\mathcal{#1}} 
\def\-#1{\mathrm{#1}} 
\def\=#1{\boldsymbol{#1}} 
\def\!#1{\mathtt{#1}}
\def\@#1{\mathscr{#1}}

\newcommand{\abs}[1]{\left | #1 \right |}

\newcommand{\id}[1]{\ensuremath{\mathbbm{1}}\left[#1\right]}

\renewcommand{\Pr}[2][]{ \ifthenelse{\isempty{#1}}
  {\mathop{\mathbf{Pr}}\left[#2\right]} {\mathop{\mathbf{Pr}}_{#1} \left[#2\right]} }
\newcommand{\E}[2][]{ \ifthenelse{\isempty{#1}}
  {\mathbf{E}\left[#2\right]}
  {\mathbf{E}_{#1}\left[#2\right]} }

\newcommand{\diag}{\mathsf{diag}}
\newcommand{\perman}{\mathsf{per}}
\newcommand{\norm}[1]{\left\Vert#1\right\Vert}
\newcommand{\poly}{\textnormal{\textsf{poly}}}

\title{Phase transition of the Sinkhorn-Knopp algorithm}

\author{Kun He}
\address[Kun He]{Renmin University of China. \textnormal{E-mail: \url{hekun.threebody@foxmail.com}}}

\begin{document}

\begin{abstract}
The matrix scaling problem, particularly the Sinkhorn–Knopp algorithm, has been studied for over 60 years. In practice, the algorithm often yields high-quality approximations within just a few iterations. Theoretically, however, the best known upper bound on its iteration count scales polynomially with the accuracy parameter $\varepsilon$, placing it in the class of pseudopolynomial-time approximation algorithms. Meanwhile, the lower-bound landscape remains largely unexplored. Two fundamental questions persist: what accounts for the algorithm’s strong empirical performance, and can a tight bound on its iteration count be established?

For an $ n \times n $ matrix, its normalized version is obtained by dividing each entry by the largest entry in the matrix. We say that a normalized matrix has a density $\gamma$ if there exists a constant $\rho > 0$ such that one row or column has exactly $\lceil \gamma n \rceil$ entries with values at least $\rho$, and every other row and column has at least $\lceil \gamma n \rceil$ such entries.

For the upper bound, we show that the Sinkhorn–Knopp algorithm produces a nearly doubly stochastic matrix in $O(\log n - \log \varepsilon)$ iterations and $\tilde{O}(n^2)$ time for all nonnegative square matrices whose normalized version has a density $\gamma > 1/2$. 
Such matrices cover both the algorithm’s principal practical inputs and its typical theoretical regime.
This $\tilde{O}(n^2)$ runtime is optimal, as merely reading the input requires $\Omega(n^2)$ time. 
In fact, the algorithm is optimal for nearly all nonnegative matrices with entries bounded above by a constant, demonstrating its practical efficiency.

For the lower bound, we establish a tight bound of $\tilde{\Omega}\left(\sqrt{n}/\varepsilon\right)$ iterations for positive matrices under the $\ell_2$-norm error measure. Moreover, for every $\gamma < 1/2$, there exists a matrix with density $\gamma$ for which the Sinkhorn–Knopp algorithm requires $\Omega\left(\sqrt{n}/\varepsilon\right)$ iterations.

In summary, our results reveal a sharp phase transition in the Sinkhorn–Knopp algorithm at the density threshold $\gamma = 1/2$. 
Moreover, we demonstrate that convergence improves as matrix density increases, thereby accounting for prior observations in entropically regularised optimal transport that larger inverse temperature slows convergence.
\end{abstract}

\maketitle

\setcounter{tocdepth}{1}

\tableofcontents
\pagebreak

\section{Introduction}\label{sec-intro}
\medskip
\noindent The matrix scaling problem entails finding diagonal matrices \(X\) and \(Y\) such that, when a nonnegative matrix \(A\) is transformed into \(XAY\), its row and column sums exactly match specified target vectors. This problem is central to numerous fields in both theory and practice. For instance, in solving linear systems, matrix scaling is employed as a preconditioning technique to enhance numerical stability~\cite{osborne1960pre}. In optimal transport, scaling adjusts probability distributions to meet prescribed marginal constraints, thereby facilitating the computation of transport distances~\cite{altschuler2017near}. Moreover, matrix scaling is integral to statistical data normalization~\cite{deming1940least}, image processing~\cite{rubner2000earth}, and various other applications~\cite{idel2016review}.

One of the most natural and classical approaches to matrix scaling is the Sinkhorn-Knopp algorithm~\cite{sinkhorn1967diagonal,sinkhorn1967concerning} (also known as the RAS method~\cite{bacharach1965estimating} or the Iterative Proportional Fitting Procedure~\cite{ruschendorf1995convergence}). In this iterative process, one alternates between normalizing the rows and the columns of \(A\), ensuring that the scaling progressively aligns the matrix with the prescribed targets. The simplicity of the algorithm, combined with its inherent parallelizability, makes it a popular choice in practice. 
One key issue is determining how fast the Sinkhorn-Knopp algorithm converges. In other words, when we choose an appropriate error measure to quantify the deviation from the target values and set an error threshold \(\varepsilon\), how quickly does the algorithm reduce the error below \(\varepsilon\)? Despite several significant convergence analyses, it remains somewhat surprising that the behavior of the algorithm under common error metrics, such as \(\ell_1\) or \(\ell_2\) errors, is not yet fully characterized.

A nonnegative matrix \(A \in \mathbb{R}_{\ge 0}^{n \times m}\) is \((\boldsymbol{r}, \boldsymbol{c})\)-scalable, if there exist diagonal matrices \(X\) and \(Y\) with strictly positive diagonal entries such that \(XAY\) has row sums equal to the vector \(\boldsymbol{r}\) and column sums equal to the vector \(\boldsymbol{c}\).
In the general setting, if a nonnegative matrix \(A \in \mathbb{R}_{\ge 0}^{n \times m}\) is \((\boldsymbol{r}, \boldsymbol{c})\)-scalable and \(\varepsilon > 0\) is given, the Sinkhorn-Knopp algorithm can produce, in time  
$t = O\Bigl(h^2 \varepsilon^{-2}\log\bigl(\Delta\,\mu/\nu\bigr)\Bigr)$,
a scaled matrix whose \(\ell_1\)-error is at most \(\varepsilon\), or in time  
$t = O\Bigl(\mu\,h\,\log\bigl(\Delta\,\mu/\nu\bigr)\Bigl(\varepsilon^{-1}+\varepsilon^{-2}\Bigr)\Bigr)$,
a scaled matrix whose \(\ell_2\)-error is at most \(\varepsilon\)~\cite{chakrabarty2021better}. Here, \(h\) is the sum of the target row (or column) values, \(\mu\) is the maximum target entry, \(\Delta\) is the maximum number of nonzero entries in any column of \(A\), and \(\nu\) is defined as the ratio of the smallest positive entry in \(A\) to its largest entry.
In the special case of \((\boldsymbol{1},\boldsymbol{1})\)-scaling, i.e., when \(A\) is an \(n \times n\) matrix that is to be made doubly stochastic, the above bounds simplify to  
$
t = O\Bigl(n^2\varepsilon^{-2}\,\log\bigl(\Delta/\nu\bigr)\Bigr)$
for achieving an \(\ell_1\)-error at most \(\varepsilon\), and  
$t = O\Bigl(n\,\log\bigl(\Delta/\nu\bigr)\Bigl(\varepsilon^{-1}+\varepsilon^{-2}\Bigr)\Bigr)$
for achieving an \(\ell_2\)-error at most \(\varepsilon\)~\cite{chakrabarty2021better}.
Finally, for strictly positive matrices (where every entry of \(A\) is nonzero), 
the Sinkhorn-Knopp algorithm can produce, in time  
$t = O\Bigl(\sqrt{n}\varepsilon^{-1}\,\log \nu\Bigr)$,
a scaled matrix whose \(\ell_2\)-error is at most \(\varepsilon\) for both the special case of \((\boldsymbol{1},\boldsymbol{1})\)-scaling~\cite{kalantari1993rate} and the general case~\cite{kalantari2008complexity}.
On the lower-bound side, a carefully constructed $2 \times 2$ matrix shows that the Sinkhorn–Knopp algorithm needs $\Omega(1/\varepsilon)$ iterations to reduce the $\ell_2$-error below $\varepsilon$ even for $(\mathbf{1},\mathbf{1})$-scaling; consequently, it is \emph{not} a polynomial-time approximation scheme~\cite{kalantari1993rate}.

The presented findings emphasize that the Sinkhorn-Knopp algorithm, despite its popularity and simplicity, depends polynomially on the error parameter \(\varepsilon\), categorizing it as a “pseudopolynomial” approximation scheme. In contrast, several advanced methods exhibit a \(\log(1/\varepsilon)\) dependence. For example, Kalantari and Khachiyan’s ellipsoid-based algorithm can run in \(O(n^4 \log(n/\varepsilon)\log(\nu))\) time~\cite{kalantari1996complexity}, 
while Nemirovskii and Rothblum’s approach has a complexity \(O(n^4 \log(n/\varepsilon)\log\log(\nu))\)~\cite{nemirovski1999complexity}. 
Linial, Samorodnitsky, and Wigderson proposed the first strongly polynomial-time solution with a running time of \(O(n^7 \log(h/\varepsilon))\),  with no dependence on $\nu$~\cite{linial1998deterministic}. 
An alternative line of work by Rote and Zachariasen reduces matrix scaling to a min-cost flow problem, yielding a complexity of \(O(n^4 \log(h/\varepsilon))\)~\cite{rote2007matrix}. 
Allen-Zhu et al.~\cite{allen2017much} designed several algorithms for the general scaling problem and its important special cases, showing that
if $A$ has $\ell$ nonzero entries and there exist $X$ and $Y$ with polynomially bounded entries such that
$XAY$ is doubly stochastic, then their algorithm can solve the problem in total complexity $O\left((\ell+n^{4/3})\poly(\log n\log (1/\varepsilon)\right)$; meanwhile, 
Cohen et al.~\cite{cohen2017matrix} provided algorithms running
in time $\tilde{O}(\ell\log \kappa \log^2 (1/\varepsilon))$ where $\kappa$ denotes the ratio between the largest and the smallest entries of the optimal scalings, with logarithmic factors in the matrix dimensions and entry sizes suppressed.

As noted above, the matrix-scaling problem, particularly the Sinkhorn–Knopp algorithm, has been extensively studied for more than 60 years. Empirically, the algorithm performs remarkably well, typically reaching a high-quality approximation in just a few iterations \cite{dufosse2022scaling}. The strongest theoretical guarantee, however, bounds its running time by a polynomial in $1/\varepsilon$, placing the method only in pseudopolynomial time.
Despite significant progress on upper-bound analyses, the lower-bound landscape remains largely unexplored: no non-trivial lower bound on the number of iterations is known that depends simultaneously on the matrix dimension $n$ and the accuracy parameter $\varepsilon$.
Consequently, the following two questions remain unresolved:
\begin{itemize}
\item Given that the Sinkhorn-Knopp algorithm is not a polynomial-time approximation scheme, what underlying factors account for its robust practical performance?
\item Can we establish a tight bound on the number of iterations of the algorithm?
\end{itemize}

\subsection{Main Results}
In this paper, we provide definitive answers to these open questions. In particular, we identify a sufficient condition under which the Sinkhorn-Knopp algorithm converges in \(O\left(\log n - \log \varepsilon\right)\) iterations for \((\boldsymbol{1},\boldsymbol{1})\)-scaling.
This condition captures both the primary practical applications of the algorithm and its typical theoretical behavior.
Consequently, for almost all nonnegative matrices whose entries are bounded above by a fixed constant, the Sinkhorn–Knopp algorithm computes the \((\boldsymbol{1},\boldsymbol{1})\)-scaling in optimal running time.
Furthermore, we establish a tight lower bound of \(\Omega\left(\sqrt{n}/\varepsilon\right)\) on the number of iterations required by the algorithm when applied to positive matrices, using the widely adopted \(\ell_2\)-norm for error measurement. 

Let \(\gamma\in [0,1]\), \(\rho \in (0,1]\). An \(n\times n\) matrix \(A\) with entries in \([0,1]\) is defined to be \((\gamma,\rho)\)-dense if there exists at least one row or column that contains exactly \(\lceil \gamma n \rceil\) entries with values at least \(\rho\), and every other row and column contains at least \(\lceil \gamma n \rceil\) entries with values at least \(\rho\). We say that a matrix \(A\) is \(\gamma\)-dense\footnote{We note that a matrix can be simultaneously \(\gamma_1\)-dense and \(\gamma_2\)-dense for some \(\gamma_1 \neq \gamma_2\).} (or has a density \(\gamma\)) if there exists a constant \(\rho > 0\) such that \(A\) is \((\gamma,\rho)\)-dense.
Finally, a matrix is called dense if it is \(\gamma\)-dense for some \(\gamma > 1/2\). Intuitively, a matrix with entries in \([0,1]\) is considered dense if most of the entries in each row and column exceed a certain positive threshold.

For each $i,j\in [n]$, let $A_{i,j}$ denote the element in row $i$ and column $j$ of $A$, $r_i(A)$ denote $\sum_{k\in [n]} A_{i,k}$,
and $c_j(A)$ denote $\sum_{k\in [n]} A_{k,j}$.
Let $r(A)$ denote the vector $(r_1(A),\cdots,r_n(A))$
and $c(A)$ denote the vector $(c_1(A),\cdots,c_n(A))$.

\vspace{0.3cm}
\noindent \underline{\textbf{Upper bounds.}}
Below is an upper bound on the number of iterations of the Sinkhorn–Knopp algorithm.

\begin{theorem}\label{thm-upper-bound-general}
Let $\gamma \in (1/2,1]$ and $n,\varepsilon>0$. 
Let $B$ be a nonnegative, nonzero $n\times n$ matrix
where $B/(\max_{i,j\in[n]} B_{i,j})$ is $\gamma$-dense.
With $B$ as input, the Sinkhorn-Knopp algorithm can output a nearly doubly stochastic matrix $A$ satisfying 
\[\norm{r\left(A\right)- \boldsymbol{1}}_{1} + \norm{c\left(A\right)- \boldsymbol{1}}_{1} \leq \varepsilon\]
in $O\left(\left(2\gamma - 1\right)^{-5}\left(\log n - \log \varepsilon\right)\right)$ iterations.
\end{theorem}

One can verify that our theorem holds for other norms as well, since the number of iterations is the logarithm of the error.

Note that each iteration of the Sinkhorn-Knopp algorithm requires \(O(n^2)\) time. Consequently, \Cref{thm-upper-bound-general} demonstrates that if the scaled matrix \(B/\max_{i,j\in[n]} B_{i,j}\) is dense, then the algorithm runs in \(\tilde{O}(n^2)\) time, where the \(\tilde{O}\) notation suppresses logarithmic factors and constants that depend on \(\gamma\). This running time is optimal, as merely reading the input matrix already takes \(\Omega(n^2)\) time.

Our result establishes the first class of matrices for which the Sinkhorn-Knopp algorithm converges in \(O(\log n - \log \varepsilon)\) iterations for \((\boldsymbol{1},\boldsymbol{1})\)-scaling, which improves upon the previous upper bounds of \(\tilde{O}(\sqrt{n}/\varepsilon)\) for positive matrices and \(\tilde{O}(n^2/\varepsilon^2)\) for nonnegative matrices. This finding indicates that the Sinkhorn-Knopp algorithm is significantly more efficient on many important matrices than general upper bounds suggest.

We concentrate on the dense-matrix regime in \Cref{thm-upper-bound-general} because it captures both the algorithm’s primary practical applications and its typical theoretical behavior.
First, in many canonical applications of the Sinkhorn–Knopp algorithm, the input matrix is inherently dense. A prime example is entropically regularised optimal transport~\cite{cuturi2013sinkhorn}, where the cost matrix $C$ records distances for every source–target pair and is typically normalised. 
Normalising $C$ fixes its scale so that the regularisation parameter $\eta$ has a consistent meaning across datasets and prevents numerical pathologies in the Gibbs kernel (e.g., underflow from overly large costs or excessive flattening from overly small costs).
Such normalisation places most entries of $C$ in $[0,1]$.\footnote{Common choices include scaling $C$ so that its entries lie in $[0,1]$; dividing $C$ by its median; or dividing by a higher quantile such as the 75th or 95th percentile. Each yields a normalised $C$ whose bulk lies in $[0,1]$.} Consequently, for any fixed regularisation parameter $\eta>0$, most entries of the Gibbs kernel $K=\exp(-\eta C)$ lie in $[\exp(-\eta),1]$, and hence $K$ is effectively dense.
Second, density is not only prevalent in practice but generic in theory: with respect to the Lebesgue measure, almost every $n \times n$ non-negative matrix whose entries are bounded above by a fixed constant is dense. Consequently, \Cref{thm-upper-bound-general} shows that the Sinkhorn–Knopp algorithm achieves optimal complexity for almost all such matrices.

\begin{corollary}\label{cor-upper-bound-most}
Let \(n, \varepsilon > 0\). Consider the set of all \(n \times n\) nonnegative matrices with entries bounded above by a fixed constant. For a \(1 - \exp(-\Omega(n))\) fraction of such matrices (with respect to the Lebesgue measure), the Sinkhorn-Knopp algorithm, given the matrix as input,  can output a nearly doubly stochastic matrix $A$ satisfying 
\[\norm{r\left(A\right)- \boldsymbol{1}}_{1} + \norm{c\left(A\right)- \boldsymbol{1}}_{1} \leq \varepsilon\]
in $O\left(\log n - \log \varepsilon\right)$ iterations.
\end{corollary}

\Cref{cor-upper-bound-most} demonstrates the practical efficiency of the Sinkhorn-Knopp algorithm.
Intuitively, for almost all nonnegative matrices $B$, the normalized matrix $B/(\max_{i,j\in[n]} B_{i,j})$ is dense. 
Consequently, the Sinkhorn-Knopp algorithm achieves an optimal runtime of \(\tilde{O}(n^2)\) on these matrices.

The proofs of \Cref{thm-upper-bound-general} and \Cref{cor-upper-bound-most} are provided in \Cref{sec-upper-bounds}.

\vspace{0.3cm}
\noindent \underline{\textbf{Lower bounds.}}
The next theorem provides a tight lower bound on the iteration complexity of the Sinkhorn–Knopp algorithm for positive matrices, with error measured in the widely used $\ell_{2}$-norm.

\begin{theorem}\label{thm-lower-bound-positive}
Let $\varepsilon < 10^{-3}$. 
There exists a \((\boldsymbol{1}, \boldsymbol{1})\)-scalable, positive matrix of size $n\times n$ with entries in $[\varepsilon^{8}/(100 n^{61}),1]$ such that with the matrix as input, 
the Sinkhorn-Knopp algorithm takes $\Omega(n/\varepsilon)$ iterations to output a matrix $A$ satisfying 
\[\norm{r\left(A\right)- \boldsymbol{1}}_{1} + \norm{c\left(A\right)- \boldsymbol{1}}_{1} \leq \varepsilon.\]
If $\varepsilon < 1/(1000\sqrt{n})$, the Sinkhorn-Knopp algorithm takes $\Omega(\sqrt{n}/\varepsilon)$ iterations to output a matrix $A$ satisfying 
\[\norm{r\left(A\right)- \boldsymbol{1}}_{2} + \norm{c\left(A\right)- \boldsymbol{1}}_{2} \leq \varepsilon.\]
\end{theorem}

It is well known that certain matrices cause the Sinkhorn–Knopp algorithm to \emph{diverge}, so the required number of iterations is infinite~\cite{sinkhorn1964relationship}. In \Cref{thm-lower-bound-positive} and later in \Cref{thm-lower-bound-main}, we focus exclusively on matrices for which the algorithm \emph{does converge}. For this convergent class, the best previously established lower bound on the iteration count is $\Omega(1/\varepsilon)$, achieved with a $2\times 2$ matrix~\cite{kalantari1993rate}. Beyond that, no non-trivial bound had been established that depends simultaneously on the matrix dimension $n$ and the accuracy parameter $\varepsilon$.
We close this gap by proving sharper bounds: $\Omega(n/\varepsilon)$ iterations are necessary to reach an $\ell_{1}$-error of at most $\varepsilon$, and $\Omega(\sqrt{n}/\varepsilon)$ iterations are required for the same $\ell_{2}$-error threshold.

Assume that the widely adopted \(\ell_2\)-norm is used for error measurement and that the input positive matrix \(B\) satisfies
\[
\frac{\max_{i,j\in [n]} B_{i,j}}{\min_{i,j\in [n]} B_{i,j}} = \operatorname{poly}(n,\varepsilon).
\]
It has been shown that the Sinkhorn-Knopp algorithm requires 
$
O\Bigl(\sqrt{n}(\log n-\log \varepsilon)/\varepsilon\Bigr)$
iterations on such matrices~\cite{kalantari1993rate} for both general \((\boldsymbol{r}, \boldsymbol{c})\)-scaling and specific \((\boldsymbol{1},\boldsymbol{1})\)-scaling. 
\Cref{thm-lower-bound-positive} demonstrates that there exists a matrix \(B\) for which the Sinkhorn-Knopp algorithm needs \(\Omega\Bigl(\sqrt{n}/\varepsilon\Bigr)\) iterations. This lower bound is tight up to logarithmic factors.

By \Cref{thm-lower-bound-positive}, one can establish the following lower bound for \(\gamma\)-dense matrices, where \(\gamma \in [0,1/2)\).
\begin{theorem}\label{thm-lower-bound-main}
Let $\gamma \in [0,1/2)$ and $\varepsilon < 10^{-3}$. 
There exists a $\gamma$-dense, \((\boldsymbol{1}, \boldsymbol{1})\)-scalable matrix of size $n\times n$ such that, with the matrix as input, 
the Sinkhorn-Knopp algorithm takes $\Omega(n/\varepsilon)$ iterations to output a matrix $A$ satisfying 
\[\norm{r\left(A\right)- \boldsymbol{1}}_{1} + \norm{c\left(A\right)- \boldsymbol{1}}_{1} \leq \varepsilon.\]
If $\varepsilon < 1/(2000\sqrt{n})$, the Sinkhorn-Knopp algorithm takes $\Omega(\sqrt{n}/\varepsilon)$ iterations to output a matrix $A$ satisfying 
\[\norm{r\left(A\right)- \boldsymbol{1}}_{2} + \norm{c\left(A\right)- \boldsymbol{1}}_{2} \leq \varepsilon.\]
\end{theorem}

Theorems \ref{thm-upper-bound-general} and \ref{thm-lower-bound-main} together reveal a sharp phase transition in the Sinkhorn-Knopp algorithm at a density threshold of \(\gamma = 1/2\). Under the \(\ell_1\)-norm, the algorithm requires \(O(\log n - \log \varepsilon)\) iterations when the density exceeds \(1/2\), but this number sharply increases to \(\Omega(n/\varepsilon)\) when the density falls below \(1/2\).

Our phase-transition analysis shows that the Sinkhorn–Knopp algorithm converges faster as the matrix becomes denser. This trend has been observed empirically in prior work. In entropically regularised optimal transport~\cite{cuturi2013sinkhorn}, the Gibbs kernel $K_\eta=\exp(-\eta C)$ interpolates between a sharply peaked matrix and a dense one: when $\eta$ is large, most entries of $K_\eta$ are numerically negligible and convergence is slow; as $\eta$ decreases, $K_\eta$ becomes dense and convergence accelerates markedly. Our analysis pinpoints this transition and provides a clear theoretical explanation for the observed speedups.


The proofs of Theorems \ref{thm-lower-bound-positive} and \ref{thm-lower-bound-main} are provided in \Cref{sec-lower-bounds}.

\vspace{0.5cm}
\noindent \underline{\textbf{An application.}}
As an application of our results, we provide a fast approximation algorithm for the permanent of dense 0-1 matrices.
Given a 0-1 square matrix $A = (a_{i,j})_{n\times n}$, the permanent of $A$ is defined as 
\[\perman(A) \triangleq \sum_{\sigma }\prod_{i\in [n]}a_{i,\sigma(i)},\]
where the sum is over all permutations $\sigma$ of $[n]$.
Computing the permanent of a matrix is one of the first problems shown to be $\#$P-complete~\cite{valiant1979complexity}, even if the matrix is a 0-1 matrix where the row and column sums are at least $n/2$.

\begin{theorem}\label{thm-application}
Let $\gamma \in (1/2,1]$, $\delta \in (0,1]$ and $\varepsilon \in (0,1]$. 
There exists a randomized approximation algorithm such that given any $\gamma$-dense $0$-$1$ matrix $A$ of size $n\times n$ as input,
the algorithm outputs an approximation of $\perman(A)$
within a factor of $1+\varepsilon$ with probability at least $1 - \delta$ and expected running time $O\left(n^2 \left(2\gamma - 1\right)^{-5}\left(\log n - \log \varepsilon\right)+ n^{2+(1-\gamma)/(2\gamma - 1)}\varepsilon^{-2}\log (1/\delta)\right)$.
\end{theorem}

For simplicity, let 
$R = n^{2+(1-\gamma)/(2\gamma-1)}\varepsilon^{-2}\log (1/\delta)$.
Previously, the best approximation algorithm for dense 0–1 matrices~\cite{huber2008fast} ran in 
$O\bigl(n^4(\log n-\log \varepsilon) + R\bigr)$.
In that approach, given a dense 0–1 matrix \(A\), the algorithm first scales \(A\) into a nearly doubly stochastic matrix \(B\), and then estimates \(\perman(A)\) using a novel sequential acceptance/rejection method based on the entries of \(B\). When \(A\) is scaled into \(B\) using the Sinkhorn-Knopp algorithm, the resulting approximation algorithm runs in 
$
O\bigl(n^{4.5}(\log n-\log \varepsilon) + R\bigr)$,
whereas employing the ellipsoid algorithm yields a running time of 
$O\bigl(n^4(\log n-\log \varepsilon) + R\bigr)$~\cite{huber2008fast}.
Alternatively, using the advanced scaling algorithm from~\cite{allen2017much} leads to a running time of 
$
O\bigl(n^2\,\poly(\log n,\log(1/\varepsilon)) + R\bigr)$.
However, both the ellipsoid algorithm and the method from~\cite{allen2017much} are complex. It remains a critical question whether a simple approximation algorithm based on the Sinkhorn-Knopp algorithm can achieve a fast running time.
\Cref{thm-application} provides an improved running time of 
$O\bigl(n^2(\log n-\log \varepsilon) + R\bigr)$
for the simple algorithm, which outperforms all previous approaches.

In particular, compared to the earlier bound 
$O\bigl(n^{4.5}(\log n-\log \varepsilon) + R\bigr) = \tilde{O}\bigl(n^{4.5} + n^{2+(1-\gamma)/(2\gamma-1)}\varepsilon^{-2}\bigr)$
for the approximation algorithm based on the Sinkhorn-Knopp algorithm, our result eliminates the \(n^{4.5}\) term, achieving a running time of 
$O\bigl(n^2(\log n-\log \varepsilon) + R\bigr) = \tilde{O}\bigl(n^{2+(1-\gamma)/(2\gamma-1)}\varepsilon^{-2}\bigr)$.
This new bound improves the previous running time when \(\gamma \geq 7/12\).

The proof of \Cref{thm-application} is provided in \Cref{appendix-application}.

\subsection{Technique Overview}\label{sec-technique-overview}
In this section, we outline our techniques for proving the upper and lower bounds separately. Our key insight is that when \(\gamma > 1/2\), the dense structure of the original matrix is preserved during the scaling process, enabling us to establish an upper bound. Conversely, when \(\gamma < 1/2\), the density profile may gradually evolve during scaling, which allows us to construct a matrix that yields a lower bound. Furthermore, the connection between the permanent and matrix scaling plays a key role in our proof; it not only facilitates the derivation of the upper bound but also inspires our construction of the counterexample for the lower bound.

In proving our upper bound, we in fact establish a constant upper bound on the condition number $\kappa$ of the optimal scaling matrices for any dense input matrix\footnote{This constant bound follows directly from Lemmas \ref{lem-upper-bound-max-element} and \ref{lem-lower-bound-elements}, although we did not spell it out explicitly in the proof.}.  Here, $\kappa$ is defined as the ratio between the largest to the smallest diagonal entry of the optimal scaling matrices.  Controlling $\kappa$ is pivotal: many accelerated scaling algorithms converge fast only when $\kappa$ is bounded~\cite{cohen2017matrix}. 
Until now, however, meaningful bounds on $\kappa$ were known only for strictly positive matrices or the pseudorandom instances of~\cite{kwok2018paulsen}. 
Our result fills this gap by showing that dense matrices likewise admit a small—indeed constant—condition number, and the structural ideas behind the proof suggest a combinatorial avenue toward similar bounds for broader classes of matrices.

\vspace{0.3cm}
\noindent \underline{\textbf{A combinatorial analysis exploiting the dense structure of matrices.}}
In contrast to prevalent methods rooted in continuous optimization, we establish the upper bound for the Sinkhorn-Knopp algorithm using a combinatorial approach that offers a novel perspective.
The strength of this combinatorial analysis lies in its ability to bridge structural properties of dense matrices with the iterative scaling process. 

Let \(\gamma \in \left(1/2, 1\right]\) and let \(A\) be an \(n \times n\) \((\gamma,\rho)\)-dense matrix. An entry of \(A\) is called \emph{considerable} if it exceeds \(\rho\). 
Assuming the \(\ell_1\)-norm is used, we define $t \triangleq 9n(2\gamma-1)/(20\gamma) $. 
Our analysis of the Sinkhorn-Knopp algorithm on \(A\) proceeds in two phases:

\begin{itemize}
\item \emph{Phase 1}: The Sinkhorn-Knopp algorithm reduces the error to $t$ within \(O(\log n)\) iterations.
\item \emph{Phase 2}: Once the scaled matrix has an error below $t$, the error decays almost exponentially, and the algorithm reaches an error of \(\varepsilon\) within an additional \(O(\log n - \log \varepsilon)\) iterations.
\end{itemize}

To upper bound the number of iterations in Phase 1, we build upon the framework of~\cite{linial1998deterministic}, with modifications tailored to our setting. If the error is at least $t = \Theta(n)$, then both the product of the row sums and the product of the column sums are upper bounded by \(\exp(-\Theta(n))\). Consequently, the permanent of the scaled matrix increases by a factor of \(\exp(\Theta(n))\) in each round. Since the permanent of the input dense matrix is at least \(n^{-n}\) and that of the scaled matrix is at most 1, the number of iterations in Phase 1 is bounded by \(O(\log n)\).

Let \(A^{(k)}\) denote the scaled matrix produced by the Sinkhorn-Knopp algorithm at its \(k\)-th iteration during Phase 2. Then \(A^{(k)}\) has an error of at most $t$. Without loss of generality, assume that every row of \(A^{(k)}\) sums to 1 and that \(A^{(k)} = XAY\), where \(X\) and \(Y\) are diagonal matrices.
To bound the number of iterations in Phase 2, it is crucial to show that every considerable entry in \(A\) is of order \(\Theta(1/n)\) in \(A^{(k)}\). Consequently, the normalized matrix \(A^{(k)}/\max_{i,j} A^{(k)}_{i,j}\) remains dense, thereby preserving the inherent dense structure of \(A\) throughout the scaling process.

We first show that the considerable entries in \(A\) become \(O(1/n)\) in \(A^{(k)}\) (see \Cref{lem-upper-bound-max-element}), following the ideas in~\cite{huber2008fast}. Moreover, one can further prove that these considerable entries remain \(\Omega(1/n)\) in \(A^{(k)}\) (see \Cref{lem-lower-bound-elements}). Suppose, for contradiction, that some considerable entry were \(o(1/n)\); then some scaling factor in \(X\) or \(Y\) would have to be extremely small. 
To ensure that the minimal row and column sums of \(A^{(k)}\) remain bounded below while capping the considerable entries at \(O(1/n)\), a delicate compensatory balance is required: some entries in \(X\) and \(Y\) must be assigned very large values, while most must be set to be very small. This necessary balancing act inevitably forces an \(n/2 \times n/2\) submatrix of \(A^{(k)}\) to consist of extremely small entries, thereby contradicting the fact that \(A^{(k)}\) has an error of at most $t$. Consequently, each considerable entry in \(A^{(k)}\) must be \(\Theta(1/n)\).

Accordingly, each row of \(A^{(k)}\) contains at least \(\gamma n > n/2\) entries on the order of \(\Theta(1/n)\). 
Since every row sums to 1, we can establish an upper bound \(\ell < 1\) on the sum of the largest \(n/2\) entries in each row. Additionally, \(\max_j c_j\left(A^{(k+2)}\right)\) can be bounded above by a linear combination of \(c_1\left(A^{(k)}\right), \dots, c_n\left(A^{(k)}\right)\), with weights given by the entries in a row of \(A^{(k)}\). A similar bound holds for \(1/\min_j c_j\left(A^{(k+2)}\right)\). 
Combined with the upper bound $\ell$ on the sum of the largest \(n/2\) entries, one can show that either
$\max_{j} c_j\left(A^{(k+2)}\right) - 1 \leq \ell \left(\max_{j} c_j\left(A^{(k)}\right) - 1\right)$
or
$1/\min_{j} c_j\left(A^{(k+2)}\right) - 1 \leq \ell \left(1/\min_{j} c_j\left(A^{(k)}\right) - 1\right)$
(see \Cref{lem-max-accuracy-decrease}). 
This implies that, approximately, the maximum deviation is reduced by a constant factor every two iterations.
Hence, a nearly doubly stochastic matrix with a maximum error of \(\varepsilon\) is achieved in an additional \(O(\log n - \log \varepsilon)\) iterations.

In summary, the Sinkhorn-Knopp algorithm requires \(O(\log n - \log \varepsilon)\) iterations to achieve a maximum error of \(\varepsilon\) for \(A\).

\vspace{0.3cm}
\noindent \underline{\textbf{Preserving matrix simplicity through balanced element dynamics.}} 
To establish the lower bounds and prove \Cref{thm-lower-bound-positive}, we construct a positive matrix for which the Sinkhorn-Knopp algorithm converges slowly. Our construction is inspired by the proof of the upper bound, where Lemmas \ref{lem-upper-bound-max-element} and \ref{lem-lower-bound-elements} play key roles by demonstrating that the considerable entries in the input matrix become \(\Theta(1/n)\) in the scaled matrix.
To establish the lower bound, we design a matrix whose scaled entries deviate from these bounds. Specifically, for even \(n\), we construct an \(n \times n\) matrix \(A\) such that \(A_{n/2,n/2} = A_{n/2+1,n/2+1} = 1\) while all other entries in the \((n/2+1) \times (n/2+1)\) submatrix \(B\) (located in the bottom-left corner) and the \((n/2-1) \times (n/2-1)\) submatrix \(C\) (located in the top-right corner) are set to be very small. 
Consequently, the entries \(A_{n/2,n/2}\) and \(A_{n/2+1,n/2+1}\) emerge as the pivotal elements of \(A\). The intuition is that these two entries are the only dominant elements within the submatrix \(B\) of size \((n/2+1) \times (n/2+1)\), making them critical for \(\perman(A)\) and the closely related matrix scaling problem (see \Cref{thm-Frobenius-König}). As the algorithm iterates, these two entries grow to \(\omega(1/n)\), while all other entries in the \(n/2\) and \(n/2+1\) rows and columns, excluding those in the submatrix \(B\), tend to \(o(1/n)\).

Our analysis of the Sinkhorn-Knopp algorithm applied to the matrix \(A\) focuses on the evolution of three distinct groups of entries:
\begin{itemize}
\item the pivotal entries \(A_{n/2,n/2}\) and \(A_{n/2+1,n/2+1}\);
\item the small entries in the submatrices \(B\) and \(C\),  excluding \(A_{n/2,n/2}\) and \(A_{n/2+1,n/2+1}\);
\item the large entries in the \(n/2\) and \(n/2+1\) rows and columns, excluding those within \(B\).
\end{itemize}
We demonstrate the slow convergence of the algorithm by showing that the large entries decay slowly. 
However, if the small entries in \(B\) and \(C\) grow significantly in \(A^{(k)}\), then \(A^{(k)}_{n/2,n/2}\) and \(A^{(k)}_{n/2+1,n/2+1}\) will no longer be the sole pivotal entries in \(A^{(k)}\), and the overall evolution of the entries becomes much more complicated. Therefore, to establish the algorithm's slow convergence, we must prove that the large entries diminish to very small values before the small entries become substantially large, ensuring that \(A^{(k)}\) always maintains a simple structure. In other words, besides deriving an upper bound on the decay rate of the large entries, we need to establish a corresponding lower bound on that decay rate and an upper bound on the growth rate of the small entries. These bounds are obtained through a precise characterization of the relationships among these entries (see \Cref{con-lb-ak-bk-relation-positive}). Thus, by precisely controlling the interaction between the entries of the matrix, our method preserves a simple, structured matrix, ensuring that the pivotal entries \(A^{(k)}_{n/2,n/2}\) and \(A^{(k)}_{n/2+1,n/2+1}\) remain dominant.
\newpage
\section{Preliminary}\label{sec-preliminary}
\noindent \underline{\textbf{Tools for the Permanent.}} Given each matrix $Z$ of size $n\times n$,
the permanent of $Z$ is defined as 
\[\perman(Z) \triangleq \sum_{\sigma }\prod_{i\in [n]}Z_{i,\sigma(i)},\]
where the sum is over all permutations $\sigma$ of $[n]$.
A necessary and sufficient condition for the permanent of a square non-negative matrix to be zero was provided independently by Frobenius and König~\cite{horn2012matrix}.
\begin{theorem}[Frobenius-König]\label{thm-Frobenius-König}
Let A be an $n \times n$ non-negative matrix. Then
$\perman(A)= 0$ if and only if $A$ contains an $s \times t$ zero submatrix such that $s + t = n + 1$.
\end{theorem}

\begin{lemma}[\cite{hall1948distinct}]\label{lemma-perm-lb-hall}
Let $A$ be an $n\times n$ 0-1 matrix with at least $\gamma n$ entries equal to 1 in every row and column. If $\perman(A)>0$, then $\perman(A)\geq \lfloor \gamma n \rfloor!$
\end{lemma}

The following corollary follows immediately from \Cref{lemma-perm-lb-hall}.
\begin{corollary}\label{cor-perm-lb-hall}
Let $A$ be a $(\gamma,\rho)$-dense $n\times n$ matrix. If $\perman(A)>0$, then $\perman(A)\geq \rho^n\cdot \lfloor \gamma n \rfloor!$
\end{corollary}

The following lower bound on the permanent of doubly stochastic
matrices was first conjectured by
Van der Waerden and later proved independently by Falikman~\cite{falikman1981proof} and Egorychev~\cite{egorychev1981solution}.

\begin{lemma}
For any doubly stochastic matrix $A$ of size $n\times n$, we have $\perman(A)\geq n!/n^n$.
\end{lemma}

\noindent \underline{\textbf{Tools about the Sinkhorn-Knopp algorithm.}} Given a non-negative $n \times n$ matrix, the Sinkhorn–Knopp algorithm, applied to achieve $(\mathbf{1},\mathbf{1})$-scaling, iteratively generates a sequence of matrices $A^{(0)}, A^{(1)}, \ldots$ as follows:
\begin{itemize}
\item For each $i,j\in [n]$, let $A^{(0)}_{i,j} = A_{i,j}/r_i(A)$;
\item For each integer $k>0$ and $i,j\in [n]$, if $k$ is odd, let $A^{(k)}_{i,j} = A^{(k-1)}_{i,j}/c_j(A^{(k-1)})$; 
otherwise, let $A^{(k)}_{i,j} = A^{(k-1)}_{i,j}/r_i(A^{(k-1)})$.
\end{itemize}

The following are some easy facts about $A^{(0)},A^{(1)},\cdots$ and the proof can be found in~\Cref{appendix-fact}.
\begin{fact}\label{fact-c-a-a0-a1}
Let $A$ be an $n\times n$ matrix. Then the following holds for $A^{(0)},A^{(1)},\cdots$:
\begin{enumerate}[(1)]
\item $A^{(k)}_{i,j} \in [0,1]$ for each $i\in [n],j\in[n]$ and $k\geq 0$.\label{item-first-fact-a0-a1}
\item Assume $A$ is $(\gamma,\rho)$-dense. Then $A^{(0)}_{i,j}\in (0,1/(\rho\gamma n)]$ for each $i,j\in[n]$, and every row and column of $A^{(0)}$ contains at least $\gamma n$ entries that are no less than $\rho/n$.\label{item-second-fact-c-a-a0-a1}
\item Assume $A$ is $(\gamma,\rho)$-dense. For each $k\geq 0$ and $i\in [n]$, if $k$ is even, then $c_i\left(A^{(k)}\right)\in [\rho\gamma ,1/(\rho\gamma)]$ and $r_i\left(A^{(k)}\right) = 1$. Otherwise, $r_i\left(A^{(k)}\right)\in [\rho\gamma ,1/(\rho\gamma)]$ and $c_i\left(A^{(k)}\right) = 1$.\label{item-forth-fact-c-a-a0-a1}
\end{enumerate}
\end{fact}


The following lemma from~\cite{sinkhorn1964relationship} 
shows that the maximum (or minimal) row (or column) sum is monotonic in the Sinkhorn-Knopp algorithm. 

\begin{lemma}\label{lem-monotone-rsum-csum}
For any odd $k$, we have 
\[
\min_{i\in [n]}r_i\left(A^{(k)}\right) \leq  \min_{i\in [n]}r_i\left(A^{(k+2)}\right) \leq 1 \leq \max_{i\in [n]}r_i\left(A^{(k+2)}\right) \leq \max_{i\in [n]}r_i\left(A^{(k)}\right).
\]
Similarly, for any even $k$, we have
\[
\min_{j\in [n]}c_j\left(A^{(k)}\right) \leq  \min_{j\in [n]}c_j\left(A^{(k+2)}\right) \leq 1 \leq \max_{j\in [n]}c_j\left(A^{(k+2)}\right) \leq \max_{j\in [n]}c_j\left(A^{(k)}\right).
\]
\end{lemma}

Given an $n\times n$ matrix $A$, 
The following facts are well-known in the literature ~\cite{linial1998deterministic}.

\begin{lemma}\label{lem-condition-converge}
If $\perman(A)>0$, then Sinkhorn-Knopp algorithm converges.
\end{lemma}

\begin{lemma}\label{lem-facts-ak}
For any $i\in [n]$, let $x^{(0)}_i = y^{(0)}_i = 1$.
For any $k > 0$ and $i\in [n]$, let
$x^{(k)}_i = 1/\prod_{j=0}^{k-1} r_i\left(A^{(j)}\right)$
and $y^{(k)}_i = 1/\prod_{j=0}^{k-1} c_i\left(A^{(j)}\right)$. Then we have the following facts:
\begin{itemize}
\item For any odd $k \geq 0$, we have 
\begin{align}\label{eq-ub-ri}
\prod_{i\in [n]} r_i\left(A^{(k)}\right) \leq 1
\end{align}
\begin{align}\label{eq-increase-perm-ri}
\perman\left(A^{(k+1)}\right) = \perman\left(A^{(k)}\right)\prod_{i\in [n]} r_i^{-1}\left(A^{(k)}\right). 
\end{align}
Similarly, for any even $k\geq 0$, we have 
\begin{align}\label{eq-ub-ci}
\prod_{i\in [n]} c_i\left(A^{(k)}\right) \leq 1
\end{align}
\begin{align}\label{eq-increase-perm-ci}
\perman\left(A^{(k+1)}\right) = \perman\left(A^{(k)}\right)\prod_{i\in [n]} c_i^{-1}\left(A^{(k)}\right). 
\end{align}
\item 
For any $k \geq 0$, 
\begin{align}\label{eq-ak-a0-relation}
A^{(k)} = \diag\left(x^{(k)}_1,\cdots,x^{(k)}_n\right)A^{(0)}\diag\left(y^{(k)}_1,\cdots,y^{(k)}_n\right), \quad \prod_{i\in [n]}x^{(k)}_i\geq 1, \quad \prod_{i\in [n]}y^{(k)}_i\geq 1.
\end{align}
\item 
For any $k \geq 0$, 
\begin{align}\label{eq-ak-aprime-relation}
A^{(k)} = \diag\left(\frac{x^{(k)}_1}{r_1(A)},\cdots,
\frac{x^{(k)}_n}{r_n(A)}\right)\cdot A\cdot\diag\left(y^{(k)}_1,\cdots,y^{(k)}_n\right).
\end{align}
\end{itemize}
\end{lemma}

\noindent \underline{\textbf{Definitions about accuracy and deviation.}}
The following are some key quantities used in our proof.
\begin{definition}
A matrix $Z$ of size $n\times n$ is called \emph{standardized} if either $r_{i}(Z) = 1$ for each $i\in [n]$ or $c_{i}(Z) = 1$ for each $i\in [n]$.
A matrix $Z$ has \emph{column-accuracy} $\boldsymbol{\alpha} = (\alpha_1,\cdots,\alpha_n)$ if $r_{i}(Z) = 1$ for each $i\in [n]$ and 
\begin{align}
\forall j\in [n], \quad \abs{c_j(Z) - 1}\leq \alpha_j.
\end{align}
The definition of the \emph{row-accuracy} is similar.
We say a matrix $Z$ has accuracy $\boldsymbol{\alpha}$ if $Z$ has column-accuracy $\boldsymbol{\alpha}$ or row-accuracy $\boldsymbol{\alpha}$.
Given a matrix $Z$ with accuracy $\boldsymbol{\alpha}$,
define
\begin{align}
\alpha(Z) &\triangleq \frac{2}{n}\cdot \sum_{i\in [n]}\alpha_i.\label{eq-def-alpha}
\end{align}
\end{definition}
Intuitively, $\alpha(Z)$ depicts how far $Z$ is from a doubly stochastic matrix. Given a matrix $Z$, when the notation $\alpha(Z)$ is used, we always assume that $Z$ is standardized.

We say that an $n\times n$ matrix $Z$ has a maximum deviation $t$ if 
$\abs{r_i(Z) - 1}\leq t$ and $\abs{c_i(Z) - 1}\leq t$ for each $i\in [n]$. 

\newpage

\section{Upper bounds}\label{sec-upper-bounds}

In this section, we prove \Cref{thm-upper-bound-general} and \Cref{cor-upper-bound-most}.



Let $n,t,\varepsilon>0$. 
Let \(B\) be an \(n \times n\) matrix whose entries are drawn independently and uniformly from the interval \([0,t]\).
By Chernoff's bound, one can verify that with probability $1  - (\exp(-\Omega(n)))$ there exists some $\gamma\geq 6/11$ such that
$B/(\max_{i,j\in [n]} B_{i,j})$ is $(\gamma,2t/5)$-dense. 
Thus, \Cref{cor-upper-bound-most} follows from \Cref{thm-upper-bound-general}.

Given any nonnegative, nonzero matrix $B$,
let $A = B/(\max_{i,j\in[n]} B_{i,j})$.
Let $A^{(0)},A^{(1)},\cdots$ be the sequence of matrices constructed by the Sinkhorn-Knopp algorithm with $A$ as input,
and $B^{(0)},B^{(1)},\cdots$ be the sequence of matrices constructed by the Sinkhorn-Knopp algorithm with $B$ as input.
It is straightforward to verify that 
$A^{(0)} = B^{(0)}, A^{(1)} = B^{(1)}, \cdots$. 
Therefore, to prove \Cref{thm-upper-bound-general}, we may assume without loss of generality that 
$\max_{i,j\in[n]} B_{i,j} = 1$ and that $B$ is $\gamma$-dense.
Hence, \Cref{thm-upper-bound-general} is immediate by the following theorem.

\begin{theorem}\label{thm-main-ak-max-accuracy}
Let $\gamma \in (1/2,1]$, $\rho\in (0,1]$ and $\varepsilon > 0$. Let $A$ be a $(\gamma,\rho)$-dense $n\times n$ matrix provided as input to the Sinkhorn-Knopp algorithm, and let $A^{(0)},A^{(1)},\cdots$ denote the sequence of matrices generated by the algorithm.
Then there exists some 
\begin{align}\label{eq-def-k}
k = O\left(\rho^{-18}\left(2\gamma - 1\right)^{-5}\left(\log n - \log \varepsilon - \log \rho\right)\right)
\end{align}
such that for any $\ell\geq k$, we have 
\begin{align}\label{eq-main-ak-max-accuracy-goal}
\norm{r\left(A^{(\ell)}\right)- \boldsymbol{1}}_{1} + \norm{c\left(A^{(\ell)}\right)- \boldsymbol{1}}_{1} \leq \varepsilon.
\end{align}
\end{theorem}

In the following, we prove \Cref{thm-main-ak-max-accuracy}.

\subsection{Rapid Decay of Error in Phase 1}
In this section, we prove the following theorem.
\begin{theorem}\label{thm-complexity-sinkhorn}
Assume the condition of \Cref{thm-main-ak-max-accuracy}.
For any $t >0$, define 
\begin{align}\label{eq-def-K}
K = \left\{k\geq 0\mid \norm{r\left(A^{(k)}\right)- \boldsymbol{1}}_1 + \norm{c\left(A^{(k)}\right)- \boldsymbol{1}}_1 > tn\right\}.
\end{align}
Then we have 
$\abs{K} \leq 8t^{-2}(\log n - \log \rho)$.
\end{theorem}

The follow lemma is used in the proof of \Cref{thm-complexity-sinkhorn}. Our proof is inspired by Linial et al.~\cite{linial1998deterministic}, but it departs at a critical technical juncture. The original analysis simply neglects the cubic error term, whereas in our setting this term actually dominates the linear and quadratic contributions. To tame it, we refrain from expanding the product of all row (or column) sums at once; instead, we split the expansion into two parts—the product of sums greater than 1 and the product of sums less than 1—so that the series needs to be carried only up to the quadratic terms.

\begin{lemma}\label{lemma-ub-ci-ri}
Under the condition of \Cref{thm-complexity-sinkhorn},
we have for each even $k\in K$, 
\begin{align*}
\prod_{i\in [n]}c_i\left(A^{(k)}\right)\leq \exp\left(- \frac{nt^2}{8}\right).
\end{align*}
Similarly, for each odd $k\in K$, 
\begin{align*}
\prod_{i\in [n]}r_i\left(A^{(k)}\right)\leq \exp\left(- \frac{nt^2}{8}\right).
\end{align*}
\end{lemma}
\begin{proof}
Assume \emph{w.l.o.g.} that $k$ is even.
For simplicity, let $x_i = c_i\left(A^{(k)}\right) - 1$ for each $i\in [n]$.
By \Cref{fact-c-a-a0-a1}, we have $r_i\left(A^{(k)}\right) = 1$ for each $i\in [n]$.
Combined with \eqref{eq-def-K}, we have 
\begin{align}\label{eq-sum-absx}
\sum_{i\in [n]}\abs{x_i} = \norm{c\left(A^{(k)}\right)- \boldsymbol{1}}_1 > tn.
\end{align}
Moreover, by 
\[\sum_{i\in [n]}c_i\left(A^{(k)}\right) = \sum_{i\in [n]}r_i\left(A^{(k)}\right) = n,\]
we have 
\begin{align}\label{eq-sum-x}
\sum_{i\in [n]}x_i = 0.
\end{align}
Let $R \triangleq \{i\in [n]\mid x_i>0\}$ and $L \triangleq \sum_{i\in R} x_i$.
Together with the AM–GM inequality and the elementary bound $1+x \le e^{x}$ valid for all real $x$, we have
\begin{align}\label{eq-prod-x-upperbound}
\prod_{i\in R} (1+x_i)\leq \left(1+\frac{L}{\abs{R}}\right)^{\abs{R}} \leq \exp(L).
\end{align}
By \eqref{eq-sum-absx} and \eqref{eq-sum-x}, we have 
\begin{align}\label{eq-sum-negx}
\sum_{i \not\in R} x_i = -\sum_{i\in R} x_i = -L < -\frac{tn}{2}.
\end{align}
Combined with the Cauchy-Schwarz inequality, we have
\begin{align}\label{eq-sum-x2}
\sum_{i\not\in R}x^2_i \geq \frac{L^2}{n- \abs{R}}.
\end{align}
In addition, for each $i\not\in R$, we have $-1< x_i = c_i\left(A^{(k)}\right) - 1 \leq 0$.
Thus,
\[\ln (1+x_i) \leq x_i - \frac{x_i^2}{2}.\]
Therefore, 
\[1+x_i \leq \exp\left(x_i - \frac{x_i^2}{2}\right).\]
Combined with \eqref{eq-sum-x2} and \eqref{eq-sum-negx}, we have
\begin{align*}
\prod_{i\not\in R}c_i\left(A^{(k)}\right)&= \prod_{i\not\in R}(1+ x_i) \leq \prod_{i\not\in R}\exp\left(x_i - \frac{x_i^2}{2}\right) = \exp\left(\sum_{i\not\in R}\left(x_i - \frac{x_i^2}{2}\right)\right)\\
&\leq \exp\left(-L- \frac{L^2}{2(n-\abs{R})}\right).
\end{align*}
Combined with \eqref{eq-prod-x-upperbound} and \eqref{eq-sum-negx}, we have
\begin{align*}
\prod_{i\in [n]}c_i\left(A^{(k)}\right)&= \prod_{i\in [n]}(1+ x_i) = \prod_{i\in R}(1+ x_i)\prod_{i\not\in R}(1+ x_i)\leq   \exp(L)\cdot\exp\left(-L- \frac{L^2}{2(n-\abs{R})}\right)\\
& \leq \exp\left(- \frac{L^2}{2n}\right)\leq \exp\left(- \frac{nt^2}{8}\right).
\end{align*}
\end{proof}

Now we can prove \Cref{thm-complexity-sinkhorn}.
\begin{proof}
Assume for contradiction that 
\begin{align}\label{eq-lb-K}
\abs{K}> \log_{\exp(nt^2/8)} (n^n\rho^{-n}) = \frac{n(\log n - \log \rho)}{n t^2/8} = 8t^{-2}(\log n - \log \rho).
\end{align}
Let $L$ be the maximum number in $K$. 
By \eqref{eq-increase-perm-ri} and \eqref{eq-increase-perm-ci}, we have 
\begin{align*}
\perman\left(A^{(L+1)}\right) &= \perman\left(A^{(0)}\right)\left(\prod_{2k\in [L]}\prod_{i\in [n]} c_i^{-1}\left(A^{(2k)}\right)\right)\left(\prod_{2k+1\in [L]}\prod_{i\in [n]} r_i^{-1}\left(A^{(2k+1)}\right)\right).
\end{align*}
Combined with \eqref{eq-ub-ri} and \eqref{eq-ub-ci},
we have 
\begin{align*}
\perman\left(A^{(L+1)}\right) &\geq \perman\left(A^{(0)}\right)\left(\prod_{2k\in K}\prod_{i\in [n]} c_i^{-1}\left(A^{(2k)}\right)\right)\left(\prod_{2k+1\in K}\prod_{i\in [n]} r_i^{-1}\left(A^{(2k+1)}\right)\right).
\end{align*}
Combined with \Cref{lemma-ub-ci-ri},
we have 
\begin{align}\label{eq-perm-alplusone}
\perman\left(A^{(L+1)}\right) &\geq \perman\left(A^{(0)}\right)\exp\left( \frac{n\abs{K}t^2}{8}\right).
\end{align}
Moreover, by \Cref{thm-Frobenius-König} and $\gamma\in (1/2,1]$, we have $\perman(A)>0$.
Combined with 
\Cref{cor-perm-lb-hall} we have 
$\perman(A) \geq \rho^n\cdot\lfloor \gamma n \rfloor!$.
Hence,
\[\perman\left(A^{(0)}\right) \geq \rho^n\cdot \frac{\lfloor \gamma n \rfloor!}{n^n}.\]
Combined with \eqref{eq-lb-K} and \eqref{eq-perm-alplusone},
we have 
\begin{align}
\perman\left(A^{(L+1)}\right) >\lfloor \gamma n \rfloor! \geq 1.
\end{align}
However, by \Cref{fact-c-a-a0-a1} we have either $c_i\left(A^{(L+1)}\right) = 1$ for each $i\in [n]$, or $r_i\left(A^{(L+1)}\right) = 1$ for each $i\in [n]$.
Hence, we have $\perman\left(A^{(L+1)}\right) \leq 1$, a contradiction.
Thus, we have $\abs{K}\leq 8t^{-2}(\log n - \log \rho)$.

\end{proof}




\subsection{Maintaining the Original Dense Structure During Phase 2}
In this subsection, we demonstrate that the dense structure of the original matrix \(A\) is preserved in the scaled matrix \(A^{(k)}\) throughout Phase 2; specifically, the normalized matrix \(A^{(k)}/\max_{i,j} A^{(k)}_{i,j}\) remains \(\gamma\)-dense.
Our main result is the following lemma, which shows that a majority of elements in $A^{(k)}$ are lower bounded if $\alpha\left(A^{(k)}\right)$ is upper bounded.

\begin{lemma}\label{lem-number-of-large-entries}
Assume the condition of \Cref{thm-main-ak-max-accuracy}. Given any $k\geq 0$ where $\alpha\left(A^{(k)}\right)< 1 - 1/(2\gamma)$, 
let 
\[\theta \triangleq \frac{1}{27}\cdot \rho^{15}\gamma^5\left(2\gamma\left(1 -\alpha\left(A^{(k)}\right)\right) - 1\right)\left(2\gamma - 1 - \alpha\left(A^{(k)}\right)\right)^3.\]
We have 
\begin{align*}
\forall i\in [n], \quad
\abs{\left\{t\mid A^{(k)}_{i,t}> \frac{\theta}{n}\right\}}\geq \lceil\gamma n \rceil,\quad 
\abs{\left\{t\mid A^{(k)}_{t,i}>\frac{\theta}{n}\right\}}\geq \lceil\gamma n \rceil.
\end{align*}
\end{lemma}

The following lemma, adapted from \cite{huber2008fast},
provides an upper bound for the entries in the scaled matrix obtained in the Sinkhorn-Knopp algorithm.
While the original lemma applies only to 0–1 matrices, the modified lemma extends to nonnegative matrices.

\begin{lemma}\label{lem-upper-bound-max-element}
Let \( \rho \in (0,1] \) and \(\gamma \in (1/2,1]\). Let \( A \) be a $(\gamma,\rho)$-dense \( n \times n \) matrix, and let \( X \) and \( Y \) be diagonal matrices with positive diagonal entries. Suppose that 
\( B = XAY \) is satisfy the following conditions:
\begin{itemize}
\item $B$ is standardized and has entries in \([0,1]\);
\item $2\gamma - 1 - \alpha(B) > 0$;
\item $r_t(B)\in [\rho\gamma ,1/(\rho\gamma)]$ and $c_t(B)\in [\rho\gamma ,1/(\rho\gamma)]$ for each $t\in [n]$.
\end{itemize}
Then, for each \( i, j \in [n]\), we have
\[
B_{i,j} \leq \frac{3}{\rho^3\gamma\left(2\gamma - 1 - \alpha(B)\right)n}.
\]
\end{lemma}

\begin{proof}
By \( \rho \in (0,1] \) and \(\gamma \in (1/2,1]\), we have
$\rho\gamma\leq 1.$
For each $t\in [n]$, by $r_t(B)\in [\rho\gamma ,1/(\rho\gamma)]$ and $\rho\gamma \leq 1$,
we have 
\[
\abs{r_t(B) - 1}\leq \max\left\{1 - \rho\gamma, \frac{1}{\rho\gamma} - 1\right\}\leq \max\left\{1, \frac{1}{\rho\gamma} \right\}\leq \frac{1}{\rho\gamma}.
\]
Similarly, we also have $\abs{c_t(B) - 1}\leq 1/(\rho\gamma)$.
Define $\alpha_t \triangleq \abs{c_t(B) - 1}$.
Since $B$ is standardized, we may assume \emph{w.l.o.g.} that 
\begin{align}\label{eq-rowsum-columnsum}
\forall t\in [n],\quad r_t(B) = 1, \quad \alpha_t \leq \frac{1}{\rho\gamma}.
\end{align}
For convenience, denote \(x_i \triangleq X_{i,i}\) and \(y_i \triangleq Y_{i,i}\).
Fix one index $i$ of row and another index $j$ of column.
Then $B_{i,j} = x_i A_{i,j} y_j$.
By \eqref{eq-rowsum-columnsum}
we have 
\begin{align}
\sum_{\ell\in [n]}x_i A_{i,\ell} y_{\ell} = 1, \quad \sum_{k\in [n]}x_k A_{k,j}y_j \leq 1+\alpha_j.
\end{align}
Hence,
\begin{align}\label{eq-upperbound-xkakj-ylail}
\sum_{\ell \neq j}y_{\ell}A_{i,\ell} = (1- B_{i,j})/x_i, \quad \sum_{k\neq i}x_k A_{k,j} \leq (1+\alpha_j - B_{i,j})/y_j.
\end{align}
Define
\begin{align}\label{eq-definition-R-C}
R \triangleq \{\ell \neq j | A_{i,\ell} \geq \rho\},\quad C \triangleq \{k \neq i | A_{k,j} \geq \rho\}.
\end{align}
By $A$ is $(\gamma,\rho)$-dense, 
we have 
\[\abs{R} \geq \lceil \gamma n \rceil - 1,\quad \abs{C} \geq \lceil \gamma n \rceil - 1, \quad \abs{C}+\abs{R} \geq 2(\gamma n - 1).\]
By \eqref{eq-rowsum-columnsum}
we have 
\begin{align*}
\sum_{k\in C}\sum_{\ell\in R}B_{k,\ell} + \sum_{k\in C}\sum_{\ell\not\in R}B_{k,\ell} = \abs{C}, \quad
\sum_{k\in C}\sum_{\ell\in R}B_{k,\ell} + \sum_{k\not\in C}\sum_{\ell\in R}B_{k,\ell} \geq \abs{R}-\sum_{\ell\in R}\alpha_\ell.
\end{align*}
Thus, we have 
\begin{equation}
\begin{aligned}\label{eq-n-largerthan-s1tos4}
n &= \sum_{k\in C}\sum_{\ell\in R}B_{k,\ell} + \sum_{k\in C}\sum_{\ell\not\in R}B_{k,\ell} + \sum_{k\not\in C}\sum_{\ell\in R}B_{k,\ell} + \sum_{k\not\in C}\sum_{\ell\not \in R}B_{k,\ell}\\
&\geq \left(\sum_{k\in C}\sum_{\ell\in R}B_{k,\ell} + \sum_{k\in C}\sum_{\ell\not\in R}B_{k,\ell}\right) + \left(\sum_{k\in C}\sum_{\ell\in R}B_{k,\ell}+\sum_{k\not\in C}\sum_{\ell\in R}B_{k,\ell}\right) +  B_{i,j} - \sum_{k\in C}\sum_{\ell\in R}B_{k,\ell} \\
&\geq  B_{i,j}+\abs{C} + \abs{R} -\sum_{\ell\in R}\alpha_\ell - \sum_{k\in C}\sum_{\ell\in R}B_{k,\ell}.
\end{aligned}
\end{equation}
In addition, by \eqref{eq-definition-R-C} and \eqref{eq-upperbound-xkakj-ylail} we have 
\begin{equation}\label{eq-upperbound-s1b}
\begin{aligned}
\sum_{k\in C}\sum_{\ell\in R}B_{k,\ell} &= \sum_{k\in C}\sum_{\ell \in R}x_kA_{k,\ell}y_{\ell}
      \leq \left(\sum_{k\in C}x_k\right) \left(\sum_{\ell \in R}y_{\ell}\right) \leq \frac{1}{\rho^2}\cdot \left(\sum_{k\in C}x_kA_{k,j}\right) \left(\sum_{\ell \in R}y_{\ell}A_{i,\ell}\right)\\
      &\leq \frac{1}{\rho^2}\cdot\left(\sum_{k\neq i}x_kA_{k,j}\right) \left(\sum_{\ell \neq j}y_{\ell}A_{i,\ell}\right)
      \leq \frac{(1+\alpha_j - B_{i,j})}{\rho y_j}\cdot\frac{(1-B_{i,j})}{\rho x_i}. 
\end{aligned}
\end{equation}
Recall that $\abs{C}+\abs{R} \geq 2(\gamma n - 1)$.
Combined with \eqref{eq-n-largerthan-s1tos4} and \eqref{eq-upperbound-s1b}, we have
\begin{align*}
n 
\geq B_{i,j} + 2(\gamma n - 1)-\sum_{\ell\in R}\alpha_\ell - \frac{(1+\alpha_j - B_{i,j})}{\rho y_j}\cdot\frac{(1-B_{i,j})}{\rho x_i}.
\end{align*}
Thus, we have
\begin{equation}\label{eq-nbij-geq-tgamman-bij}
\begin{aligned}
&\quad n B_{i,j}\\
&\geq B^2_{i,j} + \left(2(\gamma n - 1)-\sum_{\ell\in R}\alpha_\ell\right)B_{i,j} -  \frac{(1+\alpha_j - B_{i,j})(1 - B_{i,j}) \cdot B_{i,j}}{\rho^2 x_iy_j}\\
(\text{by $B_{i,j} = x_iA_{i,j}y_{j}\leq x_iy_i$})\quad &\geq B^2_{i,j} + \left(2(\gamma n - 1)-\sum_{\ell\in R}\alpha_\ell\right)B_{i,j} -\frac{(1+\alpha_j - B_{i,j})(1 - B_{i,j})}{\rho^2} \\
  &= \left(2\gamma n - 2 + \frac{2+\alpha_j}{\rho^2}-\sum_{\ell\in R}\alpha_\ell \right)B_{i,j} - \frac{1+\alpha_j}{\rho^2} + \frac{(\rho^2-1)B^2_{i,j}}{\rho^2}\\
\left(\text{by $B_{i,j} \leq 1$, $\alpha_j \leq \frac{1}{\rho\gamma}$, $\rho\gamma\leq 1$}\right)\quad &\geq \left(2\gamma n - 2 + \frac{2+\alpha_j}{\rho^2}-\sum_{\ell\in R}\alpha_\ell \right)B_{i,j} - \frac{3}{\rho^3\gamma}.
\end{aligned}
\end{equation}
In addition, by 
\eqref{eq-def-alpha} we have
\[\sum_{\ell\in R}\alpha_\ell \leq \sum_{\ell\in [n]}\alpha_\ell \leq n\alpha(B).\]
Combined with $2\gamma - 1 - \alpha(B)>0$ and $\rho\in (0,1]$, we have 
\[2\gamma n - 2 + \frac{2+\alpha_j}{\rho^2}-\sum_{\ell\in R}\alpha_\ell - n \geq 2\gamma n - 2 + 2+\alpha_j-\sum_{\ell\in R}\alpha_\ell - n \geq (2\gamma - 1 - \alpha(B))n>0.\]
Combined with \eqref{eq-nbij-geq-tgamman-bij}, we have
\begin{align*}
B_{i,j} \leq  \frac{3}{\rho^3\gamma(2\gamma n - 2 + \rho^{-2}(2+\alpha_j)-\sum_{\ell\in R}\alpha_\ell - n ) } \leq \frac{3}{\rho^3\gamma(2\gamma - 1 - \alpha(B))n }.
\end{align*}
\end{proof}

For each considerable entry in the original matrix,
the following lemma provides a lower bound for the corresponding scaled entry in the matrix obtained through the Sinkhorn-Knopp algorithm.

\begin{lemma}\label{lem-lower-bound-elements}
Suppose \(a\geq \rho>0 \),   \(b>1 \) and \(\gamma
\in (1/2,1]\). 
Let \( A \) be an \( n \times n \) matrix with entries in \([0,1]\).
Define the diagonal matrices $X \triangleq \diag(x_1,\cdots,x_n)$ and $Y \triangleq \diag(y_1,\cdots,y_n)$, where the entries are nonnegative and satisfy $\prod_{i\in[n]}x_i\geq 1$, $\prod_{i\in[n]}y_i \geq 1$.
Assume $A$ and $B = XAY$ satisfy the following conditions:
\begin{itemize}
\item $\forall i,j\in [n]$, $A_{i,j} \leq a/n$ and $B_{i,j} \leq b/n$;
\item $B$ is standardized with $\alpha(B)< 1- 1/(2\gamma)$;
\item the minimal row sum and column sum of $B$ are at least $\rho\gamma$.
\item each row and each column of $A$ contains at least \(\gamma n\) entries with values at least \(\rho/n\).
\end{itemize}
Then for any $i,j\in [n]$ where $A_{i,j}\geq \rho/n$, we have $B_{i,j} > \theta/n$ where 
\begin{align}\label{eq-def-theta-constant}
\theta \triangleq \frac{\rho^6\gamma^2(2\gamma(1 -\alpha(B)) - 1)}{a^3b^3}.
\end{align}
\end{lemma}

\begin{proof}
Assume for contradiction that there exists some $k,\ell\in [n]$ where $A_{k,\ell}\geq \rho/n$ and $B_{k,\ell} \leq \theta/n$.
Combined with $B_{k,\ell} = x_kA_{k,\ell}y_{\ell}$, 
we have $x_{k}y_{\ell} \leq \theta/\rho$.
Thus, we have either $x_{k} \leq \sqrt{\theta/\rho}$ or $y_{\ell} \leq \sqrt{\theta/\rho}$.
Suppose \emph{w.l.o.g.} $x_{k}\leq \sqrt{\theta/\rho}$.
By 
\[
\sum_{j\in [n]}B_{k,j} = \sum_{j\in [n]}x_kA_{k,j}y_{j}\geq \rho\gamma\]
and $A_{k,\ell} \leq a/n$,
we have there exists some $j\in [n]$ where 
\begin{align}\label{eq-upper-bound-yj-rhogammathetaa}
y_{j}\geq \frac{\rho \gamma}{a}\sqrt{\frac{\rho}{\theta}}.
\end{align}
Let $R = \{i| A_{i,j}\geq \rho/n\}$.
By the definition of $A$,
we have $\abs{R}\geq \gamma n$.
In addition,
\begin{align}\label{eq-upper-bound-bij}
\forall i\in [n], \quad  B_{i,j} = x_iA_{i,j}y_j\leq \frac{b}{n}.
\end{align}
Thus, for each $i\in R$, we have 
\begin{equation}
\begin{aligned}\label{eq-xi-small}
&\quad x_i\\
(\text{by \eqref{eq-upper-bound-yj-rhogammathetaa} and \eqref{eq-upper-bound-bij}})\quad &\leq \frac{ab}{\rho^2\gamma} \cdot \sqrt{\frac{\theta}{\rho}} \\
(\text{by \eqref{eq-def-theta-constant}})\quad & = \sqrt{\frac{(2\gamma(1 -\alpha(B)) - 1)\rho}{ab}}\\
(\text{by $\gamma\in (1/2,1]$})\quad &\leq \sqrt{\frac{\rho}{ab}}\\
(\text{by $a\geq \rho > 0$, $b>1$})\quad & <1.
\end{aligned}
\end{equation}
Moreover, by \Cref{lem-facts-ak} we have 
\[\prod_{i\in [n]}x_i = \prod_{i\in R}x_i\prod_{t\in [n]\setminus R}x_t\geq 1.\]
Combined with \eqref{eq-xi-small} and $\abs{R}\geq \gamma n> n/2$, we have there exists some $i\in [n]\setminus R$ such that 
\begin{align}\label{eq-lower-bound-xk}
x_i> \frac{\rho^2\gamma}{ab}\cdot \sqrt{\frac{\rho}{\theta}}.
\end{align}
Let $C = \{t\in [n]| A_{i,t}\geq \rho/n\}$.
Similar to \eqref{eq-upper-bound-bij}, we have
\begin{align}\label{eq-upper-bound-bit}
\forall t\in C, \quad  B_{i,t} = x_iA_{i,t}y_t\leq \frac{b}{n}.
\end{align}
By \eqref{eq-lower-bound-xk} and \eqref{eq-upper-bound-bit}, we have
\begin{align}\label{eq-yj-small}
\forall t\in C, \quad   y_{t}< \frac{ab^2}{\rho^3\gamma} \cdot \sqrt{\frac{\theta}{\rho}}.
\end{align}
Combining \eqref{eq-xi-small}, \eqref{eq-yj-small} with $A_{r,t}\leq a/n$ for each $r,t\in [n]$,
we have 
\begin{equation}\label{eq-upper-bound-rc}
\begin{aligned}
\sum_{r\in R}\sum_{t\in C}x_{r}A_{r,t}y_{t}&< n^2\cdot \frac{a}{n}\cdot \frac{ab}{\rho^2\gamma} \cdot \sqrt{\frac{\theta}{\rho}} \cdot \frac{ab^2}{\rho^3\gamma} \cdot \sqrt{\frac{\theta}{\rho}} = \frac{a^3b^3\theta n}{\rho^6\gamma^2} = (2\gamma(1 -\alpha(B)) - 1)n,
\end{aligned}
\end{equation}
where the last equality is by \eqref{eq-def-theta-constant}.
In addition, by \eqref{eq-def-alpha} and $\abs{R}\geq n/2$,
we have 
\[\abs{R}\alpha(B) \geq \frac{2\abs{R}}{n}\sum_{r\in [n]}\abs{\sum_{t\in [n]}B_{r,t} -1} \geq \sum_{r\in R}\abs{\sum_{t\in [n]}B_{r,t} -1} \geq \abs{R} - \sum_{r\in R}\sum_{t\in [n]}B_{r,t}.\]
Thus, 
we have
\[\sum_{r\in R}\sum_{t\in [n]}B_{r,t} \geq (1 -\alpha(B))\abs{R}.\]
Thus, by $B$ is standardized, we have
\[\sum_{r\not\in R}\sum_{t\in [n]}B_{r,t} \leq n - (1 -\alpha(B))\abs{R}.\]
Similarly, we also have 
\[\sum_{r\in [n]}\sum_{t\not\in C}B_{r,t} \leq n - (1 -\alpha(B))\abs{C}.\]
Thus, we have 
\begin{align*}
\sum_{r\in R}\sum_{t\in C} B_{r,t} \geq n - \sum_{r\not\in R}\sum_{t\in [n]}B_{r,t} - \sum_{r\in [n]}\sum_{t\not\in C}B_{r,t} \geq n - (2n - (1 -\alpha(B))(\abs{R} + \abs{C})).
\end{align*}
Combined with $\abs{R} + \abs{C} \geq 2\gamma n$, we have
\begin{align*}
\sum_{r\in R}\sum_{t\in C} B_{r,t}\geq n - (2n - (1 -\alpha(B))\cdot 2\gamma n) \geq (2\gamma(1 -\alpha(B)) - 1)n.
\end{align*}
This is contradictory with \eqref{eq-upper-bound-rc}.
The lemma is proved.
\end{proof}

Now we can prove \Cref{lem-number-of-large-entries}.
\begin{proof}[Proof of \Cref{lem-number-of-large-entries}]
Recall that $A$ is $(\gamma,\rho)$-dense where $\gamma \in (1/2,1]$.
By \ref{item-first-fact-a0-a1} of \Cref{fact-c-a-a0-a1},
we have $A^{(k)}_{i,j}\in [0,1]$ for each $i,j\in [n]$.
By \ref{item-forth-fact-c-a-a0-a1} of \Cref{fact-c-a-a0-a1},
we have $c_i\left(A^{(k)}\right)\in [\rho\gamma ,1/(\rho\gamma)]$ and $r_i\left(A^{(k)}\right)\in [\rho\gamma ,1/(\rho\gamma)]$ for each $i\in [n]$.
By \eqref{eq-ak-aprime-relation}, we have 
$A^{(k)} = XAY$ for some diagonal $X$ and $Y$.
In addition, by $\alpha\left(A^{(k)}\right)< 1 - 1/(2\gamma) $ and $\gamma\in (1/2,1]$, we have $\alpha\left(A^{(k)}\right) < 2\gamma - 1$.
Thus, one can apply \Cref{lem-upper-bound-max-element} to $A^{(k)}$. Hence, for each $i,j\in [n]$ we have 
\begin{align}\label{eq-akij-upperbound-threeoverrhosquare}
A^{(k)}_{i,j} \leq \frac{3}{\rho^3\gamma\left(2\gamma - 1 - \alpha\left(A^{(k)}\right)\right)n}.
\end{align}
In addition, 
by \ref{item-first-fact-a0-a1} of \Cref{fact-c-a-a0-a1}, we have $A^{(0)}_{i,j}\in (0,1]$ for each $i,j\in[n]$.
By \ref{item-second-fact-c-a-a0-a1} of \Cref{fact-c-a-a0-a1}, we have every row and column of $A^{(0)}$ has at least $\gamma n$ entries no less than $\rho/n$ and $A^{(0)}_{i,j}\in (0,1/(\rho \gamma n)]$.
By \ref{item-forth-fact-c-a-a0-a1} of \Cref{fact-c-a-a0-a1}, we have $c_i\left(A^{(k)}\right)\geq \rho\gamma$ and $r_i\left(A^{(k)}\right)\geq \rho\gamma$.
Moreover, by $\alpha\left(A^{(k)}\right)< 1 - 1/(2\gamma)$ and $\gamma\in (1/2,1]$, we have 
\begin{align*}
\frac{3}{\rho^3\gamma\left(2\gamma - 1 - \alpha\left(A^{(k)}\right)\right)} \geq \frac{3}{\rho^3\gamma} > 1.
\end{align*}
Combined with 
\eqref{eq-ak-a0-relation} and \eqref{eq-akij-upperbound-threeoverrhosquare},
one can apply \Cref{lem-lower-bound-elements} to $A^{(0)}$ and $A^{(k)}$, where $A^{(0)},A^{(k)}$ are substituted for the matrices $A,B$ in the lemma, respectively.
Thus, for any $i,j\in [n]$ where $A^{(0)}_{i,j}\geq \rho/n$, we have 
\begin{align*}
A^{(k)}_{i,j} &> \frac{\rho^6\gamma^2\left(2\gamma\left(1 -\alpha\left(A^{\left(k\right)}\right)\right) - 1\right)}{27 n(\rho\gamma)^{-3}\left(\rho^3\gamma\left(2\gamma - 1 - \alpha\left(A^{(k)}\right)\right)\right)^{-3}} \\
&= \frac{\rho^{18}\gamma^8\left(2\gamma\left(1 -\alpha\left(A^{(k)}\right)\right) - 1\right)\left(2\gamma - 1 - \alpha\left(A^{(k)}\right)\right)^3}{27n}. 
\end{align*}
Since every row and column of $A^{(0)}$ has at least $\gamma n$ entries no less than $\rho/n$, the lemma follows immediately.
\end{proof}

\subsection{Rapid Decay of Error in Phase 2}
In this subsection, we prove the following lemma.
Given any even $k$, the lemma demonstrates that if most entries of each row and each column of $A^{(k)}$ are lower bounded, 
then the column sums of $A^{(k+2)}$ become much closer to 1 compared to those of $A^{(k)}$.

\begin{lemma}\label{lem-max-accuracy-decrease}
Let $A$ be an $n\times n$ matrix provided as input to the Sinkhorn-Knopp algorithm, and let $A^{(0)},A^{(1)},\cdots$ denote the sequence of matrices generated by the algorithm.
Let $L > n/2$ be an integer and $\theta>0$. Define $\tau = 1 - \theta\left(L - n/2 \right)/n$.
Given any even $k\geq 0$ where 
\begin{align*}
\forall i\in [n], \quad
\abs{\left\{t\mid A^{(k)}_{i,t}> \frac{\theta}{n}\right\}}\geq L,\quad 
\abs{\left\{t\mid A^{(k)}_{t,i}>\frac{\theta}{n}\right\}}\geq L,
\end{align*} 
we have at least one of the following inequalities holds:
\begin{align*}
\max_{j\in [n]} c_j\left(A^{(k+2)}\right) - 1 &\leq
\tau\cdot\left(\max_{j\in [n]} c_j\left(A^{(k)}\right) - 1\right)\\
\left(\min_{j\in [n]} c_j\left(A^{(k+2)}\right)\right)^{-1} - 1 &\leq \tau\cdot\left(\left(\min_{j\in [n]} c_j\left(A^{(k)}\right)\right)^{-1} - 1\right).
\end{align*}
\end{lemma}



\begin{proof}
Given any $k$,
let $c'_1\geq c'_2\geq \cdots \geq c'_n$ be the sorted sequence of $c_1\left(A^{(k)}\right), \cdots, c_n\left(A^{(k)}\right)$.
Given any $i\in [n]$, let $a_{i,1}\geq a_{i,2} \geq \cdots \geq a_{i,n}$ be the sorted sequence of $A^{(k)}_{i,1},A^{(k)}_{i,2},\cdots,A^{(k)}_{i,n}$.
In the following, we prove the lemma for even $k$ by considering two separate cases.
\begin{itemize}
\item $c'_{\lceil n/2\rceil} < 1$.
Recall that $k$ is even.
We have $A^{(k+1)}_{i,j} = A^{(k)}_{i,j}/c_j\left(A^{(k)}\right)$ and $r_i\left(A^{(k)}\right) = 1$
for each $i,j\in [n]$.
Thus, we have 
\begin{align*}
\forall i\in [n], \quad r_i\left(A^{(k+1)}\right) = \sum_{j\in [n]}A^{(k)}_{i,j}/c_j\left(A^{(k)}\right).
\end{align*}
Similarly, we also have  
\begin{align}\label{eq-cjakplus1}
\forall j\in [n], \quad c_j\left(A^{(k+2)}\right) = \sum_{i\in [n]}A^{(k+1)}_{i,j}/r_i\left(A^{(k+1)}\right), \quad \sum_{i\in [n]}A^{(k+1)}_{i,j} = c_j \left(A^{(k+1)}\right)= 1.
\end{align}
Therefore, we have 
\begin{align}\label{eq-upper-bound-maxcj-general}
\max_{j\in [n]} c_j\left(A^{(k+2)}\right) \leq \max_{i\in [n]} \frac{1}{r_i\left(A^{(k+1)}\right)} = \max_{i\in [n]} \left(\sum_{j\in [n]}\frac{A^{(k)}_{i,j}}{c_j\left(A^{(k)}\right)}\right)^{-1}.
\end{align}
Combined with $\sum_{j\in [n]}A^{(k)}_{i,j} = r_i \left(A^{(k)}\right)= 1$ and Jensen's inequality,
we have 
\begin{align}\label{eq-upper-bound-maxcj}
\max_{j\in [n]} c_j\left(A^{(k+2)}\right) \leq \max_{i\in [n]} \left(\sum_{j\in [n]}\frac{A^{(k)}_{i,j}}{c_j\left(A^{(k)}\right)}\right)^{-1} \leq \max_{i\in [n]} \sum_{j\in [n]}A^{(k)}_{i,j}c_j\left(A^{(k)}\right).
\end{align}
Moreover, by the definitions of $a_{i,j}$ and $c'_j$, we have 
\begin{align}\label{eq-sum-aijcprimej}
\forall i\in [n], \quad \sum_{j\in [n]}A^{(k)}_{i,j}c_j\left(A^{(k)}\right) \leq \sum_{j\in [n]}a_{i,j}c'_j
\end{align}
Combined with \eqref{eq-upper-bound-maxcj}, we have
\begin{align}\label{eq-relation-maxcj-maxsumaijcjprime}
\max_{j\in [n]} c_j\left(A^{(k+2)}\right) \leq \max_{i\in [n]}\sum_{j\in [n]}a_{i,j}c'_j.
\end{align}
Moreover, by $c'_{\lceil n/2\rceil} < 1$ and $c'_1 \geq \cdots\geq c'_n$,
we have $c'_{j} < 1$  for each $j> n/2$.
Hence, for each $i$,
\begin{align*}
\quad \quad \sum_{j\in [n]}a_{i,j}c'_j \leq \sum_{j\leq n/2}a_{i,j}c'_j + \sum_{j>n/2}a_{i,j}c'_j \leq \sum_{j\leq n/2}a_{i,j}c'_i+\sum_{j>n/2}a_{i,j} = \sum_{j\leq \frac{n}{2}}a_{i,j}(c'_i - 1) + \sum_{j\in [n]}a_{i,j}.
\end{align*}
In addition, we have 
\begin{align}\label{eq-sum-aij}
\sum_{j\in [n]}a_{i,j} = \sum_{j\in [n]}A^{(k)}_{i,j} = r_i\left(A^{(k)}\right) = 1.
\end{align}
Therefore, we have
\begin{align}\label{eq-upper-bound-sum-aijcjprime}
\sum_{j\in [n]}a_{i,j}c'_j &\leq \sum_{j\leq \frac{n}{2}}a_{i,j}(c'_i - 1) + \sum_{j\in [n]}a_{i,j} \leq 1 + \sum_{j\leq \frac{n}{2}}a_{i,j}(c'_i - 1) \leq 1 + (c'_1 - 1)\sum_{j\leq \frac{n}{2}}a_{i,j}.
\end{align}
Moreover, by the definitions of $a_{i,1},\cdots,a_{i,n}$, we have
\begin{align*}
a_{i,1} \geq a_{i,2} \geq \cdots \geq a_{i,L} \geq \frac{\theta}{n}.
\end{align*}
Combined with \eqref{eq-sum-aij} and the integer $L>n/2$, we have 
\begin{align}\label{eq-sum-aij-partial}
\sum_{j\leq \frac{n}{2}}a_{i,j} = 1 - \sum_{j>\frac{n}{2}}^{n} a_{i,j} \leq 1 - \sum_{j>\frac{n}{2}}^{L} a_{i,j} \leq  1 - \frac{\theta}{n}\cdot \left(L - \frac{n}{2} \right) =  \tau.
\end{align}
Combined with \eqref{eq-upper-bound-sum-aijcjprime}, we have
\begin{align*}
\sum_{j\in [n]}a_{i,j}c'_j &\leq 1 + (c'_1 - 1)\sum_{j\leq \frac{n}{2}}a_{i,j} \leq 1+\tau(c'_{1} - 1).
\end{align*}
Combined with \eqref{eq-relation-maxcj-maxsumaijcjprime}, we have 
\[
\max_{j\in [n]} c_j\left(A^{(k+2)}\right) \leq \max_{i\in [n]}\sum_{j\in [n]}a_{i,j}c'_j \leq 1+\tau(c'_{1} - 1).
\]
Therefore, 
\begin{align*}
\max_{j\in [n]} c_j\left(A^{(k+2)}\right) - 1 &\leq \tau(c'_{1} - 1)
= \tau\left(\max_{j\in [n]} c_j\left(A^{(k)}\right) - 1\right).
\end{align*}

\item $c'_{\lceil n/2\rceil} \geq 1$. Similar to \eqref{eq-upper-bound-maxcj-general}, one can also verify that
\begin{align}\label{eq-lower-bound-mincj-general}
\min_{j\in [n]} c_j\left(A^{(k+2)}\right) \geq \min_{i\in [n]} \left(\sum_{j\in [n]}\frac{A^{(k)}_{i,j}}{c_j\left(A^{(k)}\right)}\right)^{-1}.
\end{align}
In addition, by the definitions of $a_{i,j}$ and $c'_j$ we have 
\begin{align*}
\forall i\in [n], \quad \sum_{j\in [n]}\frac{A^{(k)}_{i,j}}{c_j\left(A^{(k)}\right)} \leq \sum_{j\in [n]}\frac{a_{i,j}}{c'_{n-j+1}}.
\end{align*}
Combined with \eqref{eq-lower-bound-mincj-general}, we have
\begin{align}\label{eq-lower-bound-mincj}
\min_{j\in [n]} c_j\left(A^{(k+2)}\right) \geq \min_{i\in [n]} \left(\sum_{j\in [n]}\frac{A^{(k)}_{i,j}}{c_j\left(A^{(k)}\right)}\right)^{-1} \geq \min_{i\in [n]}\left(\sum_{j\in [n]}\frac{a_{i,j}}{c'_{n-j+1}}\right)^{-1}.
\end{align}
In addition, by $c'_{\lceil n/2\rceil} \geq 1$
and $c'_1 \geq \cdots\geq c'_n$, 
we have $c'_{j} \geq 1$  for each $j\leq \lceil\frac{n}{2}\rceil$.
Thus, for each $i$,
\begin{align*}
\sum_{j\in [n]}\frac{a_{i,j}}{c'_{n-j+1}} 
\leq \sum_{j\leq \frac{n}{2}}\frac{a_{i,j}}{c'_n} + \sum_{j>\frac{n}{2}} \frac{a_{i,j}}{c'_{n-j+1}}
\leq \sum_{j\leq \frac{n}{2}}\frac{a_{i,j}}{c'_n} + \sum_{j>\frac{n}{2}} a_{i,j}
= \sum_{j\in [n]}a_{i,j} + \sum_{j\leq \frac{n}{2}}a_{i,j}\left(\frac{1}{c'_n} - 1\right).
\end{align*}
Combined with \eqref{eq-sum-aij} and \eqref{eq-sum-aij-partial}, we have
\begin{align*}
\sum_{j\in [n]}\frac{a_{i,j}}{c'_{n-j+1}} 
\leq &\sum_{j\in [n]}a_{i,j} + \sum_{j\leq \frac{n}{2}}a_{i,j}\left(\frac{1}{c'_n} - 1\right) 
\leq 1+\tau\left(\frac{1}{c'_{n}} - 1\right).
\end{align*}
Combined with \eqref{eq-lower-bound-mincj}, we have 
\begin{align*}
\min_{j\in [n]} c_j\left(A^{(k+2)}\right) &\geq  \left(1+\tau\left(\frac{1}{c'_{n}} - 1\right)\right)^{-1}.
\end{align*}
Therefore,
\begin{align*}
\left(\min_{j\in [n]} c_j\left(A^{(k+2)}\right)\right)^{-1} - 1 &\leq \tau\left(\frac{1}{c'_{n}} - 1\right) = \tau\left(\left(\min_{j\in [n]} c_j\left(A^{(k)}\right)\right)^{-1} - 1\right).
\end{align*}
\end{itemize}
\end{proof}


\subsection{Combining the Two Phases}

In this section, we complete the proof of \Cref{thm-main-ak-max-accuracy} by integrating the results from both phases.

\begin{proof}[Proof of \Cref{thm-main-ak-max-accuracy}]
By \Cref{lem-monotone-rsum-csum}, to prove this theorem, it is sufficient to prove the claim that there exist some 
even $k$ and odd $k'$ satisfying \eqref{eq-def-k} such that both $A^{(k)}$ and $A^{(k')}$ have a maximum deviation $\varepsilon$.
In the following, we prove this claim for even $k$.
The proof for odd $k'$ is similar.

Define 
\begin{align}
q &\triangleq 1 - \frac{8}{135}\cdot \rho^{18}\gamma^5\left(\gamma -\frac{1}{2}\right)^5\cdot\left(\gamma-\frac{9}{20}\right)^3,\label{eq-def-g1-gamma}\\
L &\triangleq \left\{\ell\geq 0\mid \norm{r\left(A^{(\ell)}\right)- \boldsymbol{1}}_1 + \norm{c\left(A^{(\ell)}\right)- \boldsymbol{1}}_1 > \frac{9n}{20}\cdot \left(1 - \frac{1}{2\gamma}\right)\right\}.\label{eq-def-L-set-large-error}\\
t & \triangleq (\log \varepsilon + \log \rho - \log n - 2)/\log q \label{eq-def-t}\\
k & \triangleq \min \left\{i \geq 0\mid i \text{ is even and } i\geq 2\abs{L}+ 4t + 2\right\}.\label{eq-def-k-upper-bound}
\end{align}
By \Cref{thm-complexity-sinkhorn}, we have 
\begin{align}\label{eq-size-l-upperbound}
\abs{L} =O((2\gamma - 1)^{-2}(\log n - \log\rho)).
\end{align}
One can also verify that 
\begin{align}\label{eq-log-ggamma}
-\log q \geq \frac{8}{135}\cdot \rho^{18}\gamma^5\left(\gamma -\frac{1}{2}\right)^5\cdot\left(\gamma-\frac{9}{20}\right)^3.
\end{align}
By \eqref{eq-def-t}, \eqref{eq-def-k-upper-bound}, \eqref{eq-size-l-upperbound} and \eqref{eq-log-ggamma},
we have
\eqref{eq-def-k} is satisfied.

Define 
$
T \triangleq \left\{0\leq j< k\ | j \text{ is even and } j\not\in L   \right\}.$
We have 
\begin{align}\label{eq-size-T}
\abs{T}\geq (k - 2)/2 - \abs{L} \geq 2t.
\end{align}
For each $j\in T$, we have $j$ is even and $j\not\in L$.
Thus, 
by \eqref{eq-def-L-set-large-error} we have 
\[\norm{r\left(A^{(j)}\right)- \boldsymbol{1}}_1 + \norm{c\left(A^{(j)}\right)- \boldsymbol{1}}_1 \leq \frac{9n}{20}\cdot \left(1 - \frac{1}{2\gamma}\right).\]
Combined with \eqref{eq-def-alpha} and \Cref{fact-c-a-a0-a1},
we have 
\begin{equation}
\begin{aligned}\label{eq-alpha-aell-leq-t}
&\quad\alpha\left(A^{(j)}\right) = \frac{2}{n}\sum_{i\in [n]} \abs{c_i\left(A^{(j)}\right) - 1}
=\frac{2}{n}\norm{c\left(A^{(j)}\right)- \boldsymbol{1}}_1
\\&= \frac{2}{n}\left(\norm{r\left(A^{(j)}\right)- \boldsymbol{1}}_1 + \norm{c\left(A^{(j)}\right)- \boldsymbol{1}}_1\right) \leq \frac{9}{10}\cdot \left(1 - \frac{1}{2\gamma}\right).
\end{aligned}
\end{equation}
For any matrix $B$,
define two functions $\theta(\cdot)$ and $\tau(\cdot)$ as follows:
\begin{align}
\theta\left(B\right) \triangleq \frac{1}{27}\cdot \rho^{18}\gamma^8\left(2\gamma\left(1 -\alpha\left(B\right)\right) - 1\right)\left(2\gamma - 1 - \alpha\left(B\right)\right)^3, \quad
\tau\left(B\right) \triangleq 1 - \theta\left(B\right) \left(\gamma - \frac{1}{2}\right)\label{eq-def-fun-tau}.
\end{align}
By \eqref{eq-alpha-aell-leq-t} and \eqref{eq-def-fun-tau} 
we have 
\begin{equation}
\begin{aligned}\label{eq-tau-aj-equal-q}
\tau\left(A^{(j)}\right) 
&\leq 1 - \frac{1}{27}\cdot \rho^{18}\gamma^8\left(\gamma -\frac{1}{2}\right)\cdot\left(2\gamma\left(\frac{1}{10} + \frac{9}{20\gamma}\right) - 1\right) \cdot \left(2\gamma - 1 - \frac{9}{10} + \frac{9}{20\gamma}\right)^3   \\
& = 1 - \frac{1}{27}\cdot \rho^{18}\gamma^8\left(\gamma -\frac{1}{2}\right)\cdot\left(\frac{1}{5}\left(\gamma - 
\frac{1}{2}\right) \right)\cdot\left(\frac{2}{\gamma}\left(\gamma-\frac{1}{2}\right)\left(\gamma-\frac{9}{20}\right)\right)^3\\
& = 1 - \frac{8}{135}\cdot \rho^{18}\gamma^5\left(\gamma -\frac{1}{2}\right)^5\cdot\left(\gamma-\frac{9}{20}\right)^3  = q.
\end{aligned}
\end{equation}
By \Cref{lem-number-of-large-entries} and \eqref{eq-alpha-aell-leq-t}, we have for each $j\in T$,
\begin{align*}
\forall i\in [n], \quad
\abs{\left\{\ell\mid A^{(j)}_{i,\ell}> \frac{\theta\left(A^{(j)}\right)}{n}\right\}}\geq \lceil\gamma n \rceil,\quad 
\abs{\left\{\ell\mid A^{(j)}_{\ell,i}>\frac{\theta\left(A^{(j)}\right)}{n}\right\}}\geq \lceil\gamma n \rceil.
\end{align*}
Combined with \Cref{lem-max-accuracy-decrease} and \eqref{eq-tau-aj-equal-q}, we have
one of the following two inequalities is true:
\begin{align}
\max_{i\in [n]} c_i\left(A^{(j+2)}\right) - 1 &\leq
\left(\max_{i\in [n]} c_i\left(A^{(j)}\right) - 1\right)\cdot \tau\left(A^{(j)}\right) = q\left(\max_{i\in [n]} c_i\left(A^{(j)}\right) - 1\right),\label{eq-maxcj-decrease}\\
\left(\min_{i\in [n]} c_j\left(A^{(j+2)}\right)\right)^{-1} - 1 &\leq \left(\left(\min_{i\in [n]} c_j\left(A^{(j)}\right)\right)^{-1} - 1\right)\cdot  \tau\left(A^{(j)}\right) = q\left(\left(\min_{i\in [n]} c_j\left(A^{(j)}\right)\right)^{-1} - 1\right).\label{eq-maxcj-inverse-decrease}
\end{align}

\vspace{0.2cm}

Let $S \triangleq \left\{j\in T\ | \ \eqref{eq-maxcj-decrease} \text{ holds for $j$}\right\}.$ 
At first, consider the case $\abs{S} \geq \abs{T}/2$. 
By \eqref{eq-size-T}, we have $\abs{S} \geq t$.
Let $\ell_0$ be the minimum element in $T$. We have
\begin{align*}
\max_{i\in [n]} c_i\left(A^{(k)}\right) - 1 = 
\left(\max_{i\in [n]} c_i\left(A^{(\ell_0)}\right) - 1\right)\cdot \prod_{j= \ell_0/2}^{k/2-1}\frac{\left(\max_{i\in [n]} c_i\left(A^{(2j+2)}\right) - 1\right)}{\left(\max_{i\in [n]} c_i\left(A^{(2j)}\right) - 1\right)}.
\end{align*}
By \Cref{lem-monotone-rsum-csum},
we have 
\[
\forall j\geq 0,\quad \frac{\left(\max_{i\in [n]} c_i\left(A^{(2j+2)}\right) - 1\right)}{\left(\max_{i\in [n]} c_i\left(A^{(2j)}\right) - 1\right)} \leq 1.
\]
Moreover, by the definitions of $T$ and $S$, one can verify that $j$ is even and $\ell_0 \leq j \leq k-2$ for each $j\in S$.
Combined with the above two inequalities, we have 
\begin{align*}
\max_{i\in [n]} c_i\left(A^{(k)}\right) - 1 \leq 
\left(\max_{i\in [n]} c_i\left(A^{(\ell_0)}\right) - 1\right)\cdot\prod_{j\in S}\frac{\left(\max_{i\in [n]} c_i\left(A^{(j+2)}\right) - 1\right)}{\left(\max_{i\in [n]} c_i\left(A^{(j)}\right) - 1\right)}.
\end{align*}
Combined with \eqref{eq-maxcj-decrease} and $\abs{S}\geq t$, we have
\begin{align*}
\max_{i\in [n]} c_i\left(A^{(k)}\right) - 1 \leq 
q^{\abs{S}}\cdot\left(\max_{i\in [n]} c_i\left(A^{(\ell_0)}\right) - 1\right) \leq q^{t}\cdot\left(\max_{i\in [n]} c_i\left(A^{(\ell_0)}\right) - 1\right).
\end{align*}
Meanwhile, by \ref{item-forth-fact-c-a-a0-a1} of \Cref{fact-c-a-a0-a1} and $\gamma \in (1/2,1]$, we have
$c_i\left(A^{(\ell_0)}\right)\leq 2/\rho$ for each $i\in [n]$.
Hence, 
\begin{align}\label{eq-maxci-upperbound-ggamma}
\max_{i\in [n]} c_i\left(A^{(k)}\right) - 1 \leq q^{t}(2\rho^{-1} -1)< 2q^{t}\rho^{-1} < \frac{\varepsilon}{2n},
\end{align}
where the last inequality is by \eqref{eq-def-t}.
Therefore, 
\begin{align*}
\norm{c\left(A^{(k)}\right)- \boldsymbol{1}}_1 &= \sum_{i\in [n]} \left(c_i\left(A^{(k)}\right) - 1\right)\cdot \id{c_i\left(A^{(k)}\right)>1} + \sum_{i\in [n]} \left(1 - c_i\left(A^{(k)}\right)\right)\cdot \id{c_i\left(A^{(k)}\right)<1}\\
&=2\sum_{i\in [n]} \left(c_i\left(A^{(k)}\right) - 1\right)\cdot \id{c_i\left(A^{(k)}\right)>1}\\
&\leq 2n\cdot\left(\max_{i\in [n]} c_i\left(A^{(k)}\right) - 1 \right)\\
&=\varepsilon.
\end{align*}
Moreover, by \ref{item-forth-fact-c-a-a0-a1} of \Cref{fact-c-a-a0-a1} and $k$ is even, we have $\norm{r\left(A^{(k)}\right)- \boldsymbol{1}}_1 = 0$.
Thus, \eqref{eq-main-ak-max-accuracy-goal} is proved.

At last, consider the other case $\abs{S} < \abs{T}/2$.
By \eqref{eq-size-T}, we have $\abs{T\setminus S}\geq t$.
Similar to \eqref{eq-maxci-upperbound-ggamma}, we have
\begin{align}
\left(\min_{i\in [n]} c_i\left(A^{(k)}\right)\right)^{-1} - 1 
< \frac{\varepsilon}{2n}.
\end{align} 
Thus, by $\min_{i\in [n]} c_i\left(A^{(k)}\right)\leq 1$ we have
\begin{align*}
1 - \min_{i\in [n]} c_i\left(A^{(k)}\right) = \min_{i\in [n]} c_i\left(A^{(k)}\right)\cdot\left(\left(\min_{i\in [n]} c_i\left(A^{(k)}\right)\right)^{-1} - 1\right)  \leq \left(\min_{i\in [n]} c_i\left(A^{(k)}\right)\right)^{-1} - 1 \leq \frac{\varepsilon}{2n}.
\end{align*}
Therefore, 
\begin{align*}
\norm{c\left(A^{(k)}\right)- \boldsymbol{1}}_1 &= \sum_{i\in [n]} \left(c_i\left(A^{(k)}\right) - 1\right)\cdot \id{c_i\left(A^{(k)}\right)>1} + \sum_{i\in [n]} \left(1 - c_i\left(A^{(k)}\right)\right)\cdot \id{c_i\left(A^{(k)}\right)<1}\\
&=2\sum_{i\in [n]} \left(1 - c_i\left(A^{(k)}\right)\right)\cdot \id{c_i\left(A^{(k)}\right)<1}\\
&\leq 2n\cdot\left(1 - \min_{i\in [n]} c_i\left(A^{(k)}\right)\right)\\
&=\varepsilon.
\end{align*}
Moreover, by \ref{item-forth-fact-c-a-a0-a1} of \Cref{fact-c-a-a0-a1} and $k$ is even, we have $\norm{r\left(A^{(k)}\right)- \boldsymbol{1}}_1 = 0$.
Thus, \eqref{eq-main-ak-max-accuracy-goal} is proved.

In summary, there exists an even $k$ satisfying \eqref{eq-def-k} such that $A^{(k)}$ has a maximum deviation $\varepsilon$.
The theorem is proved.

\end{proof}

\newpage
\section{Lower bounds}\label{sec-lower-bounds}
In this section, we prove Theorems \ref{thm-lower-bound-positive} and \ref{thm-lower-bound-main}.

\subsection{Tight Lower Bound for Positive Matrices}
In this subsection, we prove Theorem \ref{thm-lower-bound-positive} by constructing a positive matrix for which the Sinkhorn-Knopp algorithm converges slowly.
Theorem \ref{thm-lower-bound-positive} is immediate by the following result.

\begin{figure}[htbp]
  \centering
\[
\begin{tikzpicture}[%
  mtx/.style = {
    matrix of math nodes,
    nodes in empty cells,
    left delimiter = {[},
    right delimiter = {]},
    every node/.style = {
      anchor = base,
      text width = 1.5em,
      text height = 1.2ex,
      align = center,
      anchor = base east
    },
    row sep=0.5em,
  },
  hl/.style = {opacity = 0.3, rounded corners = 2pt, inner sep = -2pt},
  index/.style = {font=\small, anchor=center}, 
  txtup/.style = {rotate = 90, right},
  txtbt/.style = {yshift = -1ex}
]

\matrix (M) [mtx] {
  1	&	$\cdots$	&	1	&	$\beta$	&	$\cdots$	&	$\beta$	&	$\frac{2\lceil \gamma n \rceil}{n}$	&	$\frac{2\lceil \gamma n \rceil}{n}$	&	$\beta$	&	$\beta$	&		&	$\cdots$	&		&	$\beta$	\\
$\beta$	&	1	&	$\cdots$	&	1	&	$\cdots$	&	$\beta$	&	$\frac{2\lceil \gamma n \rceil}{n}$	&	$\frac{2\lceil \gamma n \rceil}{n}$	&	$\beta$	&	$\beta$	&	$\ddots$	&	$\ddots$	&	$\ddots$	&	$\beta$	\\
$\vdots$	&	$\ddots$	&	$\ddots$	&	$\ddots$	&	$\ddots$	&	$\vdots$	&	$\vdots$	&	$\vdots$	&	$\vdots$	&	$\vdots$	&	$\ddots$	&	$\ddots$	&	$\ddots$	&	$\vdots$	\\
1	&	$\cdots$	&	$\beta$	&	$\cdots$	&	$\beta$	&	1	&	$\frac{2\lceil \gamma n \rceil}{n}$	&	$\frac{2\lceil \gamma n \rceil}{n}$	&	$\beta$	&	$\beta$	&		&	$\cdots$	&		&	$\beta$	\\
$\beta$	&	$\beta$	&		&	$\cdots$	&	$\beta$	&	$\beta$	&	1	&	$\beta$	&	1	&	1	&		&	$\cdots$	&		&	1	\\
$\beta$	&	$\beta$	&		&	$\cdots$	&	$\beta$	&	$\beta$	&	$\beta$	&	1	&	1	&	1	&		&	$\cdots$	&		&	1	\\
$\beta$	&	$\beta$	&		&	$\cdots$	&		&	$\beta$	&	$\beta$	&	$\beta$	&	1	&	$\cdots$	&	1	&	$\beta$	&	$\cdots$	&	$\beta$	\\
$\beta$	&	$\beta$	&	$\ddots$	&	$\ddots$	&	$\ddots$	&	$\beta$	&	$\beta$	&	$\beta$	&	$\beta$	&	1& $\cdots$	&	1	&	 $\cdots$	&	$\beta$	\\
$\vdots$	&	$\vdots$	&	$\ddots$	&	$\ddots$	&	$\ddots$	&	$\vdots$	&	$\vdots$	&	$\vdots$	&	$\vdots$	&	$\ddots$	&	$\ddots$	&	$\ddots$	&	$\ddots$	&	$\vdots$	\\
$\beta$	&	$\beta$	&		&	$\cdots$	&		&	$\beta$	&	$\beta$	&	$\beta$	&	1	&	$\cdots$	&	$\beta$	&	$\cdots$	& $\beta$	&	1	\\
};

\node at ([yshift=1.5em]M-1-1.north) {$1$};
\node at ([yshift=1.5em]M-1-2.north) {$\cdots$};
\node at ([yshift=1.5em]M-1-3.north) {$\lceil \gamma n \rceil$};
\node at ([yshift=1.5em]M-1-5.north) {$\cdots$};
\node at ([yshift=1.5em]M-1-7.north) {$\tfrac{n}{2}$};
\node at ([yshift=1.5em]M-1-8.north) {$\tfrac{n}{2}+1$};
\node at ([yshift=1.5em]M-1-9.north) {$\cdots$};
\node at ([yshift=1.5em]M-1-11.north) {$\tfrac{n}{2}+\lceil \gamma n \rceil+1$};
\node at ([yshift=1.5em]M-1-13.north) {$\cdots$};
\node at ([yshift=1.5em]M-1-14.north) {$n$};

\node at ([xshift=-2em]M-1-1.west) {$1$};
\node at ([xshift=-2em]M-3-1.west) {$\vdots$};
\node at ([xshift=-2em]M-5-1.west) {$\tfrac{n}{2}$};
\node at ([xshift=-2.5em]M-6-1.west) {$\tfrac{n}{2}+1$};
\node at ([xshift=-2em]M-8-1.west) {$\vdots$};
\node at ([xshift=-2em]M-10-1.west) {$n$};

\begin{scope}[on background layer]
  \node[hl, fill = blue, fit=(M-1-1)(M-1-3)] {};
  \node[hl, fill = blue, fit=(M-2-2)(M-2-4)] {};
  \node[hl, fill = blue, fit=(M-4-1)(M-4-2)] {};
  \node[hl, fill = blue, fit=(M-4-6)(M-4-6)] {};
  
  \node[hl, fill = orange, fit=(M-1-4)(M-1-6)] {};
  \node[hl, fill = orange, fit=(M-2-1)(M-2-1)] {};
  \node[hl, fill = orange, fit=(M-2-5)(M-2-6)] {};
  \node[hl, fill = orange, fit=(M-4-3)(M-4-5)] {};
  
  \node[hl, fill = blue!20, fit=(M-7-9)(M-7-11)] {};
  \node[hl, fill = blue!20, fit=(M-8-10)(M-8-12)] {};
  \node[hl, fill = blue!20, fit=(M-10-9)(M-10-10)] {};
  \node[hl, fill = blue!20, fit=(M-10-14)(M-10-14)] {};

  \node[hl, fill = orange!50, fit=(M-7-12)(M-7-14)] {};
  \node[hl, fill = orange!50, fit=(M-8-9)(M-8-9)] {};
  \node[hl, fill = orange!50, fit=(M-8-13)(M-8-14)] {};
  \node[hl, fill = orange!50, fit=(M-10-11)(M-10-13)] {};
  
  \node[hl, fill = black, fit=(M-5-7)(M-5-7)] {};
  \node[hl, fill = black, fit=(M-6-8)(M-6-8)] {};
  
  \node[hl, fill = blue!65, fit=(M-5-8)(M-5-8)] {};
  \node[hl, fill = blue!65, fit=(M-6-7)(M-6-7)] {};
  
  \node[hl, fill = red, fit=(M-5-1)(M-6-6)] {};
  \node[hl, fill = green, fit=(M-5-9)(M-6-14)] {};
  
  \node[hl, fill = yellow, fit=(M-7-7)(M-10-8)] {};
  \node[hl, fill = gray!50, fit=(M-7-1)(M-10-6)] {};
  \node[hl, fill = yellow!50, fit=(M-1-9)(M-4-14)] {};
  \path[hl, fill = magenta] 
    ([xshift = -1em, yshift = 2.5ex] M-1-7) 
    rectangle 
    ([xshift = 1.2em, yshift = -1.3ex] M-4-8);
\end{scope}

\end{tikzpicture}
\]
\caption{the matrix $A$ in Theorem 4.1}
  \label{fig-1}
\end{figure}

\begin{theorem}\label{thm-lowerbound-positive}
Let $\gamma\in (1/4,1/2)$ and $\varepsilon < 1/1000$.
Moreover, let \(n > 10000\) be an even integer satisfying \(n - 2\lceil \gamma n \rceil \geq 2\) and define $\beta \triangleq \varepsilon^{8}/(100 n^{61})$. 
Let $A$ be an $n\times n$ matrix satisfying the following conditions:
\begin{itemize}
    \item $A_{i,j} = 1$ if and only if one of the following is true:
\begin{itemize} 
\item $i,j\leq n/2-1$ and $0 \leq (j - i) \mod (n/2 - 1) < \lceil\gamma n\rceil$;
\item $i,j> n/2+1$ and $0 \leq (j - i) \mod (n/2 - 1) < \lceil\gamma n\rceil$;

\item $n/2\leq i\leq n/2+1$ and $j> n/2+1$;
\item $i=j = n/2$;
\item $i=j = n/2+1$.
\end{itemize}
\item $A_{i,j} = 2\lceil \gamma n \rceil/n$ if and only if $i\leq n/2-1$ and $n/2 \leq j\leq n/2+1$.
\item Otherwise, $A_{i,j} = \beta$.
\end{itemize}
Then $A$ is a positive matrix with density $\gamma$, and the Sinkhorn-Knopp algorithm converges on this matrix.
Moreover, there exists some 
$t = \Omega(n/\varepsilon )$ such that for each $k \leq t$,
\[\norm{r\left(A^{(k)}\right)- \boldsymbol{1}}_1 + \norm{c\left(A^{(k)}\right)- \boldsymbol{1}}_1 \geq \varepsilon.\]
If $\varepsilon \leq 1/(1000\sqrt{n})$, there also exists some 
$t = \Omega(\sqrt{n}/\varepsilon )$ such that for each $k \leq t$,
\[\norm{r\left(A^{(k)}\right)- \boldsymbol{1}}_2 + \norm{c\left(A^{(k)}\right)- \boldsymbol{1}}_2 \geq \varepsilon.\]
\end{theorem}

The matrix $A$ in \Cref{thm-lowerbound-positive} is as shown in \Cref{fig-1}.

The following lemma, whose proof follows immediately by induction (see \Cref{appendix-sec-lowerbound}), is used in the proof of \Cref{thm-lowerbound-positive}.


\begin{figure}[htbp]
  \centering
\[
\begin{tikzpicture}[%
  mtx/.style = {
    matrix of math nodes,
    nodes in empty cells,
    left delimiter = {[},
    right delimiter = {]},
    every node/.style = {
      anchor = base,
      text width = 1.5em,
      text height = 1.2ex,
      align = center,
      anchor = base east
    },
    row sep=0.5em,
  },
  hl/.style = {opacity = 0.3, rounded corners = 2pt, inner sep = -2pt},
  index/.style = {font=\small, anchor=center}, 
  txtup/.style = {rotate = 90, right},
  txtbt/.style = {yshift = -1ex}
]

\matrix (M) [mtx] {
  1	&	$\cdots$	&	1	&	$\beta$	&	$\cdots$	&	$\beta$	&	$a_k$	&	$a_k$	&	$u_k$	&	$u_k$	&	$\cdots$	&	$\cdots$	&	$\cdots$	&	$u_k$	\\
$\beta$	&	1	&	$\cdots$	&	1	&	$\cdots$	&	$\beta$	&	$a_k$	&	$a_k$	&	$u_k$	&	$u_k$	&	$\ddots$	&	$\ddots$	&	$\ddots$	&	$u_k$	\\
$\vdots$	&	$\ddots$	&	$\ddots$	&	$\ddots$	&	$\ddots$	&	$\vdots$	&	$\vdots$	&	$\vdots$	&	$\vdots$	&	$\vdots$	&	$\ddots$	&	$\ddots$	&	$\ddots$	&	$\vdots$	\\
1	&	$\cdots$	&	$\beta$	& $\cdots$ & $\beta$		&	1	&	$a_k$	&	$a_k$	&	$u_k$	&	$u_k$	&	$\cdots$	&	$\cdots$	&	$\cdots$	&	$u_k$	\\
$x_k$	&	$x_k$	&	$\cdots$	&	$\cdots$	&	$x_k$	&	$x_k$	&	1	&	$\beta$	&	$b_k$	&	$b_k$	&	$\cdots$	&	$\cdots$	&	$\cdots$	&	$b_k$	\\
$x_k$	&	$x_k$	&	$\cdots$	&	$\cdots$	&	$x_k$	&	$x_k$	&	$\beta$	&	1	&	$b_k$	&	$b_k$	&	$\cdots$	&	$\cdots$	&	$\cdots$	&	$b_k$	\\
$v_k$	&	$v_k$	&	$\cdots$	&	$\cdots$	&	$\cdots$	&	$v_k$	&	$y_k$	&	$y_k$	&	1	&	$\cdots$	&	1	&	$\beta$	&	$\cdots$	&	$\beta$	\\
$v_k$	&	$v_k$	&	$\ddots$	&	$\ddots$	&	$\ddots$	&	$v_k$	&	$y_k$	&	$y_k$	&	$\beta$	&	1	&	$\cdots$	& 1	&	$\cdots$	&	$\beta$	\\
$\vdots$	&	$\vdots$	&	$\ddots$	&	$\ddots$	&	$\ddots$	&	$\vdots$	&	$\vdots$	&	$\vdots$	&	$\vdots$	&	$\ddots$	&	$\ddots$	&	$\ddots$	&	$\ddots$	&	$\vdots$	\\
$v_k$	&	$v_k$	&	$\cdots$	&	$\cdots$	&	$\cdots$	&	$v_k$	&	$y_k$	&	$y_k$	&	1	&	$\cdots$	&	$\beta$	&	$\cdots$	& $\beta$	&	1	\\
};

\node at ([yshift=1.5em]M-1-1.north) {$1$};
\node at ([yshift=1.5em]M-1-2.north) {$\cdots$};
\node at ([yshift=1.5em]M-1-3.north) {$\lceil \gamma n \rceil$};
\node at ([yshift=1.5em]M-1-5.north) {$\cdots$};
\node at ([yshift=1.5em]M-1-7.north) {$\tfrac{n}{2}$};
\node at ([yshift=1.5em]M-1-8.north) {$\tfrac{n}{2}+1$};
\node at ([yshift=1.5em]M-1-9.north) {$\cdots$};
\node at ([yshift=1.5em]M-1-11.north) {$\tfrac{n}{2}+\lceil \gamma n \rceil+1$};
\node at ([yshift=1.5em]M-1-13.north) {$\cdots$};
\node at ([yshift=1.5em]M-1-14.north) {$n$};

\node at ([xshift=-2em]M-1-1.west) {$1$};
\node at ([xshift=-2em]M-3-1.west) {$\vdots$};
\node at ([xshift=-2em]M-5-1.west) {$\tfrac{n}{2}$};
\node at ([xshift=-2.5em]M-6-1.west) {$\tfrac{n}{2}+1$};
\node at ([xshift=-2em]M-8-1.west) {$\vdots$};
\node at ([xshift=-2em]M-10-1.west) {$n$};

\begin{scope}[on background layer]
  \node[hl, fill = blue, fit=(M-1-1)(M-1-3)] {};
  \node[hl, fill = blue, fit=(M-2-2)(M-2-4)] {};
  \node[hl, fill = blue, fit=(M-4-1)(M-4-2)] {};
  \node[hl, fill = blue, fit=(M-4-6)(M-4-6)] {};
  
  \node[hl, fill = orange, fit=(M-1-4)(M-1-6)] {};
  \node[hl, fill = orange, fit=(M-2-1)(M-2-1)] {};
  \node[hl, fill = orange, fit=(M-2-5)(M-2-6)] {};
  \node[hl, fill = orange, fit=(M-4-3)(M-4-5)] {};
  
  \node[hl, fill = blue!20, fit=(M-7-9)(M-7-11)] {};
  \node[hl, fill = blue!20, fit=(M-8-10)(M-8-12)] {};
  \node[hl, fill = blue!20, fit=(M-10-9)(M-10-10)] {};
  \node[hl, fill = blue!20, fit=(M-10-14)(M-10-14)] {};

  \node[hl, fill = orange!50, fit=(M-7-12)(M-7-14)] {};
  \node[hl, fill = orange!50, fit=(M-8-9)(M-8-9)] {};
  \node[hl, fill = orange!50, fit=(M-8-13)(M-8-14)] {};
  \node[hl, fill = orange!50, fit=(M-10-11)(M-10-13)] {};
  
  \node[hl, fill = black, fit=(M-5-7)(M-5-7)] {};
  \node[hl, fill = black, fit=(M-6-8)(M-6-8)] {};
  
  \node[hl, fill = blue!65, fit=(M-5-8)(M-5-8)] {};
  \node[hl, fill = blue!65, fit=(M-6-7)(M-6-7)] {};
  
  \node[hl, fill = red, fit=(M-5-1)(M-6-6)] {};
  \node[hl, fill = green, fit=(M-5-9)(M-6-14)] {};
  
  \node[hl, fill = yellow, fit=(M-7-7)(M-10-8)] {};
  \node[hl, fill = gray!50, fit=(M-7-1)(M-10-6)] {};
  \node[hl, fill = yellow!50, fit=(M-1-9)(M-4-14)] {};
  \path[hl, fill = magenta] 
    ([xshift = -1em, yshift = 2.5ex] M-1-7) 
    rectangle 
    ([xshift = 1.2em, yshift = -1.3ex] M-4-8);
\end{scope}

\end{tikzpicture}
\]
\caption{the matrix $A^{(k)}$ in Lemma 4.2}
  \label{fig-2}
\end{figure}

\newpage
\begin{lemma}\label{lemma-lb-equal-item-density-positive}
Under the condition of \Cref{thm-lowerbound-positive}, we have 
\begin{enumerate}
\item $A^{(k)}_{i,j} = A^{(k)}_{i',j'}$ for each $i,i',j,j'\leq n/2-1$ where $A_{i,j} = A_{i',j'}$
and each $k\geq 0$. \label{item-1-lemma-positive}
\item $A^{(k)}_{i,j} = A^{(k)}_{i',j'}$ for each $n/2+2\leq i,i',j,j' \leq n$ where $A_{i,j} = A_{i',j'}$ and each $k\geq 0$;  \label{item-2-lemma-positive}
\item $A^{(k)}_{i,j} = A^{(k)}_{i',j'}$ for each $k\geq 0$, $i,i'\geq n/2+2$ and $j,j'\leq n/2-1$; \label{item-3-lemma-positive}
\item $A^{(k)}_{i,j} = A^{(k)}_{i',j'}$ for each $k\geq 0$, $i,i'\leq n/2-1$, and $n/2+2\leq j,j' \leq n$;  \label{item-4-lemma-positive}
\item for each $k\geq 0$, we have 
\begin{align}\label{eq-equal-item-5}
A^{(k)}_{n/2,n/2} = A^{(k)}_{n/2+1,n/2+1},
\end{align}
\begin{align}\label{eq-equal-item-6}
A^{(k)}_{n/2,n/2+1} = A^{(k)}_{n/2+1,n/2}.
\end{align}
\begin{align}\label{eq-equal-item-1}
\forall i,j, \ell\leq n/2-1,\quad\qquad A^{(k)}_{i,n/2} = A^{(k)}_{j,n/2} = A^{(k)}_{\ell,n/2+1},
\end{align}
\begin{align}\label{eq-equal-item-2}
\forall i,j, \ell\leq n/2-1, \quad\qquad A^{(k)}_{n/2,i} = A^{(k)}_{n/2,j} = A^{(k)}_{n/2+1,\ell},
\end{align}
\begin{align}\label{eq-equal-item-3}
\forall n/2+2\leq i,j,\ell \leq n, \quad\quad A^{(k)}_{n/2,i} = A^{(k)}_{n/2,j} = A^{(k)}_{n/2+1,\ell}.
\end{align}
\begin{align}\label{eq-equal-item-4}
\forall n/2+2\leq i,j,\ell \leq n,\quad\quad A^{(k)}_{i,n/2} = A^{(k)}_{j,n/2} = A^{(k)}_{\ell,n/2+1}.
\end{align}
\label{item-5-lemma-positive}
\end{enumerate}
\end{lemma}

The matrix $A^{(k)}$ in \Cref{lemma-lb-equal-item-density-positive} is as shown in \Cref{fig-2}.

According to \Cref{lemma-lb-equal-item-density-positive}, the entries of \(A^{(k)}\) can be grouped into several distinct classes, with all entries within each class being identical. To capture the evolution of these groups, we define the following key entries:
\[a_k = A^{(k)}_{1,n/2},\quad b_k = A^{(k)}_{n/2,n}, \quad x_k = A^{(k)}_{n/2,1}, \quad y_k = A^{(k)}_{n,n/2},\quad u_k = A^{(k)}_{1,n}, \quad v_k = A^{(k)}_{n,1}.\]

Let
\begin{align}
&\frac{\varepsilon}{n}\leq \delta < \frac{1}{1000n},\label{eq-condition-delta} \\
&\rho \triangleq \frac{\delta^{2}}{100n^{2}}.\label{eq-def-rho-positive}
\end{align}
The following lemma plays a key role in establishing \Cref{thm-lowerbound-positive}. It shows that as \(k\) increases, the values \(a_k\) and \(b_k\) decrease slowly, whereas \(x_k\), \(y_k\), \(u_k\), and \(v_k\) increase slowly.

\begin{lemma}\label{lem-ab-bk-decay-slowly}
There exists some $t = \Omega(1/\delta )$ such that for each $k \leq t$,
\begin{align*}
\min\{a_{k},b_{k}\} \geq \delta, \quad \max\{x_{k},y_{k},u_{k},v_{k}\} < \frac{\rho}{4n}.
\end{align*}
\end{lemma}

The proof of \Cref{lem-ab-bk-decay-slowly} is quite complicated. We must precisely characterize the decay rates of \(a_k\) and \(b_k\) by establishing a sharp upper bound (see \eqref{eq-bk-recursion-one-step}) as well as a corresponding lower bound (see \eqref{eq-bk-leq-bkminustwo-oneplusbkminustwo}). Additionally, we need to derive an upper bound on the growth rate of \(x_k\), \(y_k\), \(u_k\), and \(v_k\) (see \eqref{eq-relation-maxratioofxyuv-ratiob}). Together, these bounds imply that \(\max\{x_k,y_k,u_k,v_k\}\) cannot become excessively large before \(\min\{a_k,b_k\}\) has decreased significantly.  
As a result, \(x_k\), \(y_k\), \(u_k\), and \(v_k\) are \(O(\rho/n)\) and can be safely neglected in the calculation for each $k \leq t$.
Consequently, there exists a lower bound \(\delta\) for \(a_k\) and \(b_k\), and an upper bound of \(\rho/(4n)\) for \(x_k\), \(y_k\), \(u_k\), and \(v_k\).

The proof sketch of \Cref{lem-ab-bk-decay-slowly} is as follows. To establish lower bounds for \(a_k\) and \(b_k\), we first derive an iterative relationship between them (see \Cref{item-1-lemma-lb-ak-bk-relation-positive} in \Cref{con-lb-ak-bk-relation-positive}). Next, we prove simple upper bounds for \(a_k\) and \(b_k\), which serve as a foundation for the subsequent arguments (see \Cref{item-2-lemma-lb-ak-bk-relation-positive} in \Cref{con-lb-ak-bk-relation-positive}). Then we establish refined relationships between \(a_k\) and \(b_k\) that are crucial in precisely characterizing their decay rates (see \Cref{item-3-lemma-lb-ak-bk-relation-positive} in \Cref{con-lb-ak-bk-relation-positive}). As a result, we derive a lower bound on the decay rate of \(b_k\) (see \Cref{item-5-lemma-lb-ak-bk-relation-positive} in \Cref{con-lb-ak-bk-relation-positive}). Subsequently, we obtain an upper bound on the growth rates of \(x_k\), \(y_k\), \(u_k\), and \(v_k\), which turns out to be cubic in the lower bound of the decay rate of \(b_k\) (see \Cref{item-lemma-akbk-xkykukvk-relation-positive} in \Cref{con-lb-ak-bk-relation-positive}). This shows that the increases of \(x_k\), \(y_k\), \(u_k\), and \(v_k\) can be controlled by the decay of \(b_k\). Finally, we prove upper bounds for \(x_k\), \(y_k\), \(u_k\), and \(v_k\) (see \Cref{item-rho-lemma-lb-ak-bk-relation-positive} in \Cref{con-lb-ak-bk-relation-positive}). All these relationships and bounds are summarized in \Cref{con-lb-ak-bk-relation-positive}.

Define
\begin{align}
\ell &\triangleq 8n(\lceil\log n\rceil -1),\label{eq-def-ell-positive}\\
K &\triangleq \min\{k \geq 0\mid \min\{a_k,b_k\} < \delta \text{ or } k\geq 1/\delta \}\label{eq-def-K-positive}.
\end{align}

\begin{condition}\label{con-lb-ak-bk-relation-positive}
Given any $T< K$, consider the following conditions:
\begin{enumerate}
\item \label{item-1-lemma-lb-ak-bk-relation-positive}
For every odd $k$ with $k\leq T$, 
\begin{align}\label{eq-ak-recursion-positive}
a_{k+1} \in \left[\frac{a_k}{1+2a_k + \rho},\frac{a_k}{1+2a_k - \rho}\right],
\end{align}
\begin{equation}
\begin{aligned}\label{eq-bk-recursion-complex-positive}
b_{k+1} \in \left[\frac{b_k}{1 + \left(\frac{n}{2} - 1\right)(b_k-a_k)+ \rho}, \frac{b_k}{1 + \left(\frac{n}{2} - 1\right)(b_k-a_k)- \rho}\right];
\end{aligned}
\end{equation}
For every even $k$ with $k\leq T$,  
\begin{align}\label{eq-bk-recursion-positive}
b_{k+1} \in\left[ \frac{b_k}{1+2b_k+ \rho}, \frac{b_k}{1+2b_k- \rho}\right],
\end{align}
\begin{equation}
\begin{aligned}\label{eq-ak-recursion-complex-positive}
a_{k+1} \in \left[\frac{a_k}{1 + \left(\frac{n}{2} - 1\right)(a_k-b_k)+ \rho},\frac{a_k}{1 + \left(\frac{n}{2} - 1\right)(a_k-b_k)- \rho}\right].
\end{aligned}
\end{equation}

\vspace{0.2cm}

\item \label{item-2-lemma-lb-ak-bk-relation-positive}
If $k\leq T$, then
\begin{align}
\max\{a_k,b_k\} < 2/(n-1) - (1/\delta - (k-1))\rho < 2/(n-1)&< 1/10,\label{eq-ub-maxab-positive}\\
1+ \left(\frac{n}{2} - 1\right)(a_k - b_k) -\rho> 1- \left(\frac{n}{2} - 1\right)b_k -\rho >0&,\label{eq-lowerbound-one-minusnovertwobk}\\
1+ \left(\frac{n}{2} - 1\right)(b_k - a_k) -\rho> 1- \left(\frac{n}{2} - 1\right)a_k -\rho >0&.\label{eq-lowerbound-one-minusnovertwoak}
\end{align}
\item \label{item-3-lemma-lb-ak-bk-relation-positive}
For every odd $k$ with $k\leq T$,
\begin{align}
b_{k} &<  a_{k} + (6k+2)\rho b_{k},\label{eq-key-relation-positive-odd-1}\\
a_{k} &< b_{k}(1 + 2a_{k}) + 6(k+1)\rho b_{k},\label{eq-key-relation-positive-odd-2}\\
b_{k+1} &< a_{k} + (6k+4)\rho b_{k}.\label{eq-key-relation-positive-odd-3}
\end{align}
For every even $k$ with $k\leq T$,
\begin{align}
a_{k} &<  b_{k} + (6k+2)\rho a_{k},\label{eq-key-relation-positive-even-1}\\
b_{k} &< a_{k}(1 + 2b_{k}) + 6(k+1)\rho a_{k},\label{eq-key-relation-positive-even-2}\\
a_{k+1}  &< b_{k} + (6k+4)\rho a_{k}.\label{eq-key-relation-positive-even-3}
\end{align}

\item \label{item-5-lemma-lb-ak-bk-relation-positive}
For every even $k$ with $\ell+2 \leq k\leq T$,
\begin{align}
\frac{b_{k}}{b_{k-2}} &\leq \frac{1}{1+b_{k-2}}\label{eq-bk-leq-bkminustwo-oneplusbkminustwo}.
\end{align}

\item \label{item-lemma-akbk-xkykukvk-relation-positive}
For every even $k$ with $2\leq k\leq T$,
\begin{align}\label{eq-relation-maxratioofxyuv-ratiob}
\max\left\{\frac{x_{k}}{x_{k-1}},\frac{y_{k}}{y_{k-1}},\frac{u_{k}}{u_{k-1}},\frac{v_{k}}{v_{k-1}},\frac{x_{k-1}}{x_{k-2}},\frac{y_{k-1}}{y_{k-2}},\frac{u_{k-1}}{u_{k-2}},\frac{v_{k-1}}{v_{k-2}}\right\}\leq \left(1+b_{k-2}\right)^3.
\end{align}

\item \label{item-rho-lemma-lb-ak-bk-relation-positive} 
If $ k\leq \min\{T,\ell\}$, then
\begin{align}\label{eq-ub-xyuv-rho-kleql}
\max\{x_{k},y_{k},u_{k},v_{k}\} \leq \frac{\rho\delta^6}{n^{51}}\cdot \left(1+\frac{2}{n-1}\right)^{3k}.
\end{align}
For every odd $k$ with $\ell \leq k< T$,
\begin{align}\label{eq-ub-xyuv-rho-kgeql-even}
\max\{x_{k},y_{k},u_{k},v_{k}\} \leq 
\frac{\rho \delta^6}{n^2}\cdot \left(\frac{b_{\ell}}{b_{k-1}}\right)^6\cdot \left(\frac{b_{k-1}}{b_{k+1}}\right)^3.
\end{align}
For every even $k$ with $\ell \leq k\leq T$,
\begin{align}\label{eq-ub-xyuv-rho-kgeql-odd}
\max\{x_{k},y_{k},u_{k},v_{k}\} \leq \frac{\rho \delta^6}{n^2}\cdot \left(\frac{b_{\ell}}{b_{k}}\right)^6.
\end{align}
\end{enumerate}
\end{condition}

Next, we present the following lemma, which is central to the proof of \Cref{lem-ab-bk-decay-slowly}.
\begin{lemma}\label{lem-condition-holds}
\Cref{con-lb-ak-bk-relation-positive} holds for $T = K - 1$.
\end{lemma}

The following lemma is used in the proof of \Cref{lem-condition-holds}.
It establishes relationships between row sums and column sums. In particular, when \(k\) is even, \Cref{lem-condition-holds} allows us to determine the corresponding column sums based on the condition that the row sums equal one. Conversely, when \(k\) is odd, \Cref{lem-condition-holds} enables us to compute the row sums based on the condition that the column sums equal one.
The lemma is immediate by induction and its proof is provided in ~\Cref{appendix-sec-lowerbound}.

\begin{lemma}\label{lem-relation-rowsum-columnsum}
For any $k\geq 0$, we have
\begin{align}\label{eq-relation-rnovertwo-cnovertwo-positive}
r_{n/2}\left(A^{(k)}\right) = c_{n/2}\left(A^{(k)}\right) + \left(\frac{n}{2} - 1\right)(b_k + x_k - a_k - y_k).   
\end{align}
\begin{align}\label{eq-relation-rnovertwoplusone-cnovertwoplusone-positive}
r_{n/2+1}\left(A^{(k)}\right) = c_{n/2+1}\left(A^{(k)}\right) + \left(\frac{n}{2} - 1\right)(b_k + x_k - a_k - y_k).   
\end{align}
\begin{align}\label{eq-relation-rone-cn-positive}
\forall i\leq n/2- 1, \quad r_{i}\left(A^{(k)}\right) = c_{i}\left(A^{(k)}\right) + 2(a_k-x_k) + \left(\frac{n}{2} - 1\right)(u_k- v_k).
\end{align}
\begin{align}\label{eq-relation-rn-cone-positive}
\forall i\geq n/2+2, \quad c_{i}\left(A^{(k)}\right) = r_{i}\left(A^{(k)}\right) + 2(b_k-y_k) + \left(\frac{n}{2} - 1\right)(u_k- v_k).
\end{align}
\end{lemma}

Our proof of \Cref{lem-condition-holds} is by induction. 
The following lemma establishes the base case, and its proof can be found in~\Cref{appendix-sec-lowerbound}.

\begin{lemma}\label{lem-base-case} 
The following holds:
\begin{align}
&\max\{x_{0},y_{0},u_{0},v_{0}\} \leq \rho \delta^{6}/n^{51}.\label{eq-ub-xyuvzero}\\
&a_0 \in \left[ \frac{2}{n + 4 + \rho}, \frac{2}{n+ 4}\right], \quad b_0 \in \left[ \frac{2}{n + \rho}, \frac{2}{n}\right].\label{eq-a0-b0-range}\\
&a_0 < 2/(n-1) - \rho/\delta - \rho, \quad b_0 < 2/(n-1) - \rho/\delta -\rho \label{eq-a0-b0-less-twonminusone-minusnrho}.
\end{align}
\end{lemma}

Lemmas \ref{lem-item-max-xyuv-boost}-\ref{lem-last-item} are employed in the inductive step of \Cref{lem-condition-holds}. Let \(0 \leq t < K\) and assume that \Cref{item-rho-lemma-lb-ak-bk-relation-positive} of \Cref{con-lb-ak-bk-relation-positive} holds for \(T = \max\{t-2, 0\}\). Then, Lemmas \ref{lem-item-max-xyuv-boost}-\ref{lem-growth-rate-xkykukvk} establish Items \ref{item-1-lemma-lb-ak-bk-relation-positive}-\ref{item-lemma-akbk-xkykukvk-relation-positive} of \Cref{con-lb-ak-bk-relation-positive} for \(T = t\). Furthermore, if Items \ref{item-1-lemma-lb-ak-bk-relation-positive}-\ref{item-lemma-akbk-xkykukvk-relation-positive} hold for \(T = t\), then \Cref{lem-last-item} establishes \Cref{item-rho-lemma-lb-ak-bk-relation-positive} of \Cref{con-lb-ak-bk-relation-positive} for \(T = t\). Consequently, to prove \Cref{con-lb-ak-bk-relation-positive} for \(T = t\), it suffices to prove \Cref{item-rho-lemma-lb-ak-bk-relation-positive} for \(T = 0\), which follows immediately from \Cref{lem-base-case}. The proofs of Lemmas \ref{lem-item-max-xyuv-boost}-\ref{lem-last-item} can be found in~\Cref{appendix-sec-lowerbound}.

The following lemma shows that if \Cref{item-rho-lemma-lb-ak-bk-relation-positive} of \Cref{con-lb-ak-bk-relation-positive} holds, then \(x_{k}\), \(y_{k}\), \(u_{k}\), and \(v_{k}\) are \(O(\rho/n)\) and can be safely neglected in the calculations.

\begin{lemma}\label{lem-item-max-xyuv-boost}
Given any $0\leq t< K$,
if \Cref{item-rho-lemma-lb-ak-bk-relation-positive} of \Cref{con-lb-ak-bk-relation-positive} holds for $T= \max\{t-2,0\}$, then we have 
\begin{align}\label{eq-max-xyuv-less-rho-over-fourn}
\forall k\leq t, \quad \max\{x_{k},y_{k},u_{k},v_{k}\} < \frac{\rho}{4n}.
\end{align}
\end{lemma}

The following lemma establishes the iterative relationship between \(a_{k}\) and \(b_{k}\) based on \Cref{item-rho-lemma-lb-ak-bk-relation-positive} of \Cref{con-lb-ak-bk-relation-positive}. 

\begin{lemma}\label{lem-iterative-relation}
Given any $t< K$,
if \Cref{item-rho-lemma-lb-ak-bk-relation-positive} of \Cref{con-lb-ak-bk-relation-positive} holds for $T= \max\{t-2,0\}$, then \Cref{item-1-lemma-lb-ak-bk-relation-positive} of \Cref{con-lb-ak-bk-relation-positive} holds for $T= t$.
\end{lemma}

The following lemma establishes simple upper bounds for \(a_{k}\) and \(b_{k}\) based on \Cref{item-1-lemma-lb-ak-bk-relation-positive} of \Cref{con-lb-ak-bk-relation-positive}. These bounds will be used in the proofs of Items \ref{item-3-lemma-lb-ak-bk-relation-positive}-\ref{item-lemma-akbk-xkykukvk-relation-positive} of \Cref{con-lb-ak-bk-relation-positive}. 

\begin{lemma}\label{lem-upperbouns-akbk-simple}
Given any $t< K$,
if \Cref{item-1-lemma-lb-ak-bk-relation-positive} of \Cref{con-lb-ak-bk-relation-positive} holds for $T= t$, then \Cref{item-2-lemma-lb-ak-bk-relation-positive} of \Cref{con-lb-ak-bk-relation-positive} holds for $T= t$.
\end{lemma}

The following lemma establishes refined relationships between \(a_{k}\) and \(b_{k}\) based on Items \ref{item-1-lemma-lb-ak-bk-relation-positive}-\ref{item-2-lemma-lb-ak-bk-relation-positive} of \Cref{con-lb-ak-bk-relation-positive}. These relationships play a key role in characterizing the decay rates of \(a_{k}\) and \(b_{k}\).

\begin{lemma}\label{lem-refine-relation-akbk}
Given any $t< K$,
if Items \ref{item-1-lemma-lb-ak-bk-relation-positive} and \ref{item-2-lemma-lb-ak-bk-relation-positive} of \Cref{con-lb-ak-bk-relation-positive} hold for $T= t$, then \Cref{item-3-lemma-lb-ak-bk-relation-positive} of \Cref{con-lb-ak-bk-relation-positive} holds for $T= t$.
\end{lemma}

With Items \ref{item-1-lemma-lb-ak-bk-relation-positive}-\ref{item-3-lemma-lb-ak-bk-relation-positive} of \Cref{con-lb-ak-bk-relation-positive}, we derive a lower bound on the decay rate of \(b_k\).

\begin{lemma}\label{lem-decay-rate-bk}
Given any $t< K$,
if Items \ref{item-1-lemma-lb-ak-bk-relation-positive}-\ref{item-3-lemma-lb-ak-bk-relation-positive} 
of \Cref{con-lb-ak-bk-relation-positive} 
hold for $T= t$, then \Cref{item-5-lemma-lb-ak-bk-relation-positive} of \Cref{con-lb-ak-bk-relation-positive} holds for $T= t$.
\end{lemma}

Leveraging Items \ref{item-1-lemma-lb-ak-bk-relation-positive}-\ref{item-3-lemma-lb-ak-bk-relation-positive} and \ref{item-rho-lemma-lb-ak-bk-relation-positive} of \Cref{con-lb-ak-bk-relation-positive}, we establish an upper bound on the growth rates of \(x_k\), \(y_k\), \(u_k\), and \(v_k\), that is cubic in the lower bound of the decay rate of \(b_k\). Consequently, the increase of \(x_k\), \(y_k\), \(u_k\), and \(v_k\) can be bounded by the decay of \(b_k\).

\begin{lemma}\label{lem-growth-rate-xkykukvk}
Given any $t< K$,
if Items \ref{item-1-lemma-lb-ak-bk-relation-positive}-\ref{item-3-lemma-lb-ak-bk-relation-positive} 
of \Cref{con-lb-ak-bk-relation-positive} 
hold for $T= t$ and \Cref{item-rho-lemma-lb-ak-bk-relation-positive} of \Cref{con-lb-ak-bk-relation-positive} holds for $T= \max\{t-2,0\}$, then \Cref{item-lemma-akbk-xkykukvk-relation-positive} of \Cref{con-lb-ak-bk-relation-positive} holds for $T= t$.
\end{lemma}

Given that Items \ref{item-1-lemma-lb-ak-bk-relation-positive}-\ref{item-lemma-akbk-xkykukvk-relation-positive} of \Cref{con-lb-ak-bk-relation-positive} hold for \(T = t\), the following lemma establishes \Cref{item-rho-lemma-lb-ak-bk-relation-positive} for \(T = t\).

\begin{lemma}\label{lem-last-item}
Given any $t< K$,
if Items \ref{item-1-lemma-lb-ak-bk-relation-positive}-\ref{item-lemma-akbk-xkykukvk-relation-positive} 
of \Cref{con-lb-ak-bk-relation-positive} 
hold for $T= t$, then \Cref{item-rho-lemma-lb-ak-bk-relation-positive} of \Cref{con-lb-ak-bk-relation-positive} holds for $T= t$.
\end{lemma}

Now we can prove \Cref{lem-condition-holds}.
\begin{proof}[Proof of \Cref{lem-condition-holds}]
We prove this lemma by induction. 
For the base step,
by \eqref{eq-ub-xyuvzero}, we have \Cref{item-rho-lemma-lb-ak-bk-relation-positive} of \Cref{con-lb-ak-bk-relation-positive} holds for $T= 0$.
Combined with Lemmas \ref{lem-iterative-relation}-\ref{lem-growth-rate-xkykukvk}, we have Items \ref{item-1-lemma-lb-ak-bk-relation-positive}-\ref{item-lemma-akbk-xkykukvk-relation-positive} of \Cref{con-lb-ak-bk-relation-positive} for \(T = 0\).
In summary, we have \Cref{con-lb-ak-bk-relation-positive} holds for $T= 0$.

For the inductive step, assume that $0<t<K-1$ and
\Cref{con-lb-ak-bk-relation-positive} holds for all $T \leq t$.
Thus, we have \Cref{item-rho-lemma-lb-ak-bk-relation-positive} of \Cref{con-lb-ak-bk-relation-positive} holds for $T = t - 1$.
By Lemmas \ref{lem-iterative-relation}-\ref{lem-growth-rate-xkykukvk}, we have Items \ref{item-1-lemma-lb-ak-bk-relation-positive}-\ref{item-lemma-akbk-xkykukvk-relation-positive} of \Cref{con-lb-ak-bk-relation-positive} for \(T = t+1\).
Combined with \Cref{lem-last-item}, we also have \Cref{item-rho-lemma-lb-ak-bk-relation-positive} of \Cref{con-lb-ak-bk-relation-positive} holds for $T= t+1$.
In summary, we have \Cref{con-lb-ak-bk-relation-positive} holds for $T\leq t+1$.
This finishes the inductive step and the lemma is proved.
\end{proof}

By \Cref{lem-condition-holds}, we can prove \Cref{lem-ab-bk-decay-slowly}.
\begin{proof}[Proof of \Cref{lem-ab-bk-decay-slowly}]
By \eqref{eq-def-K-positive}, we have
\[\min\left\{a_{K -1},b_{K -1}\right\} \geq \delta.\]
Moreover, by Lemmas \ref{lem-condition-holds} and \ref{lem-item-max-xyuv-boost}, we have
\[\max\left\{x_{K-1},y_{K-1},u_{K-1},v_{K-1}\right\} < \frac{\rho}{4n}.\]
Thus, to prove this lemma, it is sufficient to prove $K = \Omega(1/\delta)$.
If $K \geq 1/\delta$, we have $K = \Omega(1/\delta)$ and then
the lemma is immediate.
In the following, we assume \emph{w.l.o.g.} $K < 1/\delta$ and prove $K = \Omega(1/\delta)$.

By $K < 1/\delta$, by \eqref{eq-def-K-positive} we have
\begin{align}\label{eq-minabigkbbigk-less-epsilon}
\min\{a_K,b_K\} < \delta.
\end{align}
In addition, by \eqref{eq-condition-delta}, \eqref{eq-def-rho-positive}, \eqref{eq-a0-b0-range} and $n>10^4$,
we have $a_0 \geq 500\delta$ and $b_0 \geq 500\delta$.
Combined with \eqref{eq-minabigkbbigk-less-epsilon}, we have $K \geq 1$.
Moreover, one can verify that the matrix $A$ in \Cref{thm-lowerbound-positive} is $(\lceil\gamma n\rceil/n, 2\lceil\gamma n\rceil/n)$-dense.
By \Cref{item-forth-fact-c-a-a0-a1} of \Cref{fact-c-a-a0-a1} and $\gamma\in (1/4,1/2)$, we have
\[\forall i\in [n],\quad c_{i}\left(A^{(K-1)}\right) \geq 2\gamma^2 \geq 1/8, \quad r_{i}\left(A^{(K-1)}\right)\geq 2\gamma^2\geq 1/8.\]
Combined with \eqref{eq-minabigkbbigk-less-epsilon}, we have 
\begin{align}\label{eq-min-akminusone-bkminusone-upperbound}
\min\{a_{K-1},b_{K-1}\} \leq \min\{a_K,b_K\} \cdot \max_{i\in [n]}\left\{\max\left\{c^{-1}_{i}\left(A^{(K-1)}\right), r^{-1}_{i}\left(A^{(K-1)}\right)\right\}\right\} \leq 8\delta.
\end{align}
If $K-1$ is odd, by \eqref{eq-key-relation-positive-odd-1}, \eqref{eq-def-rho-positive} and $K < 1/\delta$, we have 
\begin{align}\label{eq-bkminusone-less-two-akminusone}
b_{K-1} \leq \frac{a_{K-1}}{1-(6K-4)\rho} \leq 2a_{K-1}.
\end{align}
Combined with \eqref{eq-min-akminusone-bkminusone-upperbound}, we have
\begin{align}\label{eq-bkminusone-less-sixteendelta}
b_{K-1} \leq 16\delta.
\end{align}
If $K-1$ is even,
by \eqref{eq-key-relation-positive-even-2}, \eqref{eq-ub-maxab-positive}, \eqref{eq-def-rho-positive}, and $K < 1/\delta$, we also have
\[b_{K-1} \leq a_{K-1}(1+2b_{K-1}+6K\rho) \leq 2a_{K-1}.\]
Combined with \eqref{eq-min-akminusone-bkminusone-upperbound}, we also have 
\eqref{eq-bkminusone-less-sixteendelta}.

Given any even $2\leq k< K$, by \eqref{eq-def-K-positive} we have $b_{k-1}\geq \delta$, $k < 1/\delta$.
Combined with \eqref{eq-def-rho-positive} and $b_{k-1}\leq 1$, we have
\begin{align}\label{eq-bkminustwo-klessnovervarepsilon-rho-bk-decay}
\rho \leq  \frac{b_{k-1}}{100kn^{2}} \leq \frac{1}{100kn^{2}}.
\end{align}
Thus, for any even $2\leq k<K$, we have
\begin{equation*}
\begin{aligned}
&\quad b_{k}\\
(\text{by \eqref{eq-bk-recursion-complex-positive}})\quad & \geq \frac{b_{k-1}}{1 + \left(\frac{n}{2} - 1\right)(b_{k-1}-a_{k-1})+ \rho}\\
(\text{by \eqref{eq-lowerbound-one-minusnovertwoak} and \eqref{eq-key-relation-positive-odd-1}})\quad & \geq \frac{b_{k-1}}{1 + \left(\frac{n}{2} - 1\right)(6k-4)\rho b_{k-1}+ \rho}\\
(\text{by \eqref{eq-bkminustwo-klessnovervarepsilon-rho-bk-decay} })\quad &\geq \frac{b_{k-1}}{(1+b_{k-1}/n)(1+b_{k-1}/n^2)}\\
(\text{by \eqref{eq-ub-maxab-positive} and $n>10^4$})\quad &\geq b_{k-1}(1-b_{k-1})\\
(\text{by \eqref{eq-ub-maxab-positive}})\quad &> \frac{b_{k-1}}{2}.
\end{aligned}
\end{equation*}
Similarly, given any odd $1\leq k< K$, we also have \eqref{eq-bkminustwo-klessnovervarepsilon-rho-bk-decay}.
Hence,
\begin{equation*}
\begin{aligned}
&\quad b_{k}\\
(\text{by \eqref{eq-bk-recursion-positive}})\quad & \geq \frac{b_{k-1}}{1+2b_{k-1}+ \rho}\\
(\text{by \eqref{eq-bkminustwo-klessnovervarepsilon-rho-bk-decay}})\quad &\geq \frac{b_{k-1}}{(1+2b_{k-1})(1+b_{k-1}/n^2)}\\
(\text{by \eqref{eq-ub-maxab-positive} and $n>10^4$})\quad &\geq b_{k-1}(1-3b_{k-1})\\
(\text{by \eqref{eq-ub-maxab-positive}})\quad &> \frac{b_{k-1}}{2}.
\end{aligned}
\end{equation*}
In summary, for each $1\leq k<K$, we always have 
\begin{align}\label{eq-bk-recursion-one-step}
b_k\geq b_{k-1}(1-3b_{k-1})>\frac{b_{k-1}}{2}.
\end{align}
Combined with $b_{K-1}\leq 16\delta$ and $b_0 \geq 500\delta$, we have
there exists some $r$ with $0\leq r<K$ such that 
\begin{align}\label{eq-deltar-regime}
32\delta < b_{r} \leq 64\delta.
\end{align}
We claim that for each $0\leq k < K-r$, 
\begin{align}\label{eq-recursion-bk}
b_{r+k} \geq b_r(1 - 3b_r)^{k}.
\end{align}
Combining \eqref{eq-recursion-bk} with \eqref{eq-deltar-regime}, for each $0\leq k < K-r$ we have
\begin{align}\label{eq-bkplusr-lowerbound}
b_{r+k} > 32\delta(1 - 192\delta)^{k}.
\end{align}
Hence, we have 
\begin{align}\label{eq-bigkminusr-lowerbound}
K - r > \frac{\log 2}{\log (1/(1-192\delta))} = \Omega(1/\delta).
\end{align}
Otherwise, we have 
\[K - r \leq \frac{\log 2}{\log (1/(1-192\delta))} .\]
Combined with \eqref{eq-bkplusr-lowerbound}, we have
$b_{K-1} > 16\delta$, which is contradictory with \eqref{eq-bkminusone-less-sixteendelta}.
By \eqref{eq-bigkminusr-lowerbound}, we have
$K = \Omega(1/\delta )$.
Thus, to complete the proof of the lemma, it remains only to establish \eqref{eq-recursion-bk}.

Finally, we prove \eqref{eq-recursion-bk} by induction.
For the base step when $k = 0$, we have $b_{r} \geq b_r(1 - 3b_r)$ immediately.
For the inductive step when $0< k < K-r$, we have
\begin{align*}
&\quad b_{r+k} \\
(\text{by \eqref{eq-bk-recursion-one-step}})\quad &\geq b_{r+k-1}(1 - 3b_{r+k-1}) \\
(\text{by \eqref{eq-ub-maxab-positive} and the inductive assumption})\quad &\geq b_r(1 - 3b_r)^{k-1}\left(1 - 3b_r(1 - 3b_r)^{k-1}\right)\\
(\text{by \eqref{eq-ub-maxab-positive}})\quad &\geq b_r(1 - 3b_r)^{k},
\end{align*}
which finishes the induction.
Hence, \eqref{eq-recursion-bk} is proved and the lemma is immediate.
\end{proof}

Finally, we can prove \Cref{thm-lowerbound-positive}.
\begin{proof}[Proof of \Cref{thm-lowerbound-positive}]
It is easy to verify that \(A\) is a positive matrix with density \(\gamma\).  
Meanwhile, we have $\perman(A) \geq 1$, because $A_{i,i} = 1$ for each $i\in [n]$.
Together with \Cref{lem-condition-converge}, 
we have the Sinkhorn-Knopp algorithm converges on the matrix \(A\).

In the following, we prove that the Sinkhorn-Knopp algorithm converges slowly on the matrix $A$. 
Let $\delta \triangleq \varepsilon/n$.
Recall that $\varepsilon < 1/1000$.
Then one can verify that \eqref{eq-condition-delta} is satisfied.
By \Cref{lem-ab-bk-decay-slowly}, 
we have there exists some 
$t = \Omega(n/\varepsilon )$ such that 
$\min\{a_{k},b_{k}\} \geq \delta$ and $\max\{x_{k},y_{k},u_{k},v_{k}\} < \rho/(4n)$
for each $k \leq t$.
Combined with $n > 10^4$, $\delta = \varepsilon/n$ and $\rho = \delta^2/(100n^{2})$,
we have for each $k \leq t$,
\[\left(\frac{n}{2} - 1\right)\cdot\left(\left(2a_k-n\cdot \max\{x_k,y_k,u_k,v_k\}\right)+\left(2b_k-n\cdot \max\{x_k,y_k,u_k,v_k\}\right)\right)\geq (n-2)(2\delta -  \rho) \geq \varepsilon .\]
Combined with \Cref{lem-relation-rowsum-columnsum}, we have
\begin{align*}
&\quad \norm{r\left(A^{(k)}\right)- \boldsymbol{1}}_1 + \norm{c\left(A^{(k)}\right)- \boldsymbol{1}}_1 = \sum_{i\in [n]}\abs{r_{i}\left(A^{(k)}\right) - c_{i}\left(A^{(k)}\right)} 
\\&= \left(\frac{n}{2} - 1\right)\left(2\abs{b_k + x_k - a_k - y_k}+ \abs{2(a_k-x_k) + \left(\frac{n}{2} - 1\right)(u_k- v_k)} + \abs{2(b_k-y_k) + \left(\frac{n}{2} - 1\right)(u_k- v_k)}\right)\\
&\geq \left(\frac{n}{2} - 1\right)\cdot\left(\left(2a_k-n\cdot \max\{x_k,y_k,u_k,v_k\}\right)+\left(2b_k-n\cdot \max\{x_k,y_k,u_k,v_k\}\right)\right) \geq \varepsilon .
\end{align*}
Similarly, 
let $\delta \triangleq \varepsilon/\sqrt{n}$.
If $\varepsilon \leq 10^{-3}\cdot n^{-1/2}$,
then one can verify that \eqref{eq-condition-delta} is satisfied.
By \Cref{lem-ab-bk-decay-slowly}, 
we have there exists some 
$t = \Omega(\sqrt{n}/\varepsilon )$ such that 
$\min\{a_{k},b_{k}\} \geq \delta$ and $\max\{x_{k},y_{k},u_{k},v_{k}\} < \rho/(4n)$
for each $k \leq t$.
Combined with $n > 10^4$, $\delta = \varepsilon/\sqrt{n}$ and $\rho = \delta^2/(100n^{2})$,
we have for each $k \leq t$,
\[\left(\frac{n}{2} - 1\right)\cdot \left(\left(2a_k-n\cdot \max\{x_k,y_k,u_k,v_k\}\right)^2 + \left(\left(2b_k-n\cdot \max\{x_k,y_k,u_k,v_k\}\right)\right)^2\right) \geq \varepsilon^2.\]
Combined with \Cref{lem-relation-rowsum-columnsum}, we have
\begin{align*}
&\quad \norm{r\left(A^{(k)}\right)- \boldsymbol{1}}_2 + \norm{c\left(A^{(k)}\right)- \boldsymbol{1}}_2 = \sqrt{\sum_{i\in [n]}\abs{r_{i}\left(A^{(k)}\right) - c_{i}\left(A^{(k)}\right)}^2} \\
&\geq \left(\frac{n}{2} - 1\right)^{1/2}\cdot \left(\left(2(a_k-x_k) + \left(\frac{n}{2} - 1\right)(u_k- v_k)\right)^2 + \left(2(b_k-y_k) + \left(\frac{n}{2} - 1\right)(u_k- v_k)\right)^2\right)^{1/2}\\
&\geq \left(\frac{n}{2} - 1\right)^{1/2}\cdot \left(\left(2a_k-n\cdot \max\{x_k,y_k,u_k,v_k\}\right)^2 + \left(\left(2b_k-n\cdot \max\{x_k,y_k,u_k,v_k\}\right)\right)^2\right)^{1/2} \geq \varepsilon.
\end{align*}
Hence, the theorem is immediate.
\end{proof}

\subsection{Generalized Construction for Each $\gamma \in [0, 1/2)$.}
In this section, we prove \Cref{thm-lower-bound-main}. For each \(\gamma \in [0, 1/2)\), we construct a matrix of density \(\gamma\) for which the Sinkhorn-Knopp algorithm converges slowly. \Cref{thm-lowerbound-positive} provides the construction for \(\gamma \in (1/4, 1/2)\), while the following theorem presents the construction for \(\gamma \in [0, 1/4]\).

\begin{theorem}\label{thm-lower-bound-construction}
Let $\gamma\in [0,1/4]$ and $\varepsilon <1/1000$. 
If \(\gamma > 0\), let \(m > 20000\) be an integer chosen so that \(n = m - \lceil \gamma m \rceil\) is even, define \(Z\) as the \(n \times n\) matrix with density \(1/2 - 1/n\) as specified in \Cref{thm-lowerbound-positive}, and let \(B\) be the \(\lceil \gamma m \rceil \times \lceil \gamma m \rceil\) matrix with every entry equal to 1. Otherwise, let \(m > 20000\) be an odd integer, set \(n = m - 1\), define \(Z\) as the \(n \times n\) matrix with density \(1/2 - 1/n\) as specified in \Cref{thm-lowerbound-positive}, and let \(B\) be the \(1 \times 1\) matrix whose sole entry is \(1/m\).
Let \(A\) be an \(m \times m\) block diagonal matrix with two blocks, where the matrix in the top-left corner is \(Z\) and the matrix in the bottom-right corner is \(B\).
Then we have $A$ is a matrix with density $\gamma$, and
the Sinkhorn-Knopp algorithm converges on the matrix $A$.
Moreover, there exists some 
$t = \Omega(m/\varepsilon )$ such that for each $k \leq t$,
\[\norm{r\left(A^{(k)}\right)- \boldsymbol{1}}_1 + \norm{c\left(A^{(k)}\right)- \boldsymbol{1}}_1 \geq \varepsilon.\]
If $\varepsilon \leq 1/(2000\sqrt{m})$, there also exists some 
$t = \Omega(\sqrt{m}/\varepsilon )$ such that for each $k \leq t$,
\[\norm{r\left(A^{(k)}\right)- \boldsymbol{1}}_2 + \norm{c\left(A^{(k)}\right)- \boldsymbol{1}}_2 \geq \varepsilon.\]
\end{theorem}
\begin{proof}
We prove the theorem for the case \(\gamma > 0\); the proof for \(\gamma = 0\) is analogous.
Given any \(\gamma \in (0,1/4)\), one can verify the existence of an integer \(m > 20000\) such that \(n = m - \lceil \gamma m \rceil\) is even.
By $m>20000$ and $\gamma\in (0,1/4]$,
we have $n = m - \lceil\gamma m\rceil > 10000$.
By $n$ is even, we have
\[n - 2\left\lceil \left(\frac{1}{2} - \frac{1}{n}\right)\cdot n\right\rceil = 2.\]
In addition, we also have $1/2 - 1/n \in (1/4,1/2)$.
Thus, the size $n$ and the density $1/2 - 1/n$
satisfy the conditions in \Cref{thm-lowerbound-positive} and 
the matrix $Z$ is well-defined.
Furthermore, we have
\begin{align*}
&\quad \left\lceil \left(\frac{1}{2} - \frac{1}{n}\right)\cdot n\right\rceil \\
(\text{by $n>10000$})\quad &\geq \frac{4n}{9} \\
(\text{by $n = m - \lceil\gamma m\rceil$}) \quad &= \frac{4(m - \lceil\gamma m\rceil)}{9} \\
(\text{by $m>20000$ and $\gamma \in (0,1/4]$})\quad &\geq 
\frac{4}{9}\cdot \frac{2m}{3}\\
& = \frac{8m}{27}\\
(\text{by $m>20000$ and $\gamma \in (0,1/4]$})\quad & \geq \lceil \gamma m  \rceil.
\end{align*}
Hence, each row and each column of $Z$ contains at least $\lceil \gamma m  \rceil$ entries equal to 1.
Therefore, it is straightforward to verify that \(A\) is \(\gamma\)-dense.

One can verify that the Sinkhorn-Knopp algorithm converges on the matrix $B$ in one round.
In addition, by applying \Cref{thm-lowerbound-positive} to the matrix $Z$, we have the Sinkhorn-Knopp algorithm converges on $Z$, and 
there exists some 
$t = \Omega(n/\varepsilon) = \Omega(m/\varepsilon )$ such that for each $k \leq t$,
\[\norm{r\left(Z^{(k)}\right)- \boldsymbol{1}}_1 + \norm{c\left(Z^{(k)}\right)- \boldsymbol{1}}_1 \geq \varepsilon.\]
Thus, the Sinkhorn-Knopp algorithm converges on $A$,
and there exists some 
$t = \Omega(m/\varepsilon )$ such that for each $k \leq t$,
\[\norm{r\left(A^{(k)}\right)- \boldsymbol{1}}_1 + \norm{c\left(A^{(k)}\right)- \boldsymbol{1}}_1 = \norm{r\left(Z^{(k)}\right)- \boldsymbol{1}}_1 + \norm{c\left(Z^{(k)}\right)- \boldsymbol{1}}_1 \geq \varepsilon.\]
Similarly, one can also verify the result with the \(\ell_2\)-norm.
\end{proof}

\Cref{thm-lower-bound-main} is immediate by Theorems \ref{thm-lowerbound-positive} and \ref{thm-lower-bound-construction}.



\section*{Acknowledgement}
We thank Hu Ding, Qingyuan Li, Haoran Wang and Hongteng Xu for helpful discussion. We also thank the anonymous reviewers for their valuable suggestions.

\bibliographystyle{alpha}
\bibliography{refs}

\begin{thebibliography}{AZLOW17}

\bibitem[ANWR17]{altschuler2017near}
Jason Altschuler, Jonathan Niles-Weed, and Philippe Rigollet.
\newblock Near-linear time approximation algorithms for optimal transport via sinkhorn iteration.
\newblock {\em Advances in neural information processing systems}, 30, 2017.

\bibitem[AZLOW17]{allen2017much}
Zeyuan Allen-Zhu, Yuanzhi Li, Rafael Oliveira, and Avi Wigderson.
\newblock Much faster algorithms for matrix scaling.
\newblock In {\em 2017 IEEE 58th Annual Symposium on Foundations of Computer Science (FOCS)}, pages 890--901. IEEE, 2017.

\bibitem[Bac65]{bacharach1965estimating}
Michael Bacharach.
\newblock Estimating nonnegative matrices from marginal data.
\newblock {\em International Economic Review}, 6(3):294--310, 1965.

\bibitem[CK21]{chakrabarty2021better}
Deeparnab Chakrabarty and Sanjeev Khanna.
\newblock Better and simpler error analysis of the sinkhorn--knopp algorithm for matrix scaling.
\newblock {\em Mathematical Programming}, 188(1):395--407, 2021.

\bibitem[CMTV17]{cohen2017matrix}
Michael~B Cohen, Aleksander Madry, Dimitris Tsipras, and Adrian Vladu.
\newblock Matrix scaling and balancing via box constrained newton's method and interior point methods.
\newblock In {\em 2017 IEEE 58th Annual Symposium on Foundations of Computer Science (FOCS)}, pages 902--913. IEEE, 2017.

\bibitem[Cut13]{cuturi2013sinkhorn}
Marco Cuturi.
\newblock Sinkhorn distances: Lightspeed computation of optimal transport.
\newblock {\em Advances in neural information processing systems}, 26, 2013.

\bibitem[DKPU22]{dufosse2022scaling}
Fanny Dufoss{\'e}, Kamer Kaya, Ioannis Panagiotas, and Bora U{\c{c}}ar.
\newblock Scaling matrices and counting the perfect matchings in graphs.
\newblock {\em Discrete Applied Mathematics}, 308:130--146, 2022.

\bibitem[DS40]{deming1940least}
W~Edwards Deming and Frederick~F Stephan.
\newblock On a least squares adjustment of a sampled frequency table when the expected marginal totals are known.
\newblock {\em The Annals of Mathematical Statistics}, 11(4):427--444, 1940.

\bibitem[Ego81]{egorychev1981solution}
G.~P. Egorychev.
\newblock The solution of van der waerden’s problem for permanents.
\newblock {\em Advances in Mathematics}, 42:299--305, 1981.

\bibitem[Fal81]{falikman1981proof}
D.~I. Falikman.
\newblock Proof of the van der waerden’s conjecture on the permanent of a doubly stochastic matrix.
\newblock {\em Matematicheskie Zametki}, 29(6):931--938, 1981.

\bibitem[HJ48]{hall1948distinct}
Marshall Hall~Jr.
\newblock Distinct representatives of subsets.
\newblock {\em Bulletin of the American Mathematical Society}, 54(10):922--926, 1948.

\bibitem[HJ12]{horn2012matrix}
R.~A. Horn and C.~R. Johnson.
\newblock {\em Matrix Analysis}.
\newblock Cambridge University Press, 2nd edition, 2012.

\bibitem[HL08]{huber2008fast}
Mark Huber and Jenny Law.
\newblock Fast approximation of the permanent for very dense problems.
\newblock In {\em Proceedings of the nineteenth annual ACM-SIAM symposium on Discrete algorithms}, pages 681--689, 2008.

\bibitem[Ide16]{idel2016review}
Martin Idel.
\newblock A review of matrix scaling and sinkhorn's normal form for matrices and positive maps.
\newblock {\em arXiv preprint arXiv:1609.06349}, 2016.

\bibitem[KK93]{kalantari1993rate}
Bahman Kalantari and Leonid Khachiyan.
\newblock On the rate of convergence of deterministic and randomized ras matrix scaling algorithms.
\newblock {\em Operations research letters}, 14(5):237--244, 1993.

\bibitem[KK96]{kalantari1996complexity}
Bahman Kalantari and Leonid Khachiyan.
\newblock On the complexity of nonnegative-matrix scaling.
\newblock {\em Linear Algebra and its applications}, 240:87--103, 1996.

\bibitem[KLLR18]{kwok2018paulsen}
Tsz~Chiu Kwok, Lap~Chi Lau, Yin~Tat Lee, and Akshay Ramachandran.
\newblock The paulsen problem, continuous operator scaling, and smoothed analysis.
\newblock In {\em Proceedings of the 50th annual ACM SIGACT symposium on theory of computing}, pages 182--189, 2018.

\bibitem[KLRS08]{kalantari2008complexity}
Bahman Kalantari, Isabella Lari, Federica Ricca, and Bruno Simeone.
\newblock On the complexity of general matrix scaling and entropy minimization via the ras algorithm.
\newblock {\em Mathematical Programming}, 112(2):371--401, 2008.

\bibitem[LSW98]{linial1998deterministic}
Nathan Linial, Alex Samorodnitsky, and Avi Wigderson.
\newblock A deterministic strongly polynomial algorithm for matrix scaling and approximate permanents.
\newblock In {\em Proceedings of the thirtieth annual ACM symposium on Theory of computing}, pages 644--652, 1998.

\bibitem[NR99]{nemirovski1999complexity}
Arkadi Nemirovski and Uriel Rothblum.
\newblock On complexity of matrix scaling.
\newblock {\em Linear Algebra and its Applications}, 302:435--460, 1999.

\bibitem[Osb60]{osborne1960pre}
EE~Osborne.
\newblock On pre-conditioning of matrices.
\newblock {\em Journal of the ACM (JACM)}, 7(4):338--345, 1960.

\bibitem[RTG00]{rubner2000earth}
Yossi Rubner, Carlo Tomasi, and Leonidas~J Guibas.
\newblock The earth mover's distance as a metric for image retrieval.
\newblock {\em International journal of computer vision}, 40:99--121, 2000.

\bibitem[R{\"u}s95]{ruschendorf1995convergence}
Ludger R{\"u}schendorf.
\newblock Convergence of the iterative proportional fitting procedure.
\newblock {\em The Annals of Statistics}, pages 1160--1174, 1995.

\bibitem[RZ07]{rote2007matrix}
G{\"u}nter Rote and Martin~Tvede Zachariasen.
\newblock Matrix scaling by network flow.
\newblock In {\em Proceedings of the eighteenth annual ACM-SIAM symposium on Discrete algorithms:(SODA 07)}, pages 848--854. Association for Computing Machinery, 2007.

\bibitem[Sin64]{sinkhorn1964relationship}
Richard Sinkhorn.
\newblock A relationship between arbitrary positive matrices and doubly stochastic matrices.
\newblock {\em The annals of mathematical statistics}, 35(2):876--879, 1964.

\bibitem[Sin67]{sinkhorn1967diagonal}
Richard Sinkhorn.
\newblock Diagonal equivalence to matrices with prescribed row and column sums.
\newblock {\em The American Mathematical Monthly}, 74(4):402--405, 1967.

\bibitem[SK67]{sinkhorn1967concerning}
Richard Sinkhorn and Paul Knopp.
\newblock Concerning nonnegative matrices and doubly stochastic matrices.
\newblock {\em Pacific Journal of Mathematics}, 21(2):343--348, 1967.

\bibitem[Val79]{valiant1979complexity}
Leslie~G Valiant.
\newblock The complexity of computing the permanent.
\newblock {\em Theoretical computer science}, 8(2):189--201, 1979.

\end{thebibliography}

\clearpage


\appendix
\section{Proof of Fact 2.5}\label{appendix-fact}
\Cref{item-first-fact-a0-a1} of this lemma is immediate.
\Cref{item-forth-fact-c-a-a0-a1} of this lemma is immediate by Item \ref{item-second-fact-c-a-a0-a1} and \Cref{lem-monotone-rsum-csum}.
Here, we only prove \Cref{item-second-fact-c-a-a0-a1}.
Since $A$ is $(\gamma,\rho)$-dense,
we have
\begin{align}\label{eq-lower-upper-bound-r-c}
\forall i,j\in [n], \quad \rho\lceil\gamma n\rceil \leq  r_i(A),c_i(A) \leq n.
\end{align}
Thus 
\begin{align*}
\forall i,j\in [n], \quad 0 \leq A^{(0)}_{i,j} = \frac{A_{i,j}}{r_i(A)} \leq \frac{1}{\rho\lceil\gamma n\rceil} \leq \frac{1}{\rho\gamma n}.
\end{align*}
In addition, by $r_i(A)\leq n$ for each $i\in [n]$,
we have $A^{(0)}_{i,j}\geq \rho/n$ if $A_{i,j} \geq \rho$.
Combined with the fact that $A$ is $(\gamma,\rho)$-dense, we have that each row and each column of $A^{(0)}$ contains at least $\gamma n$ entries no less than $\rho/n$.

\section{Missing proofs in Section 4}\label{appendix-sec-lowerbound}
\subsection*{Proof of \Cref{lemma-lb-equal-item-density-positive}}
We prove this lemma by reduction.
For the base step where $k=0$, all the conclusions are immediate.
For the inductive step, assume \emph{w.l.o.g.} $k$ is odd.

(\ref{item-1-lemma-positive})
For each $j,j'\leq n/2-1$ and each $k\geq 1$,
by the inductive assumption, we have 
\begin{equation}\label{eq-sum-atj-induction}
\begin{aligned}
\sum_{t\leq n/2-1}A^{(k-1)}_{t,j} &= \sum_{t\leq n/2-1}\left(A^{(k-1)}_{t,j}\cdot \id{A_{t,j} = 1} + A^{(k-1)}_{t,j}\cdot \id{A_{t,j} = \frac{2\lceil \gamma n \rceil}{n}} + A^{(k-1)}_{t,j}\cdot \id{A_{t,j} = \beta} \right)\\
& = \sum_{t\leq n/2-1}\left(A^{(k-1)}_{t,j'}\cdot \id{A_{t,j'} = 1} + A^{(k-1)}_{t,j'}\cdot \id{A_{t,j'} = \frac{2\lceil \gamma n \rceil}{n}} + A^{(k-1)}_{t,j'}\cdot \id{A_{t,j'} = \beta} \right) 
\\&= \sum_{t\leq n/2-1}A^{(k-1)}_{t,j'}.
\end{aligned}
\end{equation}
Thus, for each $i,i',j,j'\leq n/2-1$ where $A_{i,j} = A_{i',j'}$ and each $k\geq 1$,
\begin{align*}
A^{(k)}_{i,j} &= \frac{A^{(k-1)}_{i,j}}{c_{j}\left(A^{(k-1)}\right)} = \frac{A^{(k-1)}_{i,j}}{A^{(k-1)}_{n/2,j} + A^{(k-1)}_{n/2+1,j} + \left(\sum_{t\leq n/2-1}A^{(k-1)}_{t,j}\right)+\left(\sum_{t\geq n/2+2}A^{(k-1)}_{t,j}\right)}
\\&= \frac{A^{(k-1)}_{i',j'}}{A^{(k-1)}_{n/2,j'} + A^{(k-1)}_{n/2+1,j'} + \left(\sum_{t\leq n/2-1}A^{(k-1)}_{t,j'}\right)+\left(\sum_{t\geq n/2+2}A^{(k-1)}_{t,j'}\right)} = \frac{A^{(k-1)}_{i',j'}}{c_{j'}\left(A^{(k-1)}\right)} = A^{(k)}_{i',j'},
\end{align*}
where the third equality follows from the inductive assumption and \eqref{eq-sum-atj-induction}.

The proofs of (\ref{item-2-lemma-positive}), (\ref{item-3-lemma-positive}), and (\ref{item-4-lemma-positive}) are similar to that of (\ref{item-1-lemma-positive}).

(\ref{item-5-lemma-positive}) 
For each $i,j, \ell\leq n/2-1$ we have
\[
A^{(k)}_{i,n/2} = \frac{A^{(k-1)}_{i,n/2}}{c_{n/2}\left(A^{(k-1)}\right)} = \frac{A^{(k-1)}_{j,n/2}}{c_{n/2}\left(A^{(k-1)}\right)} = A^{(k)}_{j,n/2},
\]
where the second equality is by the inductive assumption. Similarly,  
\begin{align*}
A^{(k)}_{i,n/2} &= \frac{A^{(k-1)}_{i,n/2}}{c_{n/2}\left(A^{(k-1)}\right)} = \frac{A^{(k-1)}_{i,n/2}}{A^{(k-1)}_{n/2,n/2}+A^{(k-1)}_{n/2+1,n/2}+\sum_{t\leq n/2-1}A^{(k-1)}_{t,n/2} + \sum_{t\geq n/2+2}A^{(k-1)}_{t,n/2}} 
\\&= \frac{A^{(k-1)}_{\ell,n/2+1}}{A^{(k-1)}_{n/2+1,n/2+1}+A^{(k-1)}_{n/2,n/2+1}+\sum_{t\leq n/2-1}A^{(k-1)}_{t,n/2+1}+\sum_{t\geq n/2+2}A^{(k-1)}_{t,n/2+1}} = \frac{A^{(k-1)}_{\ell,n/2+1}}{c_{n/2+1}\left(A^{(k-1)}\right)} = A^{(k)}_{\ell,n/2+1},
\end{align*}
where the third equality is by the inductive assumption.
Thus, \eqref{eq-equal-item-1} is proved. Similarly, one can also prove \eqref{eq-equal-item-4}.

For each $n/2+2\leq i,j,\ell \leq n$, we have 
\begin{align*}
A^{(k)}_{n/2,i} &= \frac{A^{(k-1)}_{n/2,i}}{c_{i}\left(A^{(k-1)}\right)} = \frac{A^{(k-1)}_{n/2,i}}{A^{(k-1)}_{n/2,i}+A^{(k-1)}_{n/2+1,i} + \sum_{t\leq n/2 - 1}A^{(k-1)}_{t,i} + \sum_{n/2+2\leq t \leq n}A^{(k-1)}_{t,i}} 
\\&= \frac{A^{(k-1)}_{n/2,j}}{A^{(k-1)}_{n/2,j} + A^{(k-1)}_{n/2+1,j} + \sum_{t \leq n/2-1}A^{(k-1)}_{t,j} + \sum_{n/2+2\leq t \leq n}A^{(k-1)}_{t,j}} 
= \frac{A^{(k-1)}_{n/2,j}}{c_{j}\left(A^{(k-1)}\right)} = A^{(k)}_{n/2,j},
\end{align*}
where the third equality is by the inductive assumption. Similarly,  
\begin{align*}
A^{(k)}_{n/2,i} &= \frac{A^{(k-1)}_{n/2,i}}{c_{i}\left(A^{(k-1)}\right)} = \frac{A^{(k-1)}_{n/2,i}}{A^{(k-1)}_{n/2,i}+A^{(k-1)}_{n/2+1,i} +  \sum_{t\leq n/2 - 1}A^{(k-1)}_{t,i} + \sum_{n/2+2\leq t \leq n}A^{(k-1)}_{t,i}} 
\\&= \frac{A^{(k-1)}_{n/2+1,\ell}}{A^{(k-1)}_{n/2,\ell} + A^{(k-1)}_{n/2+1,\ell} +  \sum_{t\leq n/2 - 1}A^{(k-1)}_{t,\ell}+ \sum_{n/2+2\leq t \leq n}A^{(k-1)}_{t,\ell}} = \frac{A^{(k-1)}_{n/2+1,\ell}}{c_{\ell}\left(A^{(k-1)}\right)} = A^{(k)}_{n/2+1,\ell},
\end{align*}
where the third equality is by the inductive assumption.
Thus, \eqref{eq-equal-item-3} is proved. 
Similarly, one can also prove \eqref{eq-equal-item-2}.

Moreover, 
\begin{align*}
&\quad A^{(k)}_{n/2,n/2} = \frac{A^{(k-1)}_{n/2,n/2}}{c_{n/2}\left(A^{(k-1)}\right)} = \frac{A^{(k-1)}_{n/2,n/2}}{A^{(k-1)}_{n/2,n/2}+A^{(k-1)}_{n/2+1,n/2} + \sum_{t\leq n/2 - 1}A^{(k-1)}_{t,n/2} + \sum_{n/2+2\leq t \leq n}A^{(k-1)}_{t,n/2}} 
\\&= \frac{A^{(k-1)}_{n/2+1,n/2+1}}{A^{(k-1)}_{n/2+1,n/2+1} + A^{(k-1)}_{n/2,n/2+1} + \sum_{t \leq n/2-1}A^{(k-1)}_{t,n/2+1} + \sum_{n/2+2\leq t \leq n}A^{(k-1)}_{t,n/2+1}} 
= \frac{A^{(k-1)}_{n/2+1,n/2+1}}{c_{n/2+1}\left(A^{(k-1)}\right)} = A^{(k)}_{n/2+1,n/2+1},
\end{align*}
where the third equality is by the inductive assumption.
Thus, \eqref{eq-equal-item-5} is proved. 
Similarly, one can also prove \eqref{eq-equal-item-6}.
This finishes the proof of the lemma.

\subsection*{Proof of \Cref{lem-relation-rowsum-columnsum}}
By \Cref{lemma-lb-equal-item-density-positive}, we have 
\begin{equation}\label{eq-sum-rn/2-positive}
\begin{aligned}
r_{n/2}\left(A^{(k)}\right) &= A^{(k)}_{n/2,n/2} + A^{(k)}_{n/2,n/2+1} +\sum_{t\leq n/2-1} A^{(k)}_{n/2,t} + \sum_{t\geq n/2+2} A^{(k)}_{n/2,t} 
\\&= A^{(k)}_{n/2,n/2} + A^{(k)}_{n/2,n/2+1} + \left(\frac{n}{2} - 1\right)(b_k + x_k),   
\end{aligned}
\end{equation}
\begin{equation}
\begin{aligned}\label{eq-sum-cn/2-positive}
c_{n/2}\left(A^{(k)}\right) &= A^{(k)}_{n/2,n/2} +A^{(k)}_{n/2+1,n/2} + \sum_{t\leq n/2-1} A^{(k)}_{t,n/2} + \sum_{t\geq n/2+2} A^{(k)}_{t,n/2} 
\\&= A^{(k)}_{n/2,n/2} +A^{(k)}_{n/2,n/2+1} + \left(\frac{n}{2} - 1\right)(a_k + y_k).
\end{aligned}
\end{equation}
Thus, \eqref{eq-relation-rnovertwo-cnovertwo-positive} is immediate.
Similarly, we also have \eqref{eq-relation-rnovertwoplusone-cnovertwoplusone-positive}.
By \Cref{lemma-lb-equal-item-density-positive}, we also have
\begin{align*}
r_{1}\left(A^{(k)}\right) &= A^{(k)}_{1,n/2} + A^{(k)}_{1,n/2+1} + \sum_{t\leq n/2-1}A^{(k)}_{1,t} + \sum_{t\geq n/2+2}A^{(k)}_{1,t} = 2a_k +  \left(\frac{n}{2} - 1\right)u_k + \sum_{t\leq n/2-1}A^{(k)}_{1,t}.
\end{align*}
\begin{align*}
c_{1}\left(A^{(k)}\right) = A^{(k)}_{n/2,1} + A^{(k)}_{n/2+1,1} +  \sum_{t\leq n/2-1}A^{(k)}_{t,1} + \sum_{t \geq n/2+2}A^{(k)}_{t,1} = 2x_{k} + \left(\frac{n}{2} - 1\right) v_k + \sum_{t\leq n/2-1}A^{(k)}_{t,1}.
\end{align*}
\begin{align*}
\sum_{t\leq n/2-1}A^{(k)}_{1,t} &= 
\sum_{t\leq n/2-1}\left(A^{(k)}_{1,t}\cdot \id{A_{1,t} = 1} + A^{(k)}_{1,t}\cdot \id{A_{1,t} = \beta} \right) \\&=  \sum_{t\leq n/2-1}\left(A^{(k)}_{t,1}\cdot \id{A_{t,1} = 1} + A^{(k)}_{t,1}\cdot \id{A_{t,1} = \beta} \right) 
= \sum_{t\leq n/2-1}A^{(k)}_{t,1}.
\end{align*}
Thus, we have 
\begin{align*}
r_{1}\left(A^{(k)}\right) = c_{1}\left(A^{(k)}\right) + 2(a_k-x_k) + \left(\frac{n}{2} - 1\right)(u_k- v_k).
\end{align*}
Similarly, we also have \eqref{eq-relation-rone-cn-positive} and \eqref{eq-relation-rn-cone-positive}.

\subsection*{Proof of \Cref{lem-base-case}}
Note that
\begin{align}\label{eq-ronea-ub-positive}
r_1(A) = \lceil \gamma n \rceil + 4 \lceil \gamma n \rceil/n + (n-2 - \lceil \gamma n \rceil)\beta.
\end{align}
Thus, 
\[u_0 = \frac{A_{1,n}}{\lceil \gamma n \rceil + 4 \lceil \gamma n \rceil/n + (n-2 - \lceil \gamma n \rceil)\beta}.\]
Combined with $A_{1,n} =\beta = \varepsilon^{8}/(100 n^{61})$, \eqref{eq-def-rho-positive} and \eqref{eq-condition-delta}, we have
$u_0 < \beta \leq \rho \delta^6/n^{51}$.
Similarly, one can also verify that 
$\max\{y_{0},u_{0},v_{0}\} \leq \rho \delta^6/n^{51}$.
Thus, \eqref{eq-ub-xyuvzero} is immediate.
By \eqref{eq-ronea-ub-positive} and $\beta \leq \rho \delta^6/n^{51}$, we have
\[a_0 = \frac{2 \lceil \gamma n \rceil/n}{r_1(A)}\in \left[ \frac{2}{n + 4 + \rho}, \frac{2}{n+ 4}\right].\]
Moreover, by $\beta \leq \rho \delta^6/n^{51}$ we have 
\[r_{n/2}(A)  = A_{n/2,n/2} + A_{n/2,n/2+1} +\sum_{t\leq n/2-1} A_{n/2,t} + \sum_{t>n/2+1} A_{n/2,t} = \frac{n}{2}(1 + \beta) < \frac{n}{2} + \rho.\]
Therefore,
\[b_0 = \frac{1}{r_{n/2}(A)}\in \left[ \frac{2}{n + \rho}, \frac{2}{n}\right].\]
Thus, one can verify that 
\[\rho = \frac{\delta^2}{100n^{2}} < \frac{\delta^2}{n-1} - \frac{\delta^2}{n} \leq \frac{2/(n-1) - 2/n}{1+ 1/\delta} < \frac{2/(n-1) - b_0}{1+1/\delta}<\frac{2/(n-1) - a_0}{1+1/\delta}.\]
Hence,
\[a_0 < 2/(n-1) - \rho/\delta - \rho, \quad b_0 < 2/(n-1) - \rho/\delta -\rho.\]

\subsection*{Proof of \Cref{lem-item-max-xyuv-boost}}
Note that for each even $k\geq 3$, we have
\begin{align}\label{eq-xk-xkminustwo-conernover2}
x_k = \frac{x_{k-3}}{r_{n/2}\left(A^{(k-3)}\right)\cdot c_{1}\left(A^{(k-2)}\right)\cdot r_{n/2}\left(A^{(k-1)}\right)}.
\end{align}
Moreover, one can verify that the matrix $A$ in \Cref{thm-lowerbound-positive} is $(\lceil\gamma n\rceil/n, 2\lceil\gamma n\rceil/n)$-dense.
By \Cref{item-forth-fact-c-a-a0-a1} of \Cref{fact-c-a-a0-a1} and $\gamma\in (1/4,1/2)$, we have
\[c_{1}\left(A^{(k-2)}\right) \geq 2\gamma^2 \geq 1/8, \quad r_{n/2}\left(A^{(k-1)}\right)\geq 2\gamma^2\geq 1/8.\]
Combined \eqref{eq-xk-xkminustwo-conernover2},
we have 
\begin{align}\label{eq-xk-leq-sixtyfour-xkminustwo}
x_k \leq 512x_{k-3}.
\end{align}
Similarly, we also have \eqref{eq-xk-leq-sixtyfour-xkminustwo} for each odd $k\geq 3$.
In summary, we always have \eqref{eq-xk-leq-sixtyfour-xkminustwo} for each $k\geq 3$.
In addition, one can also verify that 
$y_k \leq 512y_{k-3}$, $u_k \leq 512u_{k-3}$ and $v_k \leq 512v_{k-3}$ for each $k\geq 3$. Thus, to prove \eqref{eq-max-xyuv-less-rho-over-fourn}, it is sufficient to prove that 
\begin{align}\label{eq-max-xyuv-less-rho-over-fourn-k-leq-tminustwo}
\forall k\leq t-3, \quad \max\{x_{k},y_{k},u_{k},v_{k}\} < \frac{\rho}{2048n}.
\end{align}
In the following, we prove \eqref{eq-max-xyuv-less-rho-over-fourn-k-leq-tminustwo}. Suppose \emph{w.l.o.g.} $t-3 \geq 0$.
For each $k\leq \min\{t-2,\ell\}$,
we have 
\begin{equation}
\begin{aligned}
&\quad \max\{x_{k},y_{k},u_{k},v_{k}\} \\
(\text{by $k\leq \min\{t-2,\ell\}$ and \eqref{eq-ub-xyuv-rho-kleql}})\quad &\leq \frac{\rho\delta^6}{n^{51}}\cdot \left(1+\frac{2}{n-1}\right)^{3k}\\
(\text{by $k\leq \ell$})\quad &\leq \frac{\rho\delta^6}{n^{51}}\cdot \left(1+\frac{2}{n-1}\right)^{3{\ell}} \\
(\text{by \eqref{eq-def-ell-positive}})\quad &=\frac{\rho\delta^6}{n^{51}} \cdot\left(1+\frac{2}{n-1}\right)^{24n(\lceil\log n\rceil -1)}\\
\left(\text{by $n>10^4$}\right)\quad &\leq \frac{\rho\delta^6}{n^{51}} \cdot\exp({49(\lceil\log n\rceil -1)})\\
&\leq \frac{\rho\delta^6}{n^2}\\
\left(\text{by \eqref{eq-condition-delta} and $n>10^4$}\right)\quad &< \frac{\rho}{2048n}.
\end{aligned}
\end{equation}
For each odd $\ell\leq k\leq t-3$, 
by $t < K$ and \eqref{eq-def-K-positive}, we have
$b_{k-1} \geq \delta$ and $b_{k+1} \geq \delta$.
In addition, by \Cref{item-first-fact-a0-a1} of \Cref{fact-c-a-a0-a1}, 
we have $b_{\ell}\leq 1$.
Thus, 
\begin{equation}
\begin{aligned}
&\quad \max\{x_{k},y_{k},u_{k},v_{k}\} \\
(\text{by $\ell\leq k\leq t-3$ and \eqref{eq-ub-xyuv-rho-kgeql-even}})\quad &\leq \frac{\rho\delta^6}{n^{2}}\cdot \left(\frac{b_{\ell}}{b_{k-1}}\right)^6\cdot \left(\frac{b_{k-1}}{b_{k+1}}\right)^3.\\
(\text{by $b_{\ell}\leq 1$})\quad &\leq \frac{\rho\delta^6}{n^{2}}\cdot \left(\frac{1}{b_{k-1}}\right)^3\cdot \left(\frac{1}{b_{k+1}}\right)^3 \\
(\text{by $b_{k-1}\geq \delta$ and $b_{k+1}\geq \delta$})\quad &\leq \frac{\rho\delta^6}{n^{2}}\cdot \delta^{-6}\\
\left(\text{by $n>10^4$}\right)\quad  &< \frac{\rho}{2048n}.
\end{aligned}
\end{equation}
Similarly, one can also verify that $\max\{x_{k},y_{k},u_{k},v_{k}\}< \rho/(2048n)$ for each even $\ell\leq k\leq t-3$.
In summary, \eqref{eq-max-xyuv-less-rho-over-fourn-k-leq-tminustwo} is proved, which finishes the proof of the lemma.

\subsection*{Proof of \Cref{lem-iterative-relation}}
By \Cref{lem-item-max-xyuv-boost}, we have \eqref{eq-max-xyuv-less-rho-over-fourn} holds.
If $k\leq t$ is odd, by \eqref{eq-relation-rone-cn-positive} we have
\begin{align*}
a_{k+1} = \frac{a_k}{r_{1}\left(A^{(k)}\right)} = \frac{a_k}{1+2(a_k - x_k) + (n/2-1)(u_k-v_k)}.
\end{align*}
Combined with \eqref{eq-max-xyuv-less-rho-over-fourn}, we have \eqref{eq-ak-recursion-positive}.
By \eqref{eq-relation-rnovertwo-cnovertwo-positive}, we have
\begin{equation*}
\begin{aligned}
b_{k+1} &= \frac{b_k}{r_{n/2}\left(A^{(k)}\right)}=\frac{b_k }{c_{n/2}\left(A^{(k)}\right) + \left(\frac{n}{2} - 1\right)(b_k+x_k - a_k - y_k)}=\frac{b_k}{1 + \left(\frac{n}{2} - 1\right)(b_k+x_k - a_k - y_k)},
\end{aligned}
\end{equation*}
Combined with \eqref{eq-max-xyuv-less-rho-over-fourn}, we have \eqref{eq-bk-recursion-complex-positive}.
Similarly, if $k\leq t$ is even, by \eqref{eq-relation-rn-cone-positive} we have
\begin{align*}
b_{k+1} = \frac{b_k}{c_{n}\left(A^{(k)}\right)} = \frac{b_k}{1+2(b_k - y_k) + (n/2-1)(u_k-v_k)}.
\end{align*}
Combined with \eqref{eq-max-xyuv-less-rho-over-fourn}, we have \eqref{eq-bk-recursion-positive}.
By \eqref{eq-relation-rnovertwo-cnovertwo-positive}, we have
\begin{equation*}
\begin{aligned}
a_{k+1} &= \frac{a_k}{c_{n/2}\left(A^{(k)}\right)}=\frac{a_k }{r_{n/2}\left(A^{(k)}\right) + \left(\frac{n}{2} - 1\right)(a_k+y_k - b_k - x_k)}=\frac{a_k}{1 + \left(\frac{n}{2} - 1\right)(a_k+y_k - b_k - x_k)}.
\end{aligned}
\end{equation*}
Combined with \eqref{eq-max-xyuv-less-rho-over-fourn}, we have \eqref{eq-ak-recursion-complex-positive}.

\subsection*{Proof of \Cref{lem-upperbouns-akbk-simple}}
We prove this lemma by induction.
For the base step where $k=0$, by \eqref{eq-a0-b0-less-twonminusone-minusnrho}, \eqref{eq-def-rho-positive} and $n>10000$, we have \eqref{eq-ub-maxab-positive}, \eqref{eq-lowerbound-one-minusnovertwobk} and \eqref{eq-lowerbound-one-minusnovertwoak} immediately.

For the inductive step where $1\leq k\leq t$, we assume \emph{w.l.o.g.} $k$ is even, since
the proof for odd $k$ is similar.
By the induction assumption, we have 
\begin{align}
&a_{k-1} < 2/(n-1) - (1/\delta - (k-2))\rho < 1/10, \label{eq-ak-less-twonminustwominusnvarepsilon}\\
&b_{k-1} < 2/(n-1) - (1/\delta - (k-2))\rho < 1/10,\label{eq-bk-less-twonminustwominusnvarepsilon}\\
&1 + \left(\frac{n}{2} - 1\right)(b_{k-1}-a_{k-1}) - \rho \geq 1 - \left(\frac{n}{2} - 1\right)a_{k-1} -\rho >0.\label{eq-one-plus-novertwominusone-bkminusak-larger-one}
\end{align}
In addition, by \Cref{item-first-fact-a0-a1} of \Cref{fact-c-a-a0-a1}, 
we have $b_{k}\leq 1$.
Thus, if $b_{k-1} \leq a_{k-1}$, we have 
\begin{align*}
&\quad b_{k}(1 - \rho) - a_{k-1} \\
(\text{by \eqref{eq-bk-recursion-complex-positive} })\quad &\leq\frac{b_{k-1}(1-\rho)}{1 + \left(\frac{n}{2} - 1\right)(b_{k-1}-a_{k-1}) -\rho} - a_{k-1}\\
&\leq \frac{b_{k-1}(1-\rho) - a_{k-1}(1 + \left(n/2 - 1\right)(b_{k-1}-a_{k-1})- \rho)}{1 + \left(\frac{n}{2} - 1\right)(b_{k-1}-a_{k-1}) - \rho}\\
\quad &= \frac{(1 - a_{k-1}\left(n/2 - 1\right) - \rho)(b_{k-1}-a_{k-1}))}{1 + \left(\frac{n}{2} - 1\right)(b_{k-1}-a_{k-1}) - \rho}\\
(\text{by $b_{k-1} \leq a_{k-1}$ and \eqref{eq-one-plus-novertwominusone-bkminusak-larger-one}})\quad &\leq 0.
\end{align*}
Hence, we have 
\begin{align*}
&\quad b_{k}\\
&\leq a_{k-1} + \rho b_{k}\\
(\text{by $b_{k}\leq 1$})\quad &\leq a_{k-1} +\rho \\
(\text{by \eqref{eq-ak-less-twonminustwominusnvarepsilon}})\quad &<  2/(n-1) - (1/\delta - (k-2))\rho + \rho \\
&= 2/(n-1) - (1/\delta - (k-1))\rho.
\end{align*}
If $b_{k-1} > a_{k-1}$, by \eqref{eq-condition-delta} and \eqref{eq-def-rho-positive}, one can verify that $\rho<1/100$.
Thus, 
we have
\begin{align*}
&\quad b_{k}\\
(\text{by \eqref{eq-bk-recursion-complex-positive} })\quad &\leq \frac{b_{k-1}}{1 + \left(\frac{n}{2} - 1\right)(b_{k-1}-a_{k-1})- \rho} \\
(\text{by $b_{k-1}>a_{k-1}$})\quad 
&< \frac{b_{k-1}}{1 - \rho}\\
(\text{by \eqref{eq-bk-less-twonminustwominusnvarepsilon} and $\rho<1/100$})\quad &\leq b_{k-1}+ \rho\\
(\text{by \eqref{eq-bk-less-twonminustwominusnvarepsilon}})\quad &\leq 2/(n-1) - (1/\delta - (k-2))\rho + \rho\\
&= 2/(n-1) - (1/\delta - (k-1))\rho.
\end{align*}
In summary, we always have 
\begin{align}\label{eq-bk-less-twonminusone-minusrho}
b_{k} < 2/(n-1) - (1/\delta - (k - 1))\rho.
\end{align}
In addition, by $k\leq t< K$ and \eqref{eq-def-K-positive},
we have $k < 1/\delta$. Thus,
\begin{align}\label{eq-nvarepsilon-minus-kminusone-larger-one}
1/\delta - (k - 1)>1. 
\end{align}
Combined with $n>10000$ and \eqref{eq-bk-less-twonminusone-minusrho}
we have 
\begin{align}\label{eq-bk-less-twonminusone}
b_{k} < 2/(n-1) - (1/\delta - (k - 1))\rho < 2/(n-1) < 1/10.
\end{align}
By \eqref{eq-bk-less-twonminusone-minusrho} and \eqref{eq-nvarepsilon-minus-kminusone-larger-one} and $n>10000$,
we also have
\begin{align}\label{lem-item-two-1}
1 + \left(\frac{n}{2} - 1\right)(a_{k}-b_{k}) - \rho \geq 1 - \left(\frac{n}{2} - 1\right)b_{k} -\rho \geq  \left(\frac{n}{2} - 1\right)\rho - \rho >0.
\end{align}

Moreover, by \eqref{eq-def-K-positive} and $k\leq t< K$,
we have $a_{k-1}\geq \delta$.
Combined with $\rho = \delta^{2}/(100n^{2})$,
we have $a_{k-1}\geq \delta>\rho$.
Thus,
\begin{equation}\label{lem-item-two-2}
\begin{aligned}
&\quad a_{k} \\
(\text{by \eqref{eq-ak-recursion-positive}})\quad &\leq \frac{a_{k-1}}{1+2a_{k-1} -\rho} \\
(\text{by $a_{k-1}> \rho$})\quad &\leq \frac{a_{k-1}}{1+a_{k-1}} \\
\quad &\leq a_{k-1}\\
(\text{by \eqref{eq-ak-less-twonminustwominusnvarepsilon}})\quad &\leq 2/(n-1) - (1/\delta - (k-2))\rho\\
\quad &\leq 2/(n-1) - (1/\delta - (k-1))\rho\\
(\text{by \eqref{eq-nvarepsilon-minus-kminusone-larger-one}})\quad &< 2/(n-1)\\
(\text{by $n>10^4$})\quad &< \frac{1}{10}.
\end{aligned}
\end{equation}
By \eqref{lem-item-two-2}, \eqref{eq-nvarepsilon-minus-kminusone-larger-one} and $n>10000$,
we also have
\begin{align}\label{lem-item-two-3}
1 + \left(\frac{n}{2} - 1\right)(b_{k}-a_{k}) - \rho \geq 1 - \left(\frac{n}{2} - 1\right)a_{k} -\rho \geq  \left(\frac{n}{2} - 1\right)\rho - \rho >0.
\end{align}
Combing \eqref{eq-bk-less-twonminusone}, \eqref{lem-item-two-1}, \eqref{lem-item-two-2} and \eqref{lem-item-two-3} together,
the induction step is complete and the lemma is proved.

\subsection*{Proof of \Cref{lem-refine-relation-akbk}}

We prove this lemma by induction.
For the base step where $k=0$, by \eqref{eq-a0-b0-range} and \eqref{eq-def-rho-positive}, we have
\begin{align*}
\frac{2}{n + 4 + \rho}\leq a_0\leq  \frac{2}{n + 4}<   \frac{2}{n + \rho}\leq b_0 \leq \frac{2}{n}.
\end{align*}
Hence, 
\begin{equation}
\begin{aligned}
&\quad a_{1} \\
(\text{by \eqref{eq-ak-recursion-complex-positive}, \eqref{eq-lowerbound-one-minusnovertwobk}})\quad &\leq \frac{2/(n + 4)}{1+(n/2-1)\left(2/(n + 4) - 2/n\right)- \rho} \\
&=\frac{2n}{n^2 - \rho n^2 -4n\rho +8}\\
(\text{by \eqref{eq-def-rho-positive}}) \quad &< \frac{2n}{n^2 + 7}\\
(\text{by \eqref{eq-def-rho-positive}}) \quad &< \frac{2}{n+\rho}.
\end{aligned}
\end{equation}
Hence, $a_0 < b_0, a_1 <b_0$.
In addition, by \eqref{eq-a0-b0-range} and $n>10000$, we also have 
\begin{align*}
a_0(1+2b_0) - b_0 + 6\rho a_0\geq \frac{2}{n+4+\rho}\left(1+\frac{4}{n} + 6\rho\right) - \frac{2}{n} =  \frac{2}{n}\cdot \frac{6\rho n - \rho}{n+4+\rho} >0.
\end{align*}
Thus, one can verify that \eqref{eq-key-relation-positive-even-1}, \eqref{eq-key-relation-positive-even-2} and \eqref{eq-key-relation-positive-even-3} are satisfied.
The base case is proved.

\vspace{0.5cm}

For the inductive step where $k> 0$,
we assume \emph{w.l.o.g.} that $k$ is even,
since the proof for odd $k$ is similar.
According to the inductive assumption for odd $k-1$, we have
\begin{align}
b_{k-1} &< a_{k-1} + (6(k-1) + 2)\rho b_{k-1},\label{eq-reduction-assumption-7-positive}\\
a_{k-1} &<  b_{k-1}(1 + 2a_{k-1}) + 6k\rho b_{k-1} ,\label{eq-reduction-assumption-1-positive}\\
b_{k} &< a_{k-1} + (6(k-1) + 4)\rho b_{k-1}.\label{eq-reduction-assumption-2-positive}
\end{align}
In addition, by $k\leq t<K$ and \eqref{eq-def-K-positive},
we have $k < 1/\delta$, $b_k \geq \delta$.
Combined with \eqref{eq-def-rho-positive}, we have 
\begin{align}
\rho &\leq 1/(100n^2k^2)\label{eq-klessdelta-rho}.
\end{align}

At first, we show $a_{k} <  b_{k} + (6k+2)\rho a_{k}$.
By \eqref{eq-lowerbound-one-minusnovertwoak} we have 
\begin{align*}
1 + \left(\frac{n}{2} - 1\right)(b_{k-1}-a_{k-1}) + \rho > 0.
\end{align*}
One can also verify that 
\begin{align}
1+2a_{k-1}-\rho &> 0,\label{eq-oneplustwoakminusone-large-one}\\
1 - (6(k-1)+2)\rho &> 1 - (6k+2)\rho > 0.\label{eq-oneminussixkplustworho-large-zero}
\end{align}
Thus, 
\begin{align}\label{eq-prod-oneplustwoakminusone-large-one}
(1+2a_{k-1}-\rho)\left(1 + \left(\frac{n}{2} - 1\right)(b_{k-1}-a_{k-1})+\rho\right) > 0.
\end{align}
Moreover, by \eqref{eq-reduction-assumption-7-positive} and \eqref{eq-oneminussixkplustworho-large-zero} we have
\begin{align}\label{eq-bkminus1-upper-bound-akminusone}
b_{k-1} \leq \frac{a_{k-1}}{1 - (6(k-1)+2)\rho}.
\end{align}
Thus, 
\begin{equation}
\begin{aligned}\label{eq-bkminus1-upper-bound}
&\quad 6k\rho  b_{k-1} + \rho b_{k-1} \\
(\text{by \eqref{eq-bkminus1-upper-bound-akminusone}})\quad &\leq \frac{6k + 1}{1 - (6(k-1)+2)\rho}\cdot \rho a_{k-1} \\
(\text{by \eqref{eq-klessdelta-rho}})\quad &< (6k + 2)\rho a_{k-1}.
\end{aligned}
\end{equation}
Therefore, 
if $\left(n/2 - 1\right)(b_{k-1}-a_{k-1})+\rho \geq 0$, we have
\begin{align*}
&\quad a_{k} - b_{k} - (6k+2)\rho a_{k}\\
(\text{by \eqref{eq-ak-recursion-positive}, \eqref{eq-bk-recursion-complex-positive}})\quad &<  \frac{(1-(6k+2)\rho)a_{k-1}}{1+2a_{k-1} - \rho} - \frac{b_{k-1}}{1 + \left(\frac{n}{2} - 1\right)(b_{k-1}-a_{k-1})+ \rho}\\
\left(\text{by $\left(\frac{n}{2} - 1\right)(b_{k-1}-a_{k-1})+\rho \geq 0$}\right)\quad &\leq  \frac{(1-(6k+2)\rho)a_{k-1}}{1+2a_{k-1} - \rho} - b_{k-1}\\
&= \frac{a_{k-1} - 2a_{k-1}b_{k-1} - b_{k-1} + \rho b_{k-1} - (6k+2)\rho a_{k-1}}{1+2a_{k-1}-\rho}\\
(\text{by \eqref{eq-reduction-assumption-1-positive}, \eqref{eq-oneplustwoakminusone-large-one}})\quad&<   \frac{6k\rho b_{k-1} +\rho b_{k-1}- (6k+2)\rho a_{k-1}}{1+2a_{k-1}-\rho}\\
(\text{by \eqref{eq-bkminus1-upper-bound}, \eqref{eq-oneplustwoakminusone-large-one}})\quad &< \frac{(6k+2)\rho  a_{k-1} - (6k+2)\rho a_{k-1}}{1+2a_{k-1}-\rho}\\
\quad &= 0.
\end{align*}
Otherwise, $\left(n/2 - 1\right)(b_{k-1}-a_{k-1})+\rho <0$. We also have
\begin{align*}
&\quad a_{k} - b_{k} - (6k+2)\rho a_{k}\\
(\text{by \eqref{eq-ak-recursion-positive}, \eqref{eq-bk-recursion-complex-positive}})\quad &\leq  \frac{(1 - (6k+2)\rho)a_{k-1}}{1+2a_{k-1} - \rho} - \frac{b_{k-1}}{1 + \left(\frac{n}{2} - 1\right)(b_{k-1}-a_{k-1})+ \rho}\\
\left(\text{by \eqref{eq-oneminussixkplustworho-large-zero}, $\left(\frac{n}{2} - 1\right)(b_{k-1}-a_{k-1})+\rho <0$}\right)\quad &< \frac{a_{k-1} - (6k+2)\rho a_{k-1} - 2a_{k-1}b_{k-1} - b_{k-1} +\rho b_{k-1} }{(1+2a_{k-1}-\rho)(1 + \left(\frac{n}{2} - 1\right)(b_{k-1}-a_{k-1})+ \rho)}\\
(\text{by \eqref{eq-reduction-assumption-1-positive}, \eqref{eq-prod-oneplustwoakminusone-large-one}})\quad&<   \frac{6k\rho b_{k-1} + \rho b_{k-1} - (6k+2)\rho a_{k-1} }{(1+2a_{k-1}-\rho)(1 + \left(\frac{n}{2} - 1\right)(b_{k-1}-a_{k-1})+ \rho)} \\
(\text{by \eqref{eq-bkminus1-upper-bound}, \eqref{eq-prod-oneplustwoakminusone-large-one}})\quad &< \frac{(6k+2)\rho a_{k-1} - (6k+2)\rho a_{k-1}}{(1+2a_{k-1}-\rho)(1 + \left(\frac{n}{2} - 1\right)(b_{k-1}-a_{k-1})+ \rho)}\\
\quad &= 0.
\end{align*}
In summary, we always have $a_{k}  < b_{k} + (6k+2)\rho a_{k}$.

\vspace{0.3cm}

Secondly, we show $a_{k+1} <  b_{k}+ (6k+4)\rho a_{k}$.
Note that
\begin{align*}
 1 - b_{k}\left(\frac{n}{2} - 1\right) -\rho - \left(1 - \left(\frac{n}{2} - 1\right)(b_{k} - a_{k}) - \rho\right)
= \left(\frac{n}{2} - 1\right)(b_{k} - a_{k} - b_{k})
\leq 0.
\end{align*}
Hence, we have 
\begin{align}\label{eq-one-minus-bk-prod-novertwominusone-upperbound}
1 - b_{k}\left(\frac{n}{2} - 1\right) -\rho \leq  \left(1 - \left(\frac{n}{2} - 1\right)(b_{k} - a_{k}) - \rho\right).
\end{align}
Moreover, by \Cref{item-first-fact-a0-a1} of \Cref{fact-c-a-a0-a1}, 
we have $b_{k}\leq 1$.
Thus,
\begin{equation}
\begin{aligned}\label{eq-one-minus-novertwominusone-bkminusak-lowerbound}
&\quad 1 - \left(\frac{n}{2} - 1\right)(b_{k} - a_{k}) - \rho \\
(\text{by \eqref{eq-ak-recursion-positive}})\quad &\geq 1 - \left(\frac{n}{2} - 1\right)\left(b_{k} - \frac{a_{k-1}}{1+2a_{k-1}+\rho}\right) - \rho \\
\quad &= 1 - \left(\frac{n}{2} - 1\right)\left( \frac{b_{k}  + 2a_{k-1}b_{k}+\rho b_{k}-a_{k-1}}{1+2a_{k-1}+\rho}\right) - \rho \\
\left(\text{by \eqref{eq-reduction-assumption-2-positive}}\right)\quad &\geq 1 - \left(\frac{n}{2} - 1\right)\left( \frac{(6(k-1) + 4)\rho b_{k-1}  + 2a_{k-1}b_{k}+\rho b_{k}}{1+2a_{k-1}+\rho}\right) - \rho \\
\left(\text{by \eqref{eq-klessdelta-rho}, \eqref{eq-ub-maxab-positive}, $n>10^4$}\right)\quad &> \frac{1}{2}.
\end{aligned}
\end{equation}
Thus,
\begin{align*}
&\quad a_{k+1} - b_{k} - (6k + 4)\rho a_{k} \\
(\text{by \eqref{eq-ak-recursion-complex-positive}})\quad &\leq \frac{a_{k}}{1 + \left(\frac{n}{2} - 1\right)(a_{k}-b_{k}) - \rho} - b_{k} - (6k + 4)\rho a_{k}\\
&=\frac{\left(1 - b_k\left(\frac{n}{2} - 1\right) -\rho\right)(a_{k} - b_k) + \rho a_{k}}{1 - \left(\frac{n}{2} - 1\right)(b_{k}-a_{k}) - \rho} - (6k + 4)\rho a_{k}\\
(\text{by \eqref{eq-one-minus-bk-prod-novertwominusone-upperbound} and \eqref{eq-one-minus-novertwominusone-bkminusak-lowerbound}})\quad &< \max\{a_k-b_k,0\} + 2 \rho a_{k} - (6k + 4)\rho a_k\\
\left(\text{by $a_{k}  < b_{k} + (6k+2)\rho a_{k}$}\right)\quad &< (6k+2)\rho a_k + 2\rho a_k - (6k + 4)\rho a_{k}\\
&\leq 0.
\end{align*}
Thus, we have $a_{k+1} < b_{k} + (6k+4)\rho a_{k}$.

\vspace{0.3cm}

Thirdly, we prove $b_{k} < a_{k}(1 + 2b_{k}) + 6(k+1)\rho a_{k}$.
By \eqref{eq-ak-recursion-positive},
we have 
\begin{align}\label{eq-relation-alplustwoakakminus1-akminusone}
a_{k} + 2a_{k}a_{k-1}+ \rho a_{k} \geq  a_{k-1}.
\end{align}
Thus,
\begin{align}\label{eq-relation-alplustwoakakminus1-akminusone-simple-form}
(1+ \rho) a_{k} \geq  (1-2a_k)a_{k-1}.
\end{align}
Hence, 
\begin{equation}\label{eq-fourkminustwosquareplusfour-upperbound}
\begin{aligned}
&\quad (6k+4)(1-2a_k)b_{k-1} \\
(\text{by \eqref{eq-bkminus1-upper-bound-akminusone}})\quad &\leq (6k+4)\cdot \frac{(1-2a_k)a_{k-1}}{1 - (6(k-1)+2)\rho}\\ 
(\text{by \eqref{eq-relation-alplustwoakakminus1-akminusone-simple-form}})\quad &\leq (6k+4)\cdot\frac{(1+\rho)a_k}{1 - (6(k-1)+2)\rho}\\
(\text{by \eqref{eq-klessdelta-rho}})\quad &< (6k+5)a_k.
\end{aligned}
\end{equation}
Moreover, by \eqref{eq-ub-maxab-positive} we have
$1 - 2a_{k}>0$.
Therefore,
\begin{align*}
 &\quad a_{k}(1 + 2b_{k}) + 6(k+1)\rho a_{k}- b_{k}  \\
&=  a_{k} - a_{k-1} + a_{k-1} - b_{k} + 2a_{k}b_{k} + 6(k+1)\rho a_{k}\\
(\text{by \eqref{eq-relation-alplustwoakakminus1-akminusone}}) \quad&\geq - 2a_{k}a_{k-1} - \rho a_{k}+ a_{k-1} - b_{k}  + 2a_{k}b_{k}  + 6(k+1)\rho a_{k}\\
&= (a_{k-1} - b_{k})(1 - 2a_{k}) - \rho a_{k} + 6(k+1)\rho a_{k}\\
(\text{by \eqref{eq-reduction-assumption-2-positive} and $1 - 2a_{k}>0$})\quad& \geq -(6k+4)\rho b_{k-1}(1-2a_k) 
 - \rho a_{k} +  6(k+1)\rho a_{k}\\
(\text{by \eqref{eq-fourkminustwosquareplusfour-upperbound}})\quad &>-(6k+5)\rho a_k- \rho a_{k} +  6(k+1)\rho a_{k}\\
&=0.
\end{align*}
Hence, $b_{k} < a_{k}(1 + 2b_{k}) + 6(k+1)\rho a_{k}$.
This finishes the induction and the proof of the lemma.

\subsection*{Proof of \Cref{lem-decay-rate-bk}}
For each $k\leq t< K$, by \eqref{eq-def-K-positive} we have $b_k\geq \delta$ and $k<1/\delta$.
Combined with \eqref{eq-def-rho-positive},
we have
\begin{align}\label{eq-twominusnminustwobk-largerthan-twominusnbk}
2 - \left(n - 2\right)(b_{k} + 4k\rho) \geq 2 - \left(n - 2\right)b_{k} - \delta > 2-(n-1) b_k.
\end{align}
In addition, for each even $2\leq k\leq t$, one can also verify that $2b_{k-2} - \rho >0$.
Combined with \eqref{eq-bk-recursion-positive},
we have $b_{k-1} < b_{k-2}$.
Thus, for each even $k\leq t$ we have
\begin{equation}\label{eq-prod-oneplustwobkminustwominusrho-larger-zero}
\begin{aligned}
&\quad (1+2b_{k-2} - \rho)\left(1 - \frac{\left(n - 2\right)(a_{k-1} + 3k\rho)b_{k-2}}{1+2b_{k-2}-\rho} - \rho\right)\\
&= 1+2b_{k-2} - \left(n - 2\right)(a_{k-1} + 3k\rho)b_{k-2} - (2+2b_{k-2}-\rho)\rho\\
(\text{by \eqref{eq-key-relation-positive-even-3}})\quad &\geq 1+2b_{k-2} - \left(n - 2\right)(b_{k-2} + (6k-8)\rho a_{k-2} + 3k\rho)b_{k-2} - (2+2b_{k-2}-\rho)\rho\\
(\text{by $2b_{k-2} > \rho$})\quad &> 1+2b_{k-2} - \left(n - 2\right)(b_{k-2} + (6k-8)\rho a_{k-2} + 3k\rho)b_{k-2} \\
(\text{by \eqref{eq-ub-maxab-positive}})\quad &> 1+2b_{k-2} - \left(n - 2\right)(b_{k-2} + 4k\rho)b_{k-2} \\
(\text{by \eqref{eq-twominusnminustwobk-largerthan-twominusnbk}})\quad &>1+2b_{k-2}\left(1 - \frac{(n-1) b_{k-2}}{2}\right)\\
(\text{by \eqref{eq-ub-maxab-positive}})\quad &>0.
\end{aligned}
\end{equation}
Thus, for each even $k\leq t$,
\begin{equation}\label{eq-bkplustwo-bk-recursion-positive-init}
\begin{aligned}
&\quad b_{k}\\ 
(\text{by \eqref{eq-bk-recursion-complex-positive}})\quad &\leq \frac{b_{k-1}}{1 + \left(\frac{n}{2} - 1\right)(b_{k-1}-a_{k-1}) - \rho} \\
(\text{by \eqref{eq-key-relation-positive-odd-2}})\quad &\leq \frac{b_{k-1}}{1 - \left(\frac{n}{2} - 1\right)(2a_{k-1} + 6k\rho)b_{k-1}- \rho} \\
&= \frac{b_{k-1}}{1 - \left(n - 2\right)(a_{k-1} + 3k\rho)b_{k-1}- \rho} \\
(\text{by \eqref{eq-bk-recursion-positive}, \eqref{eq-prod-oneplustwobkminustwominusrho-larger-zero}})\quad &\leq \frac{b_{k-2}}{(1+2b_{k-2} - \rho)(1 - \left(n - 2\right)(a_{k-1} + 3k\rho)b_{k-2}/(1+2b_{k-2}-\rho) - \rho)}\\
(\text{by \eqref{eq-prod-oneplustwobkminustwominusrho-larger-zero}})\quad & < \frac{b_{k-2}}{1+2b_{k-2}\left(1 - (n-1) b_{k-2}/2\right)}\\
(\text{by \eqref{eq-ub-maxab-positive}})\quad &\leq b_{k-2}\left(1 - \frac{3}{2}b_{k-2}\left(1 - \frac{(n-1) b_{k-2}}{2}\right)\right)\\
(\text{by \eqref{eq-ub-maxab-positive}})\quad&< b_{k-2}.
\end{aligned}
\end{equation}
In addition, by \eqref{eq-def-ell-positive} we have $\ell$ is even.
Thus, given any even $\ell+2 \leq k\leq t$, 
by \eqref{eq-bkplustwo-bk-recursion-positive-init} and $\ell, k - 2$ are even,
we have 
$b_{k - 2}\leq b_{\ell}$.
Hence,
if $b_{\ell}\leq 1/(n-1)$, we have $b_{k - 2}\leq 1/(n-1)$.
Combined with \eqref{eq-bkplustwo-bk-recursion-positive-init}, we have
\[b_{k} \leq \frac{b_{k-2}}{1+2b_{k-2}\left(1 - (n-1) b_{k-2}/2\right)} \leq  \frac{b_{k-2}}{1+b_{k-2}}.\]
Hence, to prove \eqref{eq-bk-leq-bkminustwo-oneplusbkminustwo} for $T =t$, it is sufficient to prove that $b_{\ell}\leq 1/(n-1)$.
In the following, we prove $b_{\ell}\leq 1/(n-1)$,
which finishes the prove of the lemma.

Assume that $b_{\ell}> 1/(n-1)$ for contradiction.
By \eqref{eq-ub-maxab-positive}, we can assume \emph{w.l.o.g.} that 
\begin{align}\label{eq-def-lk}
b_{k} = \frac{2}{n-1}\left(1 - \frac{L_k}{n}\right)
\end{align}
where $0<L_k\leq n$ for each $k\leq t$.
Combined with $b_{\ell}> 1/(n-1)$, we have $L_{\ell}< n/2$.
Moreover, we claim that $L_{8nr} \geq 2^{r}$ for each $1\leq r\leq \lceil\log n\rceil - 1 $ where the base of the logarithm is 2.
Combined with \eqref{eq-def-ell-positive}, we have
$L_{\ell} \geq 2^{\lceil\log n\rceil - 1 }\geq n/2$, which is contradictory with $L_{\ell}< n/2$.
Thus, we have $b_{\ell}\leq 1/(n-1)$.

In the following, we prove the claim that $L_{8nr} \geq 2^{r}$ for each $1\leq r\leq \lceil\log n\rceil - 1 $ by induction, which finishes the proof of the lemma.
For the base step where $r = 1$, by \eqref{eq-a0-b0-range} 
we have $L_0 \geq 1$.
In addition, by \eqref{eq-bkplustwo-bk-recursion-positive-init} we have $L_{k} > L_{k-2}$ for each even $k \leq \ell$.
Combined with $L_0 \geq 1$ and $L_{\ell}< \frac{n}{2}$,
we have
$1\leq L_k < n/2$ for each even $k \leq \ell$.
Hence, for each even $k \leq \ell$ we have
\begin{equation}\label{eq-bkplustwo-bk-recursion-positive}
\begin{aligned}
&\quad b_{k}\\ 
(\text{by \eqref{eq-bkplustwo-bk-recursion-positive-init}})\quad &\leq b_{k-2}\left(1 - \frac{3}{2}b_{k-2}\left(1 - \frac{(n-1) b_{k-2}}{2} \right)\right)\\
(\text{by \eqref{eq-def-lk}})\quad &\leq b_{k-2}\left(1 - \frac{3}{n-1}\left(1 - \frac{L_{k-2}}{n}\right)\left(1 - \left(1 - \frac{L_{k-2}}{n}\right)\right)\right)\\
&\leq b_{k-2}\left(1 - \frac{3}{n-1}\cdot \frac{n - L_{k-2}}{n}\cdot  \frac{L_{k-2}}{n}\right)\\
\left(\text{by $L_{k-2} < \frac{n}{2}$}\right)\quad & \leq b_{k-2}\left(1 - \frac{3}{2(n-1)}\cdot \frac{L_{k-2}}{n}\right).
\end{aligned}
\end{equation}
Thus,
\begin{align*}
&\quad b_{8n} \\
\left(\text{by \eqref{eq-bkplustwo-bk-recursion-positive} and $L_k\geq 1$}\right)\quad &\leq b_0\cdot \left(1 - \frac{3}{2n(n-1)}\right)^{4n}\\
(\text{by \eqref{eq-a0-b0-range}})\quad &\leq \frac{2}{n}\cdot \left(1 - \frac{3}{2n^2}\right)^{4n}\\
(\text{by $n>10000$})\quad &\leq \frac{2}{n}\cdot \left(1 - \frac{4}{n}\right)\\
&< \frac{2}{n-1}\cdot \left(1 - \frac{4}{n}\right).
\end{align*}
Hence, we have $L_{8n} >2$.
The base step is proved.
For the inductive step, 
we have
$L_{8n(r-1)} > 2^{r-1}$ by the inductive assumption. 
Combined with $L_{k-2} < L_{k}$ for each even $k \leq \ell$,
we have $L_{2k} > 2^{r-1}$ for each $k$ with
$4n(r-1) \leq k \leq \ell$.
Moreover, by \eqref{eq-bkplustwo-bk-recursion-positive-init} and \eqref{eq-a0-b0-range}, we have 
\begin{align}\label{eq-beightnt-lessthan-twoovern}
b_{8n(r-1)}<b_{0} < 2/n
\end{align}
Thus, we have
\begin{align*}
&\quad b_{8nr} \\
(\text{by \eqref{eq-bkplustwo-bk-recursion-positive}})\quad &\leq b_{8n(r-1)}\cdot \prod_{k= 4n(r-1)}^{4nr-1}\left(1 - \frac{3}{2(n-1)}\cdot \frac{L_{2k}}{n}\right)
\\
(\text{by $k\geq 4n(r-1)$ and then $L_{2k} > 2^{r-1}$})\quad &\leq b_{8n(r-1)}\cdot \left(1 - \frac{3\cdot 2^{r-1}}{2n^2}\right)^{4n}\\
(\text{by $n>10000$})\quad &\leq b_{8n(r-1)}\cdot \left(1 - \frac{4\cdot 2^{r-1}}{n}\right)\\
(\text{by \eqref{eq-beightnt-lessthan-twoovern}})\quad &\leq \frac{2}{n}\cdot \left(1 - \frac{4\cdot 2^{r-1}}{n}\right)\\
&< \frac{2}{n-1}\cdot \left(1 - \frac{4\cdot 2^{r-1}}{n}\right).
\end{align*}
Thus, we have $L_{8nr}> 4\cdot 2^{(r-1)} >2^{r}$.
The inductive step is finished and the claim $L_{8nr} \geq 2^{r}$ for each $1\leq r\leq \lceil\log n\rceil - 1 $ is proved.
Thus, the lemma is immediate.

\subsection*{Proof of \Cref{lem-growth-rate-xkykukvk}}
Given any even $2\leq k\leq t < K$, by \eqref{eq-def-K-positive} we have $b_{k-2}\geq \delta$, $k < 1/\delta$.
Combined with $\rho = \delta^2/(100n^{2})$, 
we have
\begin{align}
 \rho &< \frac{b_k}{100kn^2}, \quad  \rho < \frac{1}{100n^2}\label{eq-bkminustwo-klessnovervarepsilon-rho-1} \\
 \rho &< \frac{b_k^2}{100n^2}, \label{eq-bkminustwo-klessnovervarepsilon-rho-2} \\
\rho &< \frac{b_k}{100} \label{eq-bkminustwo-klessnovervarepsilon-rho-3}.
\end{align}
Thus, for each even $2\leq k\leq t$ we have
\begin{equation}
\begin{aligned}\label{eq-one-plusbkplusoneakplusone-lowerbound}
&\quad 1 + \left(\frac{n}{2} - 1\right)(b_{k-1}-a_{k-1})-\rho \\
(\text{by \eqref{eq-key-relation-positive-odd-2}}) \quad &\geq 1 - \left(\frac{n}{2} - 1\right)(2a_{k-1} + 6k\rho)b_{k-1}- \rho\\
&= 1 - \left(n-2\right)(a_{k-1} + 3k\rho)b_{k-1}- \rho \\
(\text{by \eqref{eq-bk-recursion-positive}})\quad &\geq 1 - \frac{\left(n-2\right)(a_{k-1} + 3k\rho)b_{k-2}}{1+2b_{k-2}-\rho}- \rho \\
(\text{by \eqref{eq-key-relation-positive-even-3}})\quad &\geq 1 - \frac{\left(n-2\right)(b_{k-2} + (6k-8)\rho a_{k-2} + 3k\rho)b_{k-2}}{1+2b_{k-2}-\rho}- \rho \\
(\text{by \eqref{eq-ub-maxab-positive} and $a_{k-2}\leq 1$})\quad & \geq 1 - \frac{\left(n-2\right)(b_{k-2} + 9k\rho)b_{k-2}}{1+2b_{k-2}-\rho} - \rho\\
\left(\text{by \eqref{eq-bkminustwo-klessnovervarepsilon-rho-1}}\right)\quad & \geq 1 - \frac{\left(n-2\right)b_{k-2}^2(1+1/(10n^2))}{(1+2b_{k-2})(1-1/(10n^2))} - \rho \\
\left(\text{by \eqref{eq-bkminustwo-klessnovervarepsilon-rho-2} and $n>10^4$}\right)\quad & \geq 1 - \frac{\left(n-2\right)b_{k-2}^2(1+1/(10n^2))}{(1+2b_{k-2})(1-1/(10n^2))}\cdot\left(1+\frac{1}{10n^2}\right) \\
(\text{by $n>10^4$})\quad &\geq 1-\frac{ (n-1)b^2_{k-2}}{1+2b_{k-2}}\\
(\text{by \eqref{eq-ub-maxab-positive}})\quad &\geq 1-\frac{2b_{k-2}}{1+2b_{k-2}}\\
\quad &> \frac{1}{(1+b_{k-2})^2}.
\end{aligned}
\end{equation}
In addition, by \Cref{lem-item-max-xyuv-boost} we have \eqref{eq-max-xyuv-less-rho-over-fourn} holds.
By \Cref{lem-relation-rowsum-columnsum} we have
\begin{align}
\quad x_{k} &= \frac{x_{k-1}}{r_{n/2}\left(A^{(k-1)}\right)} = \frac{x_{k-1}}{1+ \left(n/2- 1\right)(b_{k-1} + x_{k-1} - a_{k-1} - y_{k-1})}.
\label{eq-xk-odd-positive}
\end{align}
Thus, we have
\begin{align*}
&\quad\frac{x_{k-1}}{x_{k}} \\
(\text{by \eqref{eq-xk-odd-positive}})\quad &= 1+ \left(\frac{n}{2}- 1\right)(b_{k-1} + x_{k-1} - a_{k-1} - y_{k-1})\\
(\text{by \eqref{eq-max-xyuv-less-rho-over-fourn}})\quad &> 1+ \left(\frac{n}{2}- 1\right)(b_{k-1} - a_{k-1}) -\rho\\
(\text{by \eqref{eq-one-plusbkplusoneakplusone-lowerbound}})\quad &> \frac{1}{(1+b_{k-2})^3}.
\end{align*}

\vspace{0.5cm}

Similarly, we also have
\begin{equation}\label{eq-oneminustwobkplusoneminusrho-lowerbound}
\begin{aligned}
&\quad 1-2b_{k-1} - \rho\\
(\text{by \eqref{eq-bk-recursion-positive}})\quad &\geq 1-\frac{2b_{k-2}}{1+2b_{k-2} - \rho} - \rho\\
\left(\text{by \eqref{eq-bkminustwo-klessnovervarepsilon-rho-3}}\right)\quad & \geq 1 - 2b_{k-2} - \rho\\
\left(\text{by \eqref{eq-bkminustwo-klessnovervarepsilon-rho-3} and \eqref{eq-ub-maxab-positive}}\right)\quad & \geq \frac{1}{1+  3b_{k-2}}\\
\quad &> \frac{1}{(1+b_{k-2})^3}.
\end{aligned}
\end{equation}
In addition, by \Cref{lem-relation-rowsum-columnsum} we have
\begin{align}
\quad y_{k} &= \frac{y_{k-1}}{r_{n}\left(A^{(k-1)}\right)} = \frac{y_{k-1}}{1-\left(2(b_{k-1} - y_{k-1}) + (n/2-1)(u_{k-1}-v_{k-1})\right)}.
\label{eq-yk-odd-positive}
\end{align}
Thus, we have
\begin{align*}
&\quad\frac{y_{k-1}}{y_{k}} \\
(\text{by \eqref{eq-yk-odd-positive}})\quad &= 1-\left(2(b_{k-1} - y_{k-1}) + \left(\frac{n}{2}-1\right)(u_{k-1}-v_{k-1})\right)\\
(\text{by \eqref{eq-max-xyuv-less-rho-over-fourn}})\quad &> 1-2b_{k-1} -\rho\\
(\text{by \eqref{eq-oneminustwobkplusoneminusrho-lowerbound}})\quad &> \frac{1}{(1+b_{k-2})^3}.
\end{align*}

\vspace{0.5cm}

Moreover, one can also verify that 
\begin{align}
1+2a_{k-1} - \rho&>1,\label{eq-one-plus-twoakplusone-geone}
\end{align}
By \Cref{lem-relation-rowsum-columnsum} we have
\begin{align}
u_{k} &= \frac{u_{k-1}}{r_{1}\left(A^{(k-1)}\right)} = \frac{u_{k-1}}{1+\left(2(a_{k-1} - x_{k-1}) + (n/2-1)(u_{k-1}-v_{k-1})\right)}.
\label{eq-uk-odd-positive}
\end{align}
Thus, we have
\begin{align*}
&\quad\frac{u_{k-1}}{u_{k}} \\
(\text{by \eqref{eq-uk-odd-positive}})\quad &= 1+\left(2(a_{k-1} - x_{k-1}) + \left(\frac{n}{2}-1\right)(u_{k-1}-v_{k-1})\right)\\
(\text{by \eqref{eq-max-xyuv-less-rho-over-fourn}})\quad &> 1+2a_{k-1} -\rho\\
(\text{by \eqref{eq-one-plus-twoakplusone-geone}})\quad &> 1.
\end{align*}

\vspace{0.5cm}

At last, by \Cref{lem-relation-rowsum-columnsum} we have \begin{align}
v_{k} &= \frac{v_{k-1}}{r_{n}\left(A^{(k-1)}\right)} = \frac{v_{k-1}}{1-\left(2(b_{k-1} - y_{k-1}) + (n/2-1)(u_{k-1}-v_{k-1})\right)}.
\label{eq-vk-odd-positive}
\end{align}
Therefore, 
\begin{align*}
&\quad\frac{v_{k-1}}{v_{k}} \\
(\text{by \eqref{eq-vk-odd-positive}})\quad &= 1-\left(2(b_{k-1} - y_{k-1}) + \left(\frac{n}{2}-1\right)(u_{k-1}-v_{k-1})\right)\\
(\text{by \eqref{eq-max-xyuv-less-rho-over-fourn}})\quad &> 1-2b_{k-1} -\rho\\
(\text{by \eqref{eq-oneminustwobkplusoneminusrho-lowerbound}})\quad &> \frac{1}{(1+b_{k-2})^3}.
\end{align*}
In summary, we have 
\begin{align}\label{eq-ratio-xk-over-xkminusone}
\max\left\{\frac{x_{k}}{x_{k-1}},\frac{y_{k}}{y_{k-1}},\frac{u_{k}}{u_{k-1}},\frac{v_{k}}{v_{k-1}}\right\}\leq (1+b_{k-2})^3.
\end{align}

Similarly, for each even $2\leq k\leq t$, by \eqref{eq-key-relation-positive-even-1} and \eqref{eq-def-rho-positive},
one can also verify that $1-2a_{k-2} - \rho >(1+b_{k-2})^{-3}$
and $1+2b_{k-2} - \rho  > 1$.
Moreover, we have
\begin{equation}
\begin{aligned}
&\quad 1 + \left(\frac{n}{2} - 1\right)(a_{k-2}-b_{k-2})-\rho\\
(\text{by \eqref{eq-key-relation-positive-even-2}} )\quad &>1 - \left(\frac{n}{2} - 1\right)(2a_{k-2}b_{k-2} + 6(k-1)\rho a_{k-2})-\rho\\
(\text{by \eqref{eq-ub-maxab-positive}})\quad &>1 - \left(2b_{k-2} + 6\left(\frac{n}{2} - 1\right)(k-1)\rho a_{k-2}\right)-\rho\\
(\text{by \eqref{eq-def-rho-positive} and $b_{k-2} \geq \delta$})\quad &> (1+b_{k-2})^{-3}.
\end{aligned}
\end{equation}
Thus, by following the proof of \eqref{eq-ratio-xk-over-xkminusone},
we also have  
\[\max\left\{\frac{x_{k-1}}{x_{k-2}},\frac{y_{k-1}}{y_{k-2}},\frac{u_{k-1}}{u_{k-2}},\frac{v_{k-1}}{v_{k-2}}\right\}\leq (1+b_{k-2})^3.\]
Combined with \eqref{eq-ratio-xk-over-xkminusone}, \eqref{eq-relation-maxratioofxyuv-ratiob} is immediate.

\subsection*{Proof of \Cref{lem-last-item}}
We prove this lemma by induction. 
For the base step where $k = 0$, we have
\begin{align*}
\max\{x_{0},y_{0},u_{0},v_{0}\} \leq \frac{\rho\delta^6}{n^{51}}.
\end{align*}
immediately by \Cref{lem-base-case}.

For the inductive step where $0<k \leq t$,
it is sufficient to prove the following results about $x_k$:
\begin{itemize}
\item if $ k\leq \min\{t,\ell\}$, then
\begin{align}\label{eq-ub-xyuv-rho-kleql-xk}
x_k \leq \frac{\rho\delta^6}{n^{51}}\cdot \left(1+\frac{2}{n-1}\right)^{3k};
\end{align}
\item if $ \ell \leq k< t$ is odd, then
\begin{align}\label{eq-ub-xyuv-rho-kgeql-odd-xk}
x_k \leq \frac{\rho\delta^6}{n^{2}}\cdot \left(\frac{b_{\ell}}{b_{k-1}}\right)^6\cdot \left(\frac{b_{k-1}}{b_{k+1}}\right)^3;
\end{align}
\item if $\ell \leq k\leq t$ is even, then
\begin{align}\label{eq-ub-xyuv-rho-kgeql-even-xk}
x_k \leq \frac{\rho\delta^6}{n^{2}}\cdot \left(\frac{b_{\ell}}{b_{k}}\right)^6.
\end{align}
\end{itemize}
Similarly, one can also prove the same upper bounds for $y_{k},u_{k},v_{k}$.
Hence, the lemma is immediate.

We first assume $0< k \leq \min\{\ell,t\}$ and prove \eqref{eq-ub-xyuv-rho-kleql-xk}.
If $k$ is even, we have
\begin{align*}
&\quad x_{k} \\
(\text{by \eqref{eq-relation-maxratioofxyuv-ratiob}})\quad &= x_{k-1}\cdot (1+b_{k-2})^3 \\
(\text{by \eqref{eq-ub-maxab-positive}})\quad & \leq x_{k-1}\cdot \left(1+\frac{2}{n-1}\right)^3 \\
(\text{by the inductive assumption})\quad & \leq \frac{\rho\delta^6}{n^{51}}\cdot \left(1+\frac{2}{n-1}\right)^{3(k-1)}\cdot \left(1+\frac{2}{n-1}\right)^3\\
&\leq \frac{\rho\delta^6}{n^{51}}\cdot \left(1+\frac{2}{n-1}\right)^{3k}.
\end{align*}
If $k$ is odd, one can also verify \eqref{eq-ub-xyuv-rho-kleql-xk}.
Thus, \eqref{eq-ub-xyuv-rho-kleql-xk} always holds.

In the following, we assume $\ell \leq k \leq t$. 
If $k = \ell$, by \eqref{eq-def-ell-positive} we have $\ell$ is even. 
In addition,
\begin{align*}
&\quad x_{\ell} \\
(\text{by \eqref{eq-ub-xyuv-rho-kleql-xk}})\quad &= \frac{\rho\delta^6}{n^{51}}\cdot \left(1+\frac{2}{n-1}\right)^{3{\ell}} \\
(\text{by \eqref{eq-def-ell-positive}})\quad &=\frac{\rho\delta^6}{n^{51}} \cdot\left(1+\frac{2}{n-1}\right)^{24n(\lceil\log n\rceil -1)}\\
&\leq \frac{\rho\delta^6}{n^{51}} \cdot\exp({49(\lceil\log n\rceil -1)})\\
&\leq \frac{\rho\delta^6}{n^{2}}.
\end{align*}
Thus, we have \eqref{eq-ub-xyuv-rho-kgeql-even-xk} immediately.
If $k > \ell$, we prove \eqref{eq-ub-xyuv-rho-kgeql-odd-xk} and \eqref{eq-ub-xyuv-rho-kgeql-even-xk} by considering the following two separate cases:
\begin{itemize}
\item $\ell < k\leq t$ is even.
Recall that $\ell$ is even.
Combined with $k > \ell$,
we have $k-2 \geq \ell$. 
Hence, 
\begin{align*}
&\quad x_{k} \\
(\text{by \eqref{eq-relation-maxratioofxyuv-ratiob}})\quad &= x_{k-1}\cdot (1+b_{k-2})^3 \\
(\text{by $k\geq \ell + 2$ and \eqref{eq-bk-leq-bkminustwo-oneplusbkminustwo}})\quad & \leq x_{k-1}\cdot \left(\frac{b_{k-2}}{b_{k}}\right)^3 \\
(\text{by the inductive assumption})\quad & \leq \frac{\rho\delta^6}{n^{2}}\cdot \left(\frac{b_{\ell}}{b_{k-2}}\right)^6\cdot \left(\frac{b_{k-2}}{b_{k}}\right)^3 \cdot \left(\frac{b_{k-2}}{b_{k}}\right)^3 \\
&\leq \frac{\rho\delta^6}{n^{2}}\cdot \left(\frac{b_{\ell}}{b_{k}}\right)^6.
\end{align*}
\item $\ell < k< t$ is odd. 
We have $\ell + 2 \leq k+1\leq t$.
Hence, 
\begin{align*}
&\quad x_{k} \\
(\text{by $k+1\leq t$, $k+1$ is even and \eqref{eq-relation-maxratioofxyuv-ratiob}})\quad &= x_{k-1}\cdot (1+b_{k-1})^3 \\
(\text{by $k+1\geq \ell + 2$ and \eqref{eq-bk-leq-bkminustwo-oneplusbkminustwo}})\quad & \leq x_{k-1}\cdot \left(\frac{b_{k-1}}{b_{k+1}}\right)^3 \\
(\text{by the inductive assumption})\quad & \leq \frac{\rho\delta^6}{n^{2}}\cdot \left(\frac{b_{\ell}}{b_{k-1}}\right)^6\cdot \left(\frac{b_{k-1}}{b_{k+1}}\right)^3.
\end{align*}
\end{itemize}
Thus, \eqref{eq-ub-xyuv-rho-kgeql-even-xk} and \eqref{eq-ub-xyuv-rho-kgeql-odd-xk} are proved and the lemma is immediate.

\section{Proof of Theorem 1.5}\label{appendix-application}
In this section, we prove \Cref{thm-application}, 
which is an immediate corollary of \Cref{thm-upper-bound-general} and the following theorem adapted from \cite{huber2008fast}.

\begin{theorem}\label{thm-dense-huber}
Let $\gamma \in (1/2,1]$, $ \delta \in (0,1]$ and $\varepsilon\in (0,1]$. 
Let $A$ be any $\gamma$-dense $0$-$1$ matrix of size $n\times n$.
Assume there exists an algorithm that, given \(A\) as input, runs in expected time \(O(T)\) and outputs matrices \(B\), \(X\), and \(Y\) such that \(B = XAY\) has a maximum deviation of \(1/(10n^2)\), where \(X\) and \(Y\) are nonnegative diagonal matrices.
Then there exists a randomized approximation algorithm which outputs an approximation of $\perman(A)$ within a factor of $1+\varepsilon$ with probability at least $1 - \delta$ and expected running time $O(T+ n^{2+(1-\gamma)/(2\gamma - 1)}\varepsilon^{-2}\log (1/\delta))$.
\end{theorem}

Note that $B$ immediately attains a maximum deviation $1/(10n^2)$ if 
\[\norm{r\left(B\right)- \boldsymbol{1}}_{1} + \norm{c\left(B\right)- \boldsymbol{1}}_{1} \leq \frac{1}{10n^2}.\]
Thus, by \Cref{thm-upper-bound-general}, we have the Sinkhorn-Knopp algorithm can output the matrices $B,X,Y$ in time $O(n^2(2\gamma - 1)^{-5}\log n)$ such that $B = XAY$ has a maximum deviation $1/(10n^2)$ and 
$X$ and $Y$ are nonnegative diagonal matrices.
Combined with \Cref{thm-dense-huber}, we have \Cref{thm-application} immediately.

\end{document}